\documentclass[onecolumn,a4paper]{quantumarticle}
\pdfoutput=1

\usepackage{graphicx,epic,eepic,epsfig,amsmath,latexsym,amssymb,verbatim,mathtools,amsthm}
\usepackage[table]{xcolor}
\colorlet{SCcolor}{gray!25}
\usepackage{totcount}

\usepackage{tikz,animate}
\usepackage{tikzscale}
\usepackage{pgfplots}
\usetikzlibrary{calc}
\usetikzlibrary{patterns, snakes}
\usetikzlibrary{decorations.pathreplacing}
\usepackage{tikzpeople}
\usepackage{blox}

\newtheorem{definition}{Definition}[section]

\newtheorem{lemma}[definition]{Lemma}

\newtheorem{theorem}[definition]{Theorem}
\newtheorem{corollary}[definition]{Corollary}

\newtheorem{remark}[definition]{Remark}
\newtotcounter{numQuestions} 

\usepackage{pifont}
\newcommand{\cmark}{\ding{51}}%
\newcommand{\etal}{et. al. }%

\def\squareforqed{\hbox{\rlap{$\sqcap$}$\sqcup$}}
\def\qed{\ifmmode\squareforqed\else{\unskip\nobreak\hfil
\penalty50\hskip1em\null\nobreak\hfil\squareforqed
\parfillskip=0pt\finalhyphendemerits=0\endgraf}\fi}
\def\endenv{\ifmmode\;\else{\unskip\nobreak\hfil
\penalty50\hskip1em\null\nobreak\hfil\;
\parfillskip=0pt\finalhyphendemerits=0\endgraf}\fi}

\mathchardef\ordinarycolon\mathcode`\:
\mathcode`\:=\string"8000
\def\vcentcolon{\mathrel{\mathop\ordinarycolon}}
\begingroup \catcode`\:=\active
  \lowercase{\endgroup
  \let :\vcentcolon
  }

\usepackage[cal=boondoxo]{mathalfa}
\usepackage{todonotes}

\newcommand{\nc}{\newcommand}

\nc{\rnc}{\renewcommand}
\nc{\beq}{\begin{equation}}
\nc{\eeq}{{\end{equation}}}
\nc{\beqa}{\begin{eqnarray}}
\nc{\eeqa}{\end{eqnarray}}
\nc{\lbar}[1]{\overline{#1}}
\nc{\bra}[1]{\langle#1|}
\nc{\ket}[1]{|#1\rangle}
\nc{\ketbra}[2]{|#1\rangle\!\langle#2|}
\nc{\braket}[2]{\langle#1|#2\rangle}
\nc{\proj}[1]{| #1\rangle\!\langle #1 |}
\nc{\avg}[1]{\langle#1\rangle}
\nc{\Rank}{\operatorname{Rank}}
\nc{\smfrac}[2]{\mbox{$\frac{#1}{#2}$}}
\nc{\tr}{\operatorname{Tr}}
\nc{\ox}{\otimes}
\nc{\dg}{\dagger}
\nc{\dn}{\downarrow}
\DeclareMathOperator{\supp}{supp}
\nc{\cA}{{\cal A}}
\nc{\cB}{{\cal B}}
\nc{\cC}{{\cal C}}
\nc{\cD}{{\cal D}}
\nc{\bD}{{\mathbf{D}}}
\nc{\cE}{{\cal E}}
\nc{\ce}{{\cal e}}
\nc{\cf}{{\cal f}}
\nc{\cd}{{\cal d}}
\nc{\cn}{{\cal n}}
\nc{\cp}{{\cal p}}
\nc{\cq}{{\cal q}}
\nc{\cm}{{\cal m}}
\nc{\cF}{{\cal F}}
\nc{\cG}{{\cal G}}
\nc{\cH}{{\cal H}}
\nc{\cI}{{\cal I}}
\nc{\cJ}{{\cal J}}
\nc{\cK}{{\cal K}}
\nc{\cL}{{\cal L}}
\nc{\cM}{{\cal M}}
\nc{\cN}{{\cal N}}
\nc{\cO}{{\cal O}}
\nc{\cP}{{\cal P}}
\nc{\cQ}{{\cal Q}}
\nc{\cR}{{\cal R}}
\nc{\cS}{{\cal S}}
\nc{\cT}{{\cal T}}
\nc{\cU}{{\cal U}}
\nc{\cV}{{\cal V}}
\nc{\cW}{{\cal W}}
\nc{\cX}{{\cal X}}
\nc{\cY}{{\cal Y}}
\nc{\cZ}{{\cal Z}}
\usepackage{dsfont}
\newcommand{\Id}{{\mathds{1}}}
\nc{\csupp}{{\operatorname{csupp}}}
\nc{\qsupp}{{\operatorname{qsupp}}}
\nc{\var}{{\operatorname{var}}}
\nc{\Var}{{\operatorname{Var}}}
\nc{\rar}{\rightarrow}
\nc{\lrar}{\longrightarrow}
\nc{\polylog}{{\operatorname{polylog}}}
\nc{\1}{{\openone}}
\nc{\wt}{{\operatorname{wt}}}
\nc{\av}[1]{{\left\langle {#1} \right\rangle}}
\newcommand{\eps}{\varepsilon}

\def\r{\rho}

\nc{\RR}{{{\mathbb R}}}
\nc{\CC}{{{\mathbb C}}}
\nc{\DD}{{{\mathbb D}}}
\nc{\FF}{{{\mathbb F}}}
\nc{\NN}{{{\mathbb N}}}
\nc{\ZZ}{{{\mathbb Z}}}
\nc{\PP}{{{\mathbb P}}}
\nc{\QQ}{{{\mathbb Q}}}
\nc{\UU}{{{\mathbb U}}}
\nc{\EE}{{{\mathbb E}}}
\nc{\id}{{\operatorname{id}}}

\nc{\CHSH}{{\operatorname{CHSH}}}

\nc{\be}{\begin{equation}}
\nc{\ee}{{\end{equation}}}
\nc{\bea}{\begin{eqnarray}}
\nc{\eea}{\end{eqnarray}}
\nc{\<}{\langle}
\rnc{\>}{\rangle}
\nc{\Hom}[2]{\mbox{Hom}(\CC^{#1},\CC^{#2})}
\nc{\rU}{\mbox{U}}

\nc{\ob}[1]{#1}

\nc{\SEP}{{\text{SEP}}}
\nc{\NS}{{\text{NS}}}
\nc{\LOCC}{{\text{LOCC}}}
\nc{\PPT}{{\text{PPT}}}
\nc{\EXT}{{\text{EXT}}}
\nc{\Sym}{{\operatorname{Sym}}}

\nc{\ERLO}{{E_{\text{r,LO}}}}
\nc{\ERLOCC}{{E_{\text{r,LOCC}}}}
\nc{\ERPPT}{{E_{\text{r,PPT}}}}
\nc{\ERLOCCinfty}{{E^{\infty}_{\text{r,LOCC}}}}
\nc{\Aram}{{\operatorname{\sf A}}}

\begin{document}

\title{Quantum Network Discrimination}

\date{}

\author{Christoph Hirche}
\email{christoph.hirche@gmail.com}
\affiliation{QMATH, Department of Mathematical Sciences, University of Copenhagen, Universitetsparken 5, 2100 Copenhagen, Denmark}

\begin{abstract}
Discrimination between objects, in particular quantum states, is one of the most fundamental tasks in (quantum) information theory. Recent years have seen significant progress towards extending the framework to point-to-point quantum channels. However, with technological progress the focus of the field is shifting to more complex structures: Quantum networks. In contrast to channels, networks allow for intermediate access points where information can be received, processed and reintroduced into the network. In this work we study the discrimination of quantum networks and its fundamental limitations. 

In particular when multiple uses of the network are at hand, the roster of available strategies becomes increasingly complex. The simplest quantum network that captures the structure of the problem is given by a quantum superchannel. We discuss the available classes of strategies when considering $n$ copies of a superchannel and give fundamental bounds on the asymptotically achievable rates in an asymmetric discrimination setting. Furthermore, we discuss achievability, symmetric network discrimination, the strong converse exponent, generalization to arbitrary quantum networks and finally an application to an active version of the quantum illumination problem. 
\end{abstract}

\maketitle

\section{Introduction}
\label{sec:intro}

Hypothesis testing not only allows us to investigate the usually unavoidable error occurring when discriminating between two possible quantum states or channels. The framework has also proven useful in giving bounds, determining properties and proof operational interpretations of quantities such as the capacity of a quantum channel~\cite{wang2012one,datta2013smooth}, the entanglement in a quantum state~\cite{brandao2020adversarial} and many more~\cite{cooney2016operational, hirche2017discrimination}. 
In the case of discriminating two quantum states, a wide body of literature is available determining the optimal single-copy and asymptotic errors in several different scenarios, see e.g.~\cite{Audenaert2008,BBH}. In particular, it is known that when several copies of a quantum state are available, measuring each copy individually does not lead to the optimal error but rather one needs to use a joint measurement. 

The case of discriminating between two quantum channels is much more complex because one can also choose the input of the channels in order to facilitate the discrimination. In particular, when multiple copies or uses are available, this additional freedom allows to pick the inputs adaptively based on earlier outputs of the channel. Consequently, determining the asymptotically optimal error in this setting has long remained an open problem. Recently, a series of publications has finally lead to significant progress on the problem, in particular in the asymmetric asymptotic Stein's setting~\cite{BHKW,wang2019resource,fang2019chain}. While it was known that in the classical setting adaptive strategies and even jointly distributed channel inputs do not lead to any advantage~\cite{Hayashi09}, general converse bounds in the quantum setting were only recently shown in~\cite{BHKW}, which in particular allowed to extend the classical result to classical-quantum channels, see also~\cite{BHKW2}. Subsequent work then showed that in the Stein's setting the bounds from~\cite{BHKW} are indeed optimal also for general channels~\cite{wang2019resource} and that adaptive strategies do not outperform parallel strategies with entangled inputs~\cite{fang2019chain}. However, generally, entangled inputs are necessary and give an advantage over product state inputs~\cite{fang2019chain}. A close investigation of the power of adaptive strategies in several settings followed in~\cite{salek2020adaptive}. 

Both, quantum states and channels, already play an important role when implementing quantum technologies. However, with  advances in experimental research and implementations it becomes increasingly more relevant to investigate the properties of more complex structures such as quantum networks. These networks allow for even more ways of interacting. In particular, one could at some point receive an intermediate output of the network, process it and then reintroduce it to the remainder of the network. Relevant examples of such networks on different size scales include quantum circuits, distributed computational resources or even a quantum internet. These additional possibilities make the problem even more complex than the channel case. One can start with product or entangled states, use individual or joint measurements and process intermediate access points with individual or joint quantum channels. Additionally all the different tools can be chosen adaptively and while channels can be used either in parallel or successively, one can run the first part of a quantum network, get a state from an access point and run it through an entirely different copy of the network before reintroducing it to the next part of the first network. All these possibilities lead to a wide range of available strategies one has to consider when searching for the optimal error rates. 

In this work we classify the possible strategies and give bounds on the optimal error when discriminating between networks. We focus on the asymptotic setting and here particularly the asymmetric Stein's setting. For all classes we give converse bounds and discuss their optimality. As evident from the above description the problem is extremely complex, which is why for much of this work we discuss the results for the special case of networks with exactly one intermediate access point. This subclass is also known as superchannels since it can be understood as transforming quantum channels, as input to the access point, into quantum channels. Afterwards we extend the results to general networks and settings beyond Stein's Lemma. 

The remainder of this paper is organized as follows. First we discuss the general notation and definitions in Section~\ref{sec:preliminaries}, including definitions of quantum networks and superchannels in Section~\ref{sec:int-networks}. As a warm-up we then briefly discuss one-shot discrimination of superchannels in Section~\ref{sec:one-shot}. In Section~\ref{sec:divergences} we discuss the primary technical tool of this investigation: amortized divergences for superchannels. The main part of this work is Section~\ref{sec:asymDisc}, where we discuss different classes of strategies for discriminating multiple copies of superchannels with focus on converse bounds in the asymptotic Stein's setting. In Section~\ref{sec:examples} we discuss several examples, such as classical channels and some channels with particular parameterizations. Next, we discuss other asymptotic settings such as the symmetric Chernoff setting and the strong converse exponent in Section~\ref{sec:BeyondStein}. In Section~\ref{sec:networks} we discuss the generalization to quantum networks with an arbitrary number of access points. Finally we discuss applications including an active variant of quantum illumination in Section~\ref{sec:applications} and conclude in Section~\ref{sec:conclusions}. 

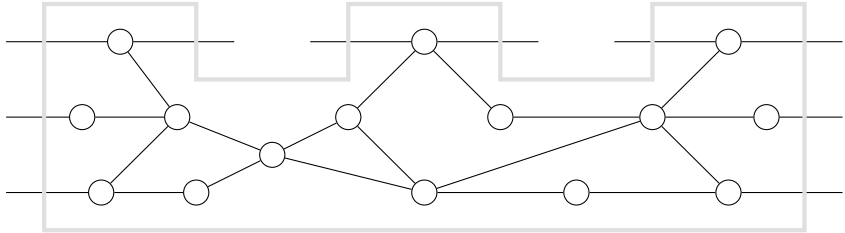
\begin{figure}[t]
\centering
\begin{tikzpicture}[]
\draw (2,0) node[draw=black,circle] (N1) {};
\draw (7,-1) node[draw=black,circle] (N2) {};
\draw (1.5,-1) node[draw=black,circle] (N3) {};
\draw (1.75,-2) node[draw=black,circle] (N4) {};
\draw (2.75,-1) node[draw=black,circle] (N5) {};
\draw (3,-2) node[draw=black,circle] (N6) {};
\draw (4,-1.5) node[draw=black,circle] (N7) {};
\draw (5,-1) node[draw=black,circle] (N8) {};
\draw (6,-2) node[draw=black,circle] (N9) {};
\draw (6,0) node[draw=black,circle] (N10) {};
\draw (9,-1) node[draw=black,circle] (N11) {};
\draw (8,-2) node[draw=black,circle] (N12) {};
\draw (10,0) node[draw=black,circle] (N13) {};
\draw (10.5,-1) node[draw=black,circle] (N14) {};
\draw (10,-2) node[draw=black,circle] (N15) {};

\draw (N11) -- (N15);
\draw (N1) -- (N5);
\draw (N3) -- (N5);
\draw (N4) -- (N6);
\draw (N4) -- (N5);
\draw (N5) -- (N7);
\draw (N6) -- (N7);
\draw (N7) -- (N8);
\draw (N7) -- (N9);
\draw (N8) -- (N10);
\draw (N8) -- (N9);
\draw (N9) -- (N12);
\draw (N9) -- (N11);
\draw (N12) -- (N15);
\draw (N11) -- (N13);
\draw (N11) -- (N14);
\draw (N11) -- (N2);
\draw (N10) -- (N2);

\draw (0.5,0) -- (N1);
\draw (0.5,-1) -- (N3);
\draw (0.5,-2) -- (N4);
\draw (N1) -- (3.5,0); 
\draw (4.5,0) -- (N10); 
\draw (N10) -- (7.5,0); 
\draw (8.5,0) -- (N13); 
\draw (N13) -- (11.5,0); 
\draw (N14) -- (11.5,-1); 
\draw (N15) -- (11.5,-2); 

\draw[ultra thick,color=SCcolor] (1,0.5) -- (3,0.5) -- (3,-0.5) -- (5,-0.5) -- (5,0.5) -- (7,0.5) -- (7,-0.5) -- (9,-0.5) -- (9,0.5) -- (11,0.5) -- (11,-2.5) -- (1,-2.5) -- cycle;

\end{tikzpicture}
\caption{\label{Fig:3comb} Any network can be considered as a $k$-comb, where $k-1$ is the number of intermediate access points. In this example, the chosen network takes the form of a $3$-comb with $2$ intermediate access points. }
\end{figure}

\section{Preliminaries}\label{sec:preliminaries}

Here we introduce our notation and give the relevant definitions needed later. In particular, Section~\ref{sec:entropies} introduces the required entropic quantities and Section~\ref{sec:int-networks} discuses definitions of quantum superchannels and networks. 


\subsection{Setup}

Throughout, quantum systems are denoted by capital letters $A$, $B$, etc. and have finite dimensions $|A|$, $|B|$, etc., respectively. Linear operators acting on system $A$ are denoted by $L_A\in\cL(A)$ and positive semi-definite operators by $P_A\in\cP(A)$. Quantum states of system $A$ are denoted by $\rho_A\in\cS(A)$ and pure quantum states by $\Psi_A\in\cV(A)$. 
Quantum channels are completely positive and trace-preserving maps from $\cL(A)$ to $\cL(B)$ and denoted by $\cN_{A\to B}\in\cQ(A\to B)$. The Choi state of a quantum channel $\cN_{A\to B}$ is a standard concept in quantum information and is defined as $\cN_{A\to B}( \Phi_{RA})$, where $\Phi_{RA}$ is the maximally entangled state. Classical systems are denoted by $X$, $Y$, and $Z$ and have finite dimensions $|X|$, $|Y|$, and $|Z|$, respectively. We will drop the indices when denoting quantum states or channels whenever we deem them clear from context. For $p\geq1$ the Schatten norms are defined for $L_A\in\cL(A)$ as
\begin{align}
\|L_A\|_p\coloneqq\Big(\tr\big[|L_A|^p\big]\Big)^{1/p}.
\end{align}


\subsection{Quantum entropies}\label{sec:entropies}

Asymptotic quantum hypothesis testing is closely related to several entropic quantities and their properties~\cite{hirche2018asymptotic}. In this section we will define the entropies used in the remainder of this work. 

The quantum relative entropy for $\rho,\sigma\in\cS(A)$ is defined as \cite{Ume62}
\begin{align}
D(\rho\|\sigma)\coloneqq\begin{cases} \tr\big[\rho\left(\log\rho-\log\sigma\right)\big] \qquad & \supp(\rho)\subseteq\supp(\sigma)\\ +\infty & \text{otherwise,}\end{cases}
\end{align}
where in the above and throughout the paper we employ the convention that all logarithms are evaluated using base two.
The Petz-R\'enyi divergences are defined for $\rho,\sigma\in\cS(A)$ and $\alpha\in(0,1)\cup(1,\infty)$ as \cite{P85,P86}
\begin{align}
D_\alpha(\rho\|\sigma)\coloneqq\frac{1}{\alpha-1}\log\tr\left[\rho^\alpha\sigma^{1-\alpha}\right],
\end{align}
whenever either $\alpha\in(0,1)$ and $\rho$ is not orthogonal to $\sigma$ in Hilbert-Schmidt inner product or $ \alpha>1$ and $\supp(\rho)\subseteq\supp(\sigma)$. Otherwise, we set $D_\alpha(\rho\|\sigma)\coloneqq +\infty$.
In the above and throughout the paper, we employ the convention that inverses are to be understood as generalized inverses.
For $\alpha\in\{0,1\}$, we define the Petz-R\'enyi divergence in the respective limit as
\begin{align}
D_0(\rho\|\sigma)& \coloneqq\lim_{\alpha\to0}D_\alpha(\rho\|\sigma)=-\log
\tr\left[ \Pi_{\rho} \sigma\right],
\\
D_1(\rho\|\sigma)& \coloneqq\lim_{\alpha\to1}D_\alpha(\rho\|\sigma)=D(\rho\|\sigma),
\end{align}
where $\Pi_{\rho}$ denotes the projection onto the support of $\rho$.
Another quantity of interest related to the Petz-R\'enyi divergences is the Chernoff divergence~\cite{ACMBMAV07,Audenaert2008,NS09}
\begin{align}
C(\rho\|\sigma) & \coloneqq-\inf_{0\leq\alpha\leq1}\log\tr\left[\rho^\alpha\sigma^{1-\alpha}\right]\\
& = \sup_{0\leq\alpha\leq1} (1-\alpha) D_\alpha(\rho \| \sigma). \label{eq:Chernoff-div-expr-renyi}
\end{align}
The sandwiched R\'enyi divergences are defined for $\rho,\sigma\in\cS(A)$ and $\alpha\in(0,1)\cup(1,\infty)$  as~\cite{muller2013quantum,WWY14}
\begin{align}
\widetilde{D}_\alpha(\rho\|\sigma)\coloneqq\frac{1}{\alpha-1}\log\tr\left[\left(\sigma^{\frac{1-\alpha}{2\alpha}}\rho\sigma^{\frac{1-\alpha}{2\alpha}}\right)^\alpha\right]
\end{align}
whenever either $ \alpha\in(0,1)$ and $\rho$ is not orthogonal to $\sigma$ in Hilbert-Schmidt inner product or $\alpha>1$ and $\supp(\rho)\subseteq\supp(\sigma)$. Otherwise we set $\widetilde{D}_\alpha(\rho\|\sigma)\coloneqq\infty$. For $\alpha=1$, we define the sandwiched R\'enyi relative entropy in the limit as
\begin{align}
\widetilde{D}_1(\rho\|\sigma)\coloneqq\lim_{\alpha\to1}\widetilde{D}_\alpha(\rho\|\sigma)=D(\rho\|\sigma).
\end{align}
We have that
\begin{align}
\widetilde{D}_{1/2}(\rho\|\sigma)=-\log F(\rho,\sigma),
\end{align}
with the Uhlmann's fidelity defined as $F(\rho,\sigma)\coloneqq\|\sqrt{\rho}\sqrt{\sigma}\|_1^2$~\cite{U76}.
In the limit $\alpha\to\infty$, the sandwiched R\'enyi relative entropy converges to the max-relative entropy~\cite{Datta09,Jain02}
\begin{align}
D_{\max}(\rho\|\sigma) \coloneqq\widetilde{D}_\infty(\rho\|\sigma)& \coloneqq\lim_{\alpha\to\infty}\widetilde{D}_\alpha(\rho\|\sigma)
\\
& =\log \left\|\sigma^{-1/2}\rho\sigma^{-1/2}\right\|_\infty \\
& =\inf\left\{\lambda:\rho\leq2^{\lambda}\cdot\sigma\right\},
\end{align}
as shown in \cite{muller2013quantum}. 
All of the above quantum R\'enyi divergences reduce to the corresponding classical version by embedding probability distributions into diagonal, commuting quantum states. 


\subsection{Quantum networks and superchannels}\label{sec:int-networks}\label{SSec:Networks}

The main subject of this work are quantum networks, which can be seen as maps taking channels as input and outputing a quantum channel. Such networks have previously been considered in various contexts in quantum information theory~\cite{LM15,WFD17,CG18,chiribella2009theoretical,G18}. We aim to determine the optimal discrimination errors and therefore assume a single experimenter who has access to the entire network, hence leaving settings with several distributed experimenters for future work. 

In~\cite{chiribella2009theoretical} it was shown that every quantum network as considered in this work can be described by a sequence of $k$ quantum channels, 
\begin{align}
(\cN^1_{C\rightarrow A_1S_1}, \cN^2_{B_1S_1\rightarrow A_2S_2}, \dots,  \cN^{k-1}_{B_{k-2}S_{k-2}\rightarrow A_{k-1}S_{k-1}}, \cN^{k}_{B_{k-1}S_{k-1}\rightarrow D}).
\end{align}
Such a sequence takes the shape of a comb and is therefore called a $k$-comb which has $k-1$ access points, compare also Figure~\ref{Fig:3comb}. We often use the convenient notation $\Theta^k \equiv (\cN^i)_k$ to denote $k$-combs. Sometimes we call the individual channels $\cN^i$ the components of $\Theta^k$. 

Let $\mathcal{Q}(C\rightarrow D)$ be the set of quantum channels from $C$ to $D$, then a $k$-comb acts on $k-1$ channels $\cA^i\in\mathcal{Q}(A_i\rightarrow B_i)$ as 
\begin{align}
\Theta^k ((\cA^i)_{k-1}) = \cN^k\circ \bigcirc_{i=k-1}^{1} ( \cA^i \circ \cN^i  ) \in\mathcal{Q}(C\rightarrow D).  
\end{align}
Whenever the resulting channel acts on a state we denote that by $\Theta^k ((\cA^i)_{k-1})(\rho)$. Sometimes we will need parts of $k$-combs and we denote as $\Theta^{k,m}$ the $m$-comb given by the first $m$ components of $\Theta^k$. Generally quantum networks constructed as described here naturally preserve complete positivity and trace preservation, i.e. if we input quantum channels the output is again a quantum channel.

In this work we will mostly consider $2$-combs, also known as superchannels~\cite{CDP08}, as they capture the important features of our problem. For simplicity we will usually denote $2$-combs as $\Theta\equiv\Theta^2$ and we will often use the decomposition $(\cE_{C\rightarrow AS}, \cD_{BS\rightarrow D})$. Generally we will drop subscripts when the system are clear from the context. 

If one takes e.g. a quantum circuit and writes it in the above form it is often not a priori clear whether a certain gate should belong to the channel $\cE$ or $\cD$, meaning many different networks can describe the same circuit. Here, we assume that a description of he network is given. Nevertheless for many examples it will be useful to make this ambiguity explicit by introducing a side-channel with which we parametrize a superchannel as $(\cE_{C\rightarrow AS},\cS_{S\rightarrow S'}, \cD_{BS'\rightarrow D})$. For example depictions of superchannels see Figures~\ref{Fig:super} and~\ref{Fig:superExamples}. 
It is however worth pointing out that one can limit the possible descriptions without losing generality. First notice that we can always extend $\cE$ to an isometry $\cV_\cE$ and $\cD$ to also trace out the additional system. This does not change the implemented superchannel~\cite{G18}. An analog statement for general networks had previously been observed in~\cite{bisio11,bisio11b}. Furthermore, it was shown in~\cite[Theorem 2]{gour2020dynamical} that all parameterizations $(\cV_\cE, \cD)$ with minimal system size $|S|$ are unique up to the choice of a unitary.


\section{Single-copy superchannel discrimination}\label{sec:one-shot}

As a warm-up to the topic of this work, we will in this section discuss the single-copy problem of quantum network discrimination. For related discussions and more background see also~\cite{chiribella2008memory,chiribella2009theoretical,nakahira2020ultimate}. It is well known that in symmetric state discrimination the optimal one-shot error is related to the trace distance~\cite{H69,H73,Hel76}. Here we want to consider the case of two quantum superchannels. That is an experimenter has access to one use of a superchannel, not knowing which out of the two available options $\Theta_1$ or $\Theta_2$ it is. In order to decide which  is the case, the most general approach is to feed an arbitrary state $\rho_{CR}$ into the superchannel which itself get applied to a channel $\cN_{AR\rightarrow BR}$, resulting into an output state $\rho^1_{DR}$ if the superchannel was $\Theta_1$ and $\rho^2_{DR}$ if the superchannel was $\Theta_2$. To the output state one can apply a measurement to determine the superchannel. Without loss of generality we can assume the use of a binary measurement $\cQ=\{Q_1, Q_2=\1-Q_1\}$. As usual in the hypothesis testing setting, one is left with two possibilities of erroneously determining the result: the Type-1 and the Type-2 error, given by
\begin{align}
\alpha(\cS) &= \tr(Q_2\Theta_1(\cN)(\rho) ), \\
\beta(\cS) &= \tr(Q_1 \Theta_2(\cN)(\rho) )
\end{align}
respectively. Throughout this manuscript $\cS$ will denote the chosen strategy, in this case the set $\cS=\{\rho,\cN,\cQ\}$. 

The most common setting in the single copy case is to minimize the average error. For a fixed channel $\cN_{AR\rightarrow BR}$ we can follow directly from the channel case that the probability of error is given by 
\begin{align}
p_{err}(\Theta_1,\Theta_2,\cN) = \frac12 - \frac14 \| \Theta_1(\cN) - \Theta_2(\cN) \|_\diamond, 
\end{align}
where we used the usual diamond norm for quantum channels. Of course we are also allowed to optimize over the channel $\cN$ leading to the optimal one shot error probability
\begin{align}
p_{err}(\Theta_1,\Theta_2) = \inf_{\cN} p_{err}(\Theta_1,\Theta_2,\cN) = \frac12 - \frac14 \sup_{\cN}\| \Theta_1(\cN) - \Theta_2(\cN) \|_\diamond, 
\end{align}
which motivates the definition of a diamond norm equivalent on superchannels,
\begin{align}
\| \Theta_1 - \Theta_2 \|_\diamond := \sup_{\cN}\| \Theta_1(\cN) - \Theta_2(\cN) \|_\diamond.
\end{align}
Note that a priori the optimization is over channels with arbitrary large reference system $R$. For this reason we are also free to omit the additional reference system usually attached to the state, that does not pass through the channel $\cN$, as the latter already includes channels of the form $\cN\otimes\id$, a fact that will later reemerge when discussing superchannel entropies. 

As a simple example we can consider replacer superchannels $\Theta_\cR$ that act as $\Theta_\cR(\cN)=\cR$ for every $\cN$ and a particular fixed $\cR$. Those channels generalize the commonly considered replacer channels and their implementation is shown in Figure~\ref{Fig:super}. We easily get that for two replacer superchannels $\Theta_{\cR_1}$ and $\Theta_{\cR_2}$ we have
\begin{align}
\| \Theta_{\cR_1} - \Theta_{\cR_2} \|_\diamond = \| \cR_1 - \cR_2 \|_\diamond.
\end{align}
Replacer superchannels will be considered in more detail in Section~\ref{sec:replacer}. 

Now, the same can easily be done for two arbitrary quantum networks. The main difference is that when defining the diamond norm for $k$-combs one has to optimize over $(k-1)$-combs instead of the quantum channel $\cN$ (or alternatively one could optimize over $k-1$ separate quantum channels).

In contrast to the symmetric case, one is often interested in minimizing one error, say the Type-2 error, while only bounding the Type-1 error to be below a certain constant $\epsilon$. Similarly to the discussion above for the diamond norm, this motivates us to define a superchannel hypothesis testing relative entropy that extends the state and channel case~\cite{wang2012one, Cooney2016} as follows, 
\begin{align}
D_H^\epsilon(\Theta_1\| \Theta_2) &:= \sup_{\cN} D_H^\epsilon(\Theta_1(\cN) \| \Theta_2(\cN)) \nonumber\\
&= \sup_{\cS} -\log\left\{ \tr(Q_1 \Theta_2(\cN)(\rho) ) \, |\, \tr((\1 -Q_1) \Theta_1(\cN)(\rho) ) \leq \epsilon \right\}, 
\end{align}
where we again optimize over strategies $\cS=\{\rho,\cN,\cQ\}$. As before, a version for networks is defined similarly. 

As we will see in the remainder of this work, the picture becomes much more complicated if we allow for multiple uses of the quantum superchannel or network since they can be combined in a variety of different configurations. To proceed,  we will first introduce the relevant entropic distance measures on superchannels in the next section.


\section{Generalized divergences ...}\label{sec:divergences}
We want to define distance measures for superchannels based on generalized divergences. We will later see that they are operationally meaningful in terms of superchannel discrimination. In the build up, we will first discuss generalized divergences for states and channels. The generalization to arbitrary networks is discussed later in Section~\ref{sec:networks}.

\subsection{ ... for quantum states}\label{sec:divergences-st}

\begin{figure}[t]
\centering
\resizebox{\textwidth}{!}{
\begin{tikzpicture}
\node at (-2.1,1.3){a)};
\draw[] (0,0) rectangle (1,1);
\node at (0.5,0.5){$\cN$};
\draw[fill=SCcolor] (-1.5,1) -- (-0.5,1) -- (-0.5,-0.5) -- (1.5,-0.5) -- (1.5,1) -- (2.5,1) -- (2.5,-1.5) -- (-1.5,-1.5) -- cycle;
\draw (-0.5,0.5) -- (0,0.5) node[pos=0.5,sloped,above]{$A$};
\draw (1,0.5) -- (1.5,0.5) node[pos=0.5,sloped,above]{$B$};
\draw (-2.5,-0.25) -- (-1.5,-0.25) node[pos=0.5,sloped,above]{$C$};
\draw (2.5,-0.25) -- (3.5,-0.25) node[pos=0.5,sloped,above]{$D$};
\node at (0.5,-1){$\Theta$};

\node at (-2.1+7,1.3){b)};
\draw[] (0+7,0) rectangle (1+7,1);
\node at (0.5+7,0.5){$\cN$};
\draw[fill=SCcolor] (-1.5+7,1) rectangle (-0.5+7,-1.5) node[pos=0.5]{$\cE_\Theta$};
\draw[fill=SCcolor] (1.5+7,-1.5) rectangle (2.5+7,1) node[pos=0.5]{$\cD_\Theta$};
\draw (-0.5+7,-1) -- (1.5+7,-1) node[pos=0.5,sloped,above]{$S$};
\draw (-0.5+7,0.5) -- (0+7,0.5) node[pos=0.5,sloped,above]{$A$};
\draw (1+7,0.5) -- (1.5+7,0.5) node[pos=0.5,sloped,above]{$B$};
\draw (-2.5+7,-0.25) -- (-1.5+7,-0.25) node[pos=0.5,sloped,above]{$C$};
\draw (2.5+7,-0.25) -- (3.5+7,-0.25) node[pos=0.5,sloped,above]{$D$};

\node at (-2.1+14,1.3){c)};
\draw[] (0+14,0) rectangle (1+14,1);
\node at (0.5+14,0.5){$\cN$};
\draw[fill=SCcolor] (-1.5+14,1) -- (-0.5+14,1) -- (-0.5+14,-1.5) -- (-1.5+14,-1.5) -- cycle;
\draw[fill=SCcolor] (1.5+14,-1.5) -- (1.5+14,1) -- (2.5+14,1) -- (2.5+14,-1.5)  -- cycle;
\draw (-0.5+14,-1) -- (1.5+14,-1);
\draw (-0.5+14,0.5) -- (0+14,0.5);
\draw (1+14,0.5) -- (1.5+14,0.5);
\draw (-2.5+14,-0.25) -- (-1.5+14,-0.25);
\draw (2.5+14,-0.25) -- (3.5+14,-0.25);
\draw (-1.3+14,0.5) -- (-0.7+14,0.9) -- (-0.7+14,0.1) -- cycle;
\draw (0+14,0.5) -- (-0.7+14,0.5);
\node at (-0.9+14,0.5){$\tau$};
\draw (2.3+14,0.5) -- (1.7+14,0.9) -- (1.7+14,0.1) -- cycle;
\draw (1.5+14,0.5) -- (1.7+14,0.5);
\node at (1.9+14,0.5){$\tr$};
\draw[] (-1.5+14,-0.25) to [out=0,in=180] (-0.5+14,-1);
\draw (1.7+14,-1.3) rectangle (2.3+14,-0.7) node[pos=0.5]{$\cR$};
\draw[] (2.3+14,-1) to [out=0,in=180] (2.5+14,-0.25);
\draw (1.5+14,-1) -- (1.7+14,-1);
\end{tikzpicture}}
\caption{\label{Fig:super} Depiction of superchannels: a) A general superchannel (or $2$-comb) $\Theta$ acting on a channel $\cN$. b) Any superchannel $\Theta$ can be understood as two quantum channels $\cE_\Theta$, $\cD_\Theta$ connected by an auxiliary system $S$. c) Implementation of a replacer superchannel: For any input $\cN$, we get $\Theta(\cN)=\cR$. We fix $\tau$ to be a maximally mixed state.}
\end{figure}
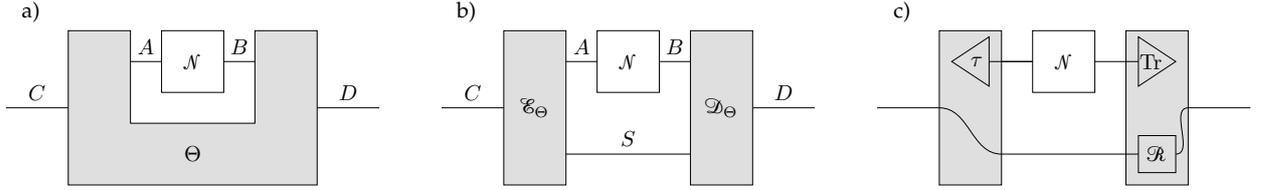

We say that a function $\mathbf{D}:\cS(A)\times\cS(A)\to\mathbb{R}\cup\{+\infty\}$ is a generalized divergence \cite{PV10,SW12} if for arbitrary Hilbert spaces $\mathcal{H}_A$ and $\mathcal{H}_B$, arbitrary states $\rho_A,\sigma_A\in\cS(A)$, and an arbitrary channel $\cN_{A\to B}\in\cQ(A\to B)$, the following data-processing inequality holds
\begin{align}\label{eq:data-processing}
\mathbf{D}(\rho_{A}\|\sigma_A)\geq\mathbf{D}(\cN_{A\to B}(\rho_A)\|\cN_{A\to B}(\sigma_A)).
\end{align}
From this inequality, we find in particular that for all states $\rho_A,\sigma_A\in\cS(A)$, $\omega_R\in\cS(R)$, the following identity holds \cite{WWY14}
\begin{align}
\mathbf{D}(\rho_A\otimes\omega_R\|\sigma_A\otimes\omega_R)=\mathbf{D}(\rho_A\|\sigma_A),
\end{align}
and that for an arbitrary isometric channel $\cU_{A\to B}\in\cQ(A\to B)$, we have that \cite{WWY14}
\begin{align}
\mathbf{D}(\cU_{A\to B}(\rho_A)\|\cU_{A\to B}(\sigma_A))=\mathbf{D}(\rho_A\|\sigma_A).
\end{align}

We call a generalized divergence \textit{faithful} if the inequality $\mathbf{D}(\rho_A\|\rho_A)\leq0$ holds for an arbitrary state $\rho_A\in\cS(A)$, and \textit{strongly faithful} if for arbitrary states $\rho_A,\sigma_A\in\cS(A)$ we have $\mathbf{D}(\rho_A\|\sigma_A)=0$ if and only if $\rho_A=\sigma_A$. Moreover, a generalized divergence is sub-additive with respect to tensor-product states if for all $\rho_A,\sigma_A\in\cS(A)$ and all $\omega_B,\tau_B\in\cS(B)$ we have
\begin{align}
\mathbf{D}(\rho_A\otimes\omega_B\|\sigma_A\otimes\tau_B)\leq\mathbf{D}(\rho_A\|\sigma_A)+\mathbf{D}(\omega_B\|\tau_B)
\end{align}
and simply additive if this holds with equality. 
Examples for generalized divergences of interest are in particular the quantum relative entropy, the Petz-R\'enyi divergences, the sandwiched R\'enyi divergences, or the Chernoff distance\,---\,as defined in Section~\ref{sec:entropies}.

\subsection{ ... for quantum channels}\label{sec:divergences-c}

Based on the generalized divergences for states, one can now define divergences for quantum channels. First, we have the generalized channel divergence for two quantum channels $\cN$, $\cM\in\cQ(A\to B)$,
\begin{align}
\bD(\cN \| \cM ) = \sup_\rho \bD((\cN\otimes\id)(\rho) \| (\cM\otimes\id)(\rho) ), 
\end{align} 
where the optimization is over states $\rho\in\cS(AR)$.  
Since the above quantity is not generally additive~\cite{fang2019chain} it is natural to define a regularized channel divergence as follows, 
\begin{align}
\bD^\infty(\cN \| \cM ) = \lim_{n\rightarrow\infty} \frac{1}{n} \bD(\cN^{\otimes n} \| \cM^{\otimes n} ).
\end{align} 
Finally, in~\cite{BHKW} the amortized channel divergence was introduced as
\begin{align}
\bD^A(\cN \| \cM ) = \sup_{\rho,\tau} \bD((\cN\otimes\id)(\rho) \| (\cM\otimes\id)(\tau) )  -  \bD(\rho \| \tau ).
\end{align} 
All of these definitions have an operational interpretation in terms of an associated channel discrimination task.
Notably, in~\cite{fang2019chain} is was proven that in the case of the relative entropy one gets
\begin{align}
D^\infty(\cN \| \cM ) = D^A(\cN \| \cM ) 
\end{align}
and an analog result was later proven for the sandwiched Renyi relative entropy~\cite{fawzi2020defining}. 

We will need some properties, especially of the amortized relative entropy. First we will summarize some results obtained in~\cite{BHKW} in the following lemma.
\begin{lemma}[From~\cite{BHKW}]\label{lem:AMOproperties}
Given a generalized state divergence $\bD$, we have that
\begin{itemize}
\item if $\bD$ is faithful, then its associated amortized channel divergence obeys $$\bD(\cN \| \cM )  \leq \bD^A(\cN \| \cM ).$$
\item if $\bD$ is strongly faithful, then its associated amortized channel divergence is also strongly faithful, i.e. $$ \bD^A(\cN \| \cM )=0 \,\text{ if and only if }\, \cN=\cM.$$
\item the associated amortized channel divergence is monotone under superchannels, i.e. $$ \bD^A(\cN \| \cM ) \geq \bD^A(\Theta(\cN) \| \Theta(\cM) ). $$
\item the associated amortized channel divergence is stable under adding an identity, i.e. $$ \bD^A(\cN\otimes\id \| \cM\otimes\id ) = \bD^A(\cN \| \cM ). $$ 
\end{itemize}
\end{lemma}

We now want to state some more properties of the above quantities that will become relevant later. 
\begin{lemma}\label{lem:circ}
For quantum channels $\cN$, $\cM\in\cQ(A\to B)$ and $\cS$, $\cT\in\cQ(B\to C)$, the following holds, 
\begin{align}
\bD^A(\cS\circ\cN \| \cT\circ\cM ) \leq \bD^A(\cS\| \cT ) + \bD^A(\cN \|\cM ). 
\end{align}
\end{lemma}
\begin{proof}
This follows from the following chain of arguments: 
\begin{align*}
&\bD^A(\cS\circ\cN \| \cT\circ\cM )  \\
&= \sup_{\rho,\tau} \left[ \bD((\cS\circ\cN )(\rho) \| (\cT\circ\cM)(\tau) )  -  \bD(\rho \| \tau ) \right] \\
&= \sup_{\rho,\tau} \left[ \bD((\cS\circ\cN )(\rho) \| (\cT\circ\cM)(\tau) ) -\bD(\cN(\rho) \| \cM(\tau) )+\bD(\cN(\rho) \| \cM(\tau) ) -  \bD(\rho \| \tau )\right] \\
&\leq \sup_{\rho,\tau} \left[ \bD((\cS\circ\cN )(\rho) \| (\cT\circ\cM)(\tau) ) -\bD(\cN(\rho) \| \cM(\tau) ) \right] +  \sup_{\rho,\tau} \left[ \bD(\cN(\rho) \| \cM(\tau) ) -  \bD(\rho \| \tau ) \right] \\
&= \sup_{\rho,\tau} \left[ \bD((\cS\circ\cN )(\rho) \| (\cT\circ\cM)(\tau) ) -\bD(\cN(\rho) \| \cM(\tau) ) \right] +  \bD^A(\cN \| \cM ) \\
&\leq \sup_{\hat\rho,\hat\tau} \left[ \bD(\cS(\hat\rho) \| \cT(\hat\tau) ) -\bD(\hat\rho \| \hat\tau ) \right] +  \bD^A(\cN \| \cM ) \\
&= \bD^A(\cS \| \cT )+  \bD^A(\cN \| \cM ) ,
\end{align*} 
where the first, third and fourth equality are by definition, the second equality by adding a zero, the first inequality by splitting the supremum and the second inequality by widening set of states which we optimize over. 
\end{proof}
\begin{lemma}
With the definitions as above, for any additive divergence the following holds, 
\begin{align}
\bD^A(\cS\otimes\cN \| \cT\otimes\cM ) = \bD^A(\cS\| \cT ) + \bD^A(\cN \|\cM ), \label{Eq:otimes1}
\end{align}
\end{lemma}
\begin{proof}
First note that the $\geq$ direction simply follows by convenient choice of the state subject to optimization and additivity. The $\leq$ direction in Equation~\ref{Eq:otimes1} follows directly from Lemma~\ref{lem:circ} and stability under additional identity channels, see Lemma~\ref{lem:AMOproperties}. 
\end{proof}

\subsection{ ... for quantum superchannels}\label{sec:divergences-s}

Let's now consider two superchannels $\Theta_1$ and $\Theta_2$ for which we would like to define similar measures. To save on writing indices, superchannels will always take a channel from $A$ to $B$ and transform it into a channel from $C$ to $D$. We begin by a generalization of the channel divergence, which is a special case of~\cite[Definition 1]{wang2019resource}.
\begin{definition}
 For two superchannels $\Theta_1$ and $\Theta_2$ and a generalized divergence $\bD$, we define the generalized superchannel divergence as
\begin{align}
\bD(\Theta_1 \| \Theta_2 ) = \sup_{\rho,\cN} \bD((\Theta_1\otimes\id_R)(\cN_{AR\rightarrow BR})(\rho_{CR}) \| (\Theta_2\otimes\id_R)(\cN_{AR\rightarrow BR})(\rho_{CR}) ). 
\end{align}
\end{definition}
With $\Theta_1\otimes\id_R$ we mean therefore that the superchannel doesn't act on the system $R$, however we optimize over a channel $\cN_{AR\rightarrow BR}$ that does. 
It was shown that the above definition obeys data processing in the sense that it is monotone under transformations that transform superchannels into general quantum networks~\cite[Theorem 4]{wang2019resource}. 

The above also gives a natural way to define a regularized superchannel divergence as
\begin{align}
\bD^\infty(\Theta_1 \| \Theta_2) = \lim_{n\rightarrow\infty} \frac{1}{n} \bD(\Theta_1^{\otimes n} \| \Theta_2^{\otimes n} ), 
\end{align}
where the optimization is effectively over $n$-party states and channels. However, one could also define intermediate version where only the states are n-partite and the channels are of product form,
\begin{align}
\bD^{s\infty}(\Theta_1 \| \Theta_2) = \lim_{n\rightarrow\infty} \frac{1}{n} \sup_{\cN} \bD(\Theta_1(\cN_{AR\rightarrow BR})^{\otimes n} \| \Theta_2(\cN_{AR\rightarrow BR})^{\otimes n} ),
\end{align}
or vice versa, where the channels are n-partite and the states are of product form,
\begin{align}
\bD^{c\infty}(\Theta_1 \| \Theta_2) = \lim_{n\rightarrow\infty} \frac{1}{n}\sup_{\rho,\cN} \bD((\Theta_1^{\otimes n}\otimes\id_R)(\cN_{A^nR\rightarrow B^nR})(\rho_{CR}^{\otimes n}) \| (\Theta_2^{\otimes n}\otimes\id_R)(\cN_{A^nR\rightarrow B^nR})(\rho_{CR}^{\otimes n}) ).
\end{align}
We will later see that these definitions are indeed of operational significance in certain superchannel discrimination scenarios. 

Similarly we can also define different extensions of the amortized channel divergences depending on whether one wants to amortize the action of the involved states, channels or both. 
\begin{definition}\label{Def:supAmoDiv}  For two superchannels $\Theta_1$ and $\Theta_2$ with decomposition $\{\cE_{\Theta_i},\cD_{\Theta_i}\}$ and a generalized divergence $\bD$, we define the state-amortized, channel-amortized, amortized and fully-amortized superchannel divergences as follows,
\begin{align}
\bD^{sA}(\Theta_1 \| \Theta_2 ) &= \sup_{\cN} \bD^A((\Theta_1\otimes\id_R)(\cN_{AR\rightarrow BR}) \| (\Theta_2\otimes\id_R)(\cN_{AR\rightarrow BR}) ), \\
\bD^{cA}(\Theta_1 \| \Theta_2 ) &= \sup_{\cN,\cM} \bD((\Theta_1\otimes\id_R)(\cN_{AR\rightarrow BR}) \| (\Theta_2\otimes\id_R)(\cM_{AR\rightarrow BR}) )  -  \bD^A(\cN \| \cM), \\
\bD^{A}(\Theta_1 \| \Theta_2 ) &= \sup_{ \cN,\cM} \bD^A((\Theta_1\otimes\id_R)(\cN_{AR\rightarrow BR}) \| (\Theta_2\otimes\id_R)(\cM_{AR\rightarrow BR}) ) \nonumber\\
&\phantom{ = \sup_{ \cN,\cM} } -  \bD^A(\cN \| \cM), \\
\bD^{A^*}(\Theta_1 \| \Theta_2 ) &= \sup_{ \substack{\cN,\cM \\ \bar\cN,\bar\cM}} \inf_{\cF} \bD^A((\Theta_1\otimes\id_R)(\cN_{AR\rightarrow BR})\circ\bar\cN \| (\Theta_2\otimes\id_R)(\cM_{AR\rightarrow BR})\circ\bar\cM ) \nonumber\\
&\phantom{ = \sup_{ \substack{\cN,\cM}} \inf_{\cF}} -  \bD^A(\cN\circ\cF\circ\bar\cN \| \cM\circ\cF\circ\bar\cM),
\end{align}
where the infimum is over arbitrary quantum channels $\cF_{C\rightarrow R'}$ with $R'$ being an additional reference system. 
\end{definition}
A rough intuition behind the different quantities is that $\bD^{sA}$ effectively handles a channel problem and amortizes its input state, $\bD^{cA}$ amortizes the distinguishability of the channel input but not the state, $\bD^{A}$ amortizes both, channel and state, and finally $\bD^{A^*}$ makes full use of the superchannel structure amortizing both inputs in terms of channel distinguishability.  Before we move on to show their operational meanings in the following sections, we will explore the relation between the different quantities. 
The next lemma states the intuitive observation that additional amortization can only increase the superchannel divergences. 
\begin{lemma}\label{lemma:ineq-DA}
For a faithful divergence $\bD$ we have, 
\begin{align}
 \bD(\Theta_1 \| \Theta_2 ) &\leq \bD^{sA}(\Theta_1 \| \Theta_2 ) \leq \bD^A(\Theta_1 \| \Theta_2 )\leq \bD^{A*}(\Theta_1 \| \Theta_2 ), \\
  \bD(\Theta_1 \| \Theta_2 ) &\leq \bD^{cA}(\Theta_1 \| \Theta_2 ) \leq \bD^A(\Theta_1 \| \Theta_2 ) \leq \bD^{A*}(\Theta_1 \| \Theta_2 ),
\end{align}
\end{lemma}
\begin{proof}
We go through the inequalities from right to left. Starting from $\bD^{A^*}(\Theta_1 \| \Theta_2 )$ we get to $\bD^{A}(\Theta_1 \| \Theta_2 )$ by setting $\bar\cN=\bar\cM=\id$ and using data processing to remove the channel $\cF$ in the second term. We then get to $\bD^{sA}(\Theta_1 \| \Theta_2 )$ by setting $\cN=\cM$ and then to get $\bD(\Theta_1 \| \Theta_2 )$ we use Lemma~\ref{lem:AMOproperties}. Similarly, we get from $\bD^{A}(\Theta_1 \| \Theta_2 )$ to $\bD^{cA}(\Theta_1 \| \Theta_2 )$ by using Lemma~\ref{lem:AMOproperties} and then to $\bD(\Theta_1 \| \Theta_2 )$ by setting $\cN=\cM$. 
\end{proof}

In the next section we will discuss the applications of these quantities in asymptotic quantum superchannel discrimination.

\section{Asymptotic bounds for superchannel discrimination}\label{sec:asymDisc}

In this section we will discuss the different strategies that are possible when one allows for multiple uses of the superchannel and what rates can be achieved asymptotically. 
We will for now focus on the asymptotic asymmetric discrimination setting, also known as Stein's setting, however the core results will be so called meta-converses which we prove in terms of generalized divergences. These will later, in Section~\ref{sec:BeyondStein}, allow us to apply the results also to other asymptotic discrimination settings. The essential setup is very similar to channel discrimination as described e.g. in~\cite{BHKW}, but we have to account for the additional freedom provided by the access point and in combining superchannels. 

Any strategy, independent of how we combine the superchannels, will ultimately result in either an output state $\rho_1$ if the superchannel was $\Theta_1$ or $\rho_2$ if the superchannel was $\Theta_2$. We measure that output with a binary POVM $\cQ=\{ Q_1, Q_2=\1 - Q_1 \}$, resulting in the usual type~I error $\alpha_n(\cS) = \tr(Q_2\rho_1)$ and type~II error $\beta_n(\cS)=\tr((\1-Q_2)\rho_2)$, where $\cS$ is the choice of strategy from the set of allowed strategies including the choice of measurement. 

For asymmetric hypothesis testing, we minimize the type~II error probability under the constraint that the type~I error probability does not exceed a constant $\eps\in(0,1)$. We are then interested in characterizing the non-asymptotic quantity
\begin{align}
\zeta^n(\eps,\Theta_1,\Theta_2)\coloneqq\sup_{\cS}\left\{-\frac{1}{n}\log\beta_n(\cS)\middle|\alpha_n(\cS)\leq\eps\right\},
\end{align}
where $\cS$ is the set of allowed strategies, as well as the asymptotic quantities
\begin{equation}
\underline{\zeta}(\eps,\Theta_1,\Theta_2)\coloneqq\liminf_{n \to \infty} \zeta^n(\eps,\Theta_1,\Theta_2), \qquad \overline{\zeta}(\eps,\Theta_1,\Theta_2)\coloneqq\limsup_{n \to \infty} \zeta^n(\eps,\Theta_1,\Theta_2).
\end{equation}
If the two limits coincide we call it simply $\zeta$. Often we will also consider the case where additionally $\epsilon$ goes to zero,  
\begin{equation}
\zeta(\Theta_1,\Theta_2)\coloneqq\lim_{\epsilon \to 0} \zeta(\epsilon, \Theta_1,\Theta_2).
\end{equation}

We will now start our investigation with the simplest possible strategy and then gradually move to more general strategies in the remainder of this section.

\subsection{Product strategy}\label{SSec:Prod}
This is the least powerful strategy that we will consider: One fixes $n$ copies of an input state $\rho_{CR}$ and of a channel $\cN_{AR\rightarrow BR}$ from $AR$ to $BR$, one pair for each copy of the superchannel, as depicted in Figure~\ref{Fig:parallel}a). Finally the resulting output state is measured by a binary measurement $\cQ=\{ Q_1, Q_2\}$ trying to determine which superchannel was used. It follows that the set of strategies to optimize over is given by the triple $\cS=\{ \rho, \cN, \cQ \}$. The product structure essentially reduces the problem to state discrimination by fixing $\rho$ and $\cN$ and discriminating the resulting output states which are again product states, therefore the relative entropy gives the optimal rate of discrimination by the Stein's Lemma for quantum states. Optimizing over all $\rho_{CR}$ and $\cN_{AR\rightarrow BR}$ shows that for these simple product strategies the optimal error rate is given by the superchannel divergence,
\begin{align}
\zeta_p(\epsilon,\Theta_1, \Theta_2) = D(\Theta_1 \| \Theta_2 ).
\end{align} 
Of course this strategy is extremely restricted and one can easily come up with more general discrimination strategies. It is however useful in setting a baseline, as every more general strategy naturally performs at least as good as the best product strategy. In the following, we will gradually move towards more complex strategies allowing for entangled states, many-party channels and intermediate adaptively chosen operations.


\subsection{Parallel strategy with product channels and successive adaptive strategies}

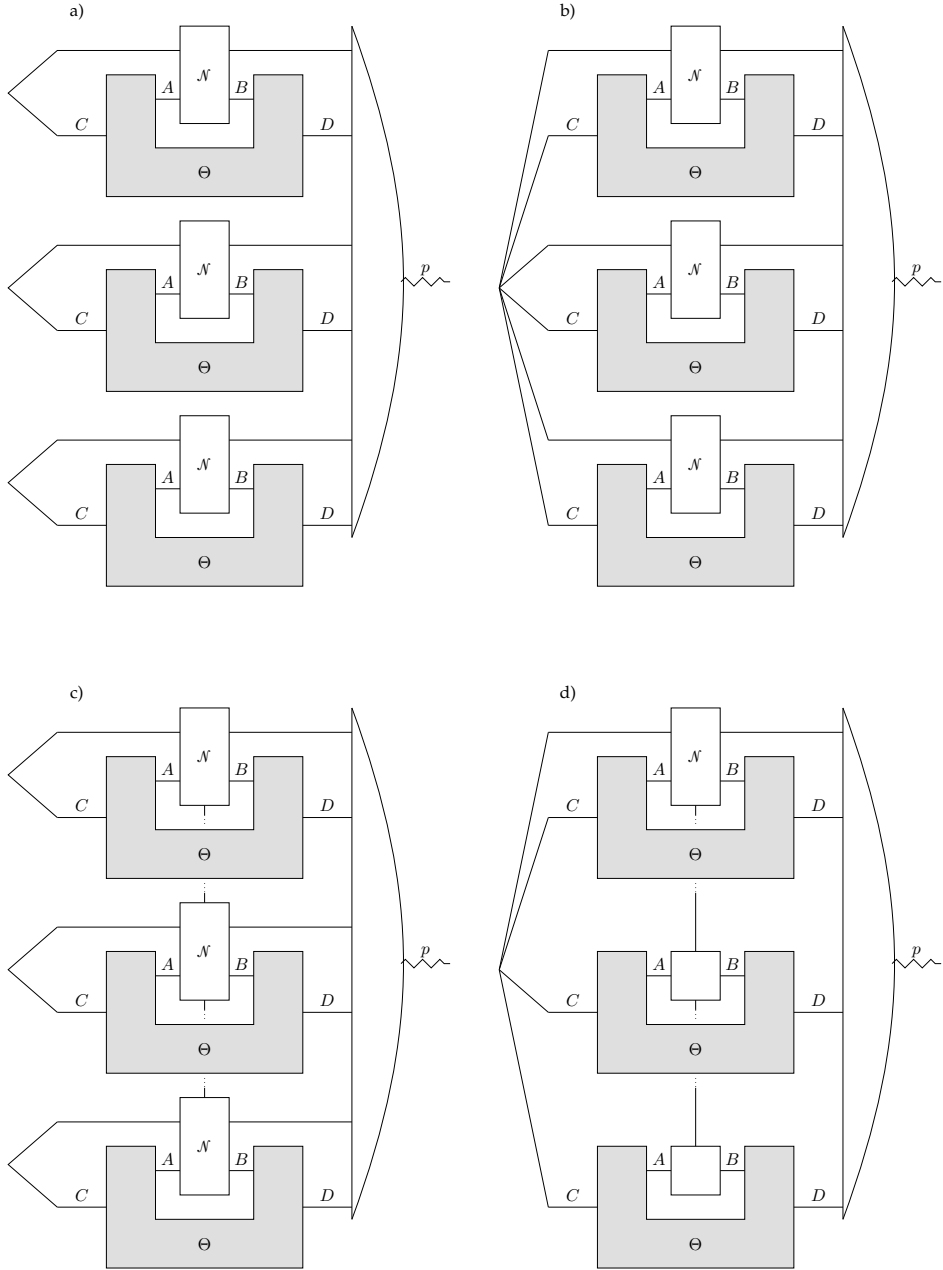
\begin{figure}[t!]
\centering
\resizebox{0.75\textwidth}{!}{
\begin{tikzpicture}
\node at (-2.1,2.3){a)};
\draw[] (0,0) rectangle (1,2) node[pos=0.5]{$\cN$};
\draw[fill=SCcolor] (-1.5,1) -- (-0.5,1) -- (-0.5,-0.5) -- (1.5,-0.5) -- (1.5,1) -- (2.5,1) -- (2.5,-1.5) -- (-1.5,-1.5) -- cycle;
\draw (-0.5,0.5) -- (0,0.5) node[pos=0.5,sloped,above]{$A$};
\draw (1,0.5) -- (1.5,0.5) node[pos=0.5,sloped,above]{$B$};
\draw (-2.5,-0.25) -- (-1.5,-0.25) node[pos=0.5,sloped,above]{$C$};
\draw (2.5,-0.25) -- (3.5,-0.25) node[pos=0.5,sloped,above]{$D$};
\node at (0.5,-1){$\Theta$};
\draw (-2.5,-0.25) -- (-3.5,0.625) -- (-2.5,1.5) -- (0,1.5);
\draw (1,1.5) -- (3.5,1.5);

\draw[] (0,0-4) rectangle (1,2-4) node[pos=0.5]{$\cN$};
\draw[fill=SCcolor] (-1.5,1-4) -- (-0.5,1-4) -- (-0.5,-0.5-4) -- (1.5,-0.5-4) -- (1.5,1-4) -- (2.5,1-4) -- (2.5,-1.5-4) -- (-1.5,-1.5-4) -- cycle;
\draw (-0.5,0.5-4) -- (0,0.5-4) node[pos=0.5,sloped,above]{$A$};
\draw (1,0.5-4) -- (1.5,0.5-4) node[pos=0.5,sloped,above]{$B$};
\draw (-2.5,-0.25-4) -- (-1.5,-0.25-4) node[pos=0.5,sloped,above]{$C$};
\draw (2.5,-0.25-4) -- (3.5,-0.25-4) node[pos=0.5,sloped,above]{$D$};
\node at (0.5,-1-4){$\Theta$};
\draw (-2.5,-0.25-4) -- (-3.5,0.625-4) -- (-2.5,1.5-4) -- (0,1.5-4);
\draw (1,1.5-4) -- (3.5,1.5-4);

\draw[] (0,0-8) rectangle (1,2-8) node[pos=0.5]{$\cN$};
\draw[fill=SCcolor] (-1.5,1-8) -- (-0.5,1-8) -- (-0.5,-0.5-8) -- (1.5,-0.5-8) -- (1.5,1-8) -- (2.5,1-8) -- (2.5,-1.5-8) -- (-1.5,-1.5-8) -- cycle;
\draw (-0.5,0.5-8) -- (0,0.5-8) node[pos=0.5,sloped,above]{$A$};
\draw (1,0.5-8) -- (1.5,0.5-8) node[pos=0.5,sloped,above]{$B$};
\draw (-2.5,-0.25-8) -- (-1.5,-0.25-8) node[pos=0.5,sloped,above]{$C$};
\draw (2.5,-0.25-8) -- (3.5,-0.25-8) node[pos=0.5,sloped,above]{$D$};
\node at (0.5,-1-8){$\Theta$};
\draw (-2.5,-0.25-8) -- (-3.5,0.625-8) -- (-2.5,1.5-8) -- (0,1.5-8);
\draw (1,1.5-8) -- (3.5,1.5-8);

\draw (3.5,2) -- (3.5,-8.5);
\draw[] (3.5,2) to [out=290,in=70] (3.5,-8.5);
\draw[snake] (2.5+2,-3.25) -- (3.5+2,-3.25) node[pos=0.5,sloped,above]{$p$};


\node at (-2.1+10,2.3){b)};
\draw[] (0+10,0) rectangle (1+10,2) node[pos=0.5]{$\cN$};
\draw[fill=SCcolor] (-1.5+10,1) -- (-0.5+10,1) -- (-0.5+10,-0.5) -- (1.5+10,-0.5) -- (1.5+10,1) -- (2.5+10,1) -- (2.5+10,-1.5) -- (-1.5+10,-1.5) -- cycle;
\draw (-0.5+10,0.5) -- (0+10,0.5) node[pos=0.5,sloped,above]{$A$};
\draw (1+10,0.5) -- (1.5+10,0.5) node[pos=0.5,sloped,above]{$B$};
\draw (-2.5+10,-0.25) -- (-1.5+10,-0.25) node[pos=0.5,sloped,above]{$C$};
\draw (2.5+10,-0.25) -- (3.5+10,-0.25) node[pos=0.5,sloped,above]{$D$};
\node at (0.5+10,-1){$\Theta$};
\draw (-2.5+10,1.5) -- (0+10,1.5);
\draw (1+10,1.5) -- (3.5+10,1.5);

\draw[] (0+10,0-4) rectangle (1+10,2-4) node[pos=0.5]{$\cN$};
\draw[fill=SCcolor] (-1.5+10,1-4) -- (-0.5+10,1-4) -- (-0.5+10,-0.5-4) -- (1.5+10,-0.5-4) -- (1.5+10,1-4) -- (2.5+10,1-4) -- (2.5+10,-1.5-4) -- (-1.5+10,-1.5-4) -- cycle;
\draw (-0.5+10,0.5-4) -- (0+10,0.5-4) node[pos=0.5,sloped,above]{$A$};
\draw (1+10,0.5-4) -- (1.5+10,0.5-4) node[pos=0.5,sloped,above]{$B$};
\draw (-2.5+10,-0.25-4) -- (-1.5+10,-0.25-4) node[pos=0.5,sloped,above]{$C$};
\draw (2.5+10,-0.25-4) -- (3.5+10,-0.25-4) node[pos=0.5,sloped,above]{$D$};
\node at (0.5+10,-1-4){$\Theta$};
\draw (-2.5+10,1.5-4) -- (0+10,1.5-4);
\draw (1+10,1.5-4) -- (3.5+10,1.5-4);

\draw[] (0+10,0-8) rectangle (1+10,2-8) node[pos=0.5]{$\cN$};
\draw[fill=SCcolor] (-1.5+10,1-8) -- (-0.5+10,1-8) -- (-0.5+10,-0.5-8) -- (1.5+10,-0.5-8) -- (1.5+10,1-8) -- (2.5+10,1-8) -- (2.5+10,-1.5-8) -- (-1.5+10,-1.5-8) -- cycle;
\draw (-0.5+10,0.5-8) -- (0+10,0.5-8) node[pos=0.5,sloped,above]{$A$};
\draw (1+10,0.5-8) -- (1.5+10,0.5-8) node[pos=0.5,sloped,above]{$B$};
\draw (-2.5+10,-0.25-8) -- (-1.5+10,-0.25-8) node[pos=0.5,sloped,above]{$C$};
\draw (2.5+10,-0.25-8) -- (3.5+10,-0.25-8) node[pos=0.5,sloped,above]{$D$};
\node at (0.5+10,-1-8){$\Theta$};
\draw (-2.5+10,1.5-8) -- (0+10,1.5-8);
\draw (1+10,1.5-8) -- (3.5+10,1.5-8);

\draw (-2.5+10,1.5) -- (-3.5+10,-3.375);
\draw (-2.5+10,-0.25) -- (-3.5+10,-3.375);
\draw (-2.5+10,1.5-4) -- (-3.5+10,-3.375);
\draw (-2.5+10,-0.25-4) -- (-3.5+10,-3.375);
\draw (-2.5+10,1.5-8) -- (-3.5+10,-3.375);
\draw (-2.5+10,-0.25-8) -- (-3.5+10,-3.375);

\draw (3.5+10,2) -- (3.5+10,-8.5);
\draw[] (3.5+10,2) to [out=290,in=70] (3.5+10,-8.5);
\draw[snake] (2.5+2+10,-3.25) -- (3.5+2+10,-3.25) node[pos=0.5,sloped,above]{$p$};


\node at (-2.1,2.3-14){c)};
\draw[] (0,0-14) rectangle (1,2-14) node[pos=0.5]{$\cN$};
\draw[fill=SCcolor] (-1.5,1-14) -- (-0.5,1-14) -- (-0.5,-0.5-14) -- (1.5,-0.5-14) -- (1.5,1-14) -- (2.5,1-14) -- (2.5,-1.5-14) -- (-1.5,-1.5-14) -- cycle;
\draw (-0.5,0.5-14) -- (0,0.5-14) node[pos=0.5,sloped,above]{$A$};
\draw (1,0.5-14) -- (1.5,0.5-14) node[pos=0.5,sloped,above]{$B$};
\draw (-2.5,-0.25-14) -- (-1.5,-0.25-14) node[pos=0.5,sloped,above]{$C$};
\draw (2.5,-0.25-14) -- (3.5,-0.25-14) node[pos=0.5,sloped,above]{$D$};
\node at (0.5,-1-14){$\Theta$};
\draw (-2.5,1.5-14) -- (0,1.5-14);
\draw (1,1.5-14) -- (3.5,1.5-14);
\draw (-2.5,-0.25-14) -- (-3.5,0.625-14) -- (-2.5,1.5-14);

\draw[] (0,0-4-14) rectangle (1,2-4-14) node[pos=0.5]{$\cN$};
\draw[fill=SCcolor] (-1.5,1-4-14) -- (-0.5,1-4-14) -- (-0.5,-0.5-4-14) -- (1.5,-0.5-4-14) -- (1.5,1-4-14) -- (2.5,1-4-14) -- (2.5,-1.5-4-14) -- (-1.5,-1.5-4-14) -- cycle;
\draw (-0.5,0.5-4-14) -- (0,0.5-4-14) node[pos=0.5,sloped,above]{$A$};
\draw (1,0.5-4-14) -- (1.5,0.5-4-14) node[pos=0.5,sloped,above]{$B$};
\draw (-2.5,-0.25-4-14) -- (-1.5,-0.25-4-14) node[pos=0.5,sloped,above]{$C$};
\draw (2.5,-0.25-4-14) -- (3.5,-0.25-4-14) node[pos=0.5,sloped,above]{$D$};
\node at (0.5,-1-4-14){$\Theta$};
\draw (-2.5,1.5-4-14) -- (0,1.5-4-14);
\draw (1,1.5-4-14) -- (3.5,1.5-4-14);
\draw (-2.5,-0.25-14-4) -- (-3.5,0.625-14-4) -- (-2.5,1.5-14-4);

\draw[] (0,0-8-14) rectangle (1,2-8-14) node[pos=0.5]{$\cN$};
\draw[fill=SCcolor] (-1.5,1-8-14) -- (-0.5,1-8-14) -- (-0.5,-0.5-8-14) -- (1.5,-0.5-8-14) -- (1.5,1-8-14) -- (2.5,1-8-14) -- (2.5,-1.5-8-14) -- (-1.5,-1.5-8-14) -- cycle;
\draw (-0.5,0.5-8-14) -- (0,0.5-8-14) node[pos=0.5,sloped,above]{$A$};
\draw (1,0.5-8-14) -- (1.5,0.5-8-14) node[pos=0.5,sloped,above]{$B$};
\draw (-2.5,-0.25-8-14) -- (-1.5,-0.25-8-14) node[pos=0.5,sloped,above]{$C$};
\draw (2.5,-0.25-8-14) -- (3.5,-0.25-8-14) node[pos=0.5,sloped,above]{$D$};
\node at (0.5,-1-8-14){$\Theta$};
\draw (-2.5,1.5-8-14) -- (0,1.5-8-14);
\draw (1,1.5-8-14) -- (3.5,1.5-8-14);
\draw (-2.5,-0.25-14-8) -- (-3.5,0.625-14-8) -- (-2.5,1.5-14-8);

\draw[] (0.5,-0-14) -- (0.5,-0.2-14);
\draw[dotted] (0.5,-0.2-14) -- (0.5,-0.4-14);
\draw[dotted] (0.5,-1.6-14) -- (0.5,-1.8-14);
\draw[] (0.5,-1.8-14) -- (0.5,-2-14);
\draw[] (0.5,-0-4-14) -- (0.5,-0.2-4-14);
\draw[dotted] (0.5,-0.2-4-14) -- (0.5,-0.4-4-14);
\draw[dotted] (0.5,-1.6-4-14) -- (0.5,-1.8-4-14);
\draw[] (0.5,-1.8-4-14) -- (0.5,-2-4-14);


\draw (3.5,2-14) -- (3.5,-8.5-14);
\draw[] (3.5,2-14) to [out=290,in=70] (3.5,-8.5-14);
\draw[snake] (2.5+2,-3.25-14) -- (3.5+2,-3.25-14) node[pos=0.5,sloped,above]{$p$};


\node at (-2.1+10,2.3-14){d)};
\draw[] (0+10,0-14) rectangle (1+10,2-14) node[pos=0.5]{$\cN$};
\draw[fill=SCcolor] (-1.5+10,1-14) -- (-0.5+10,1-14) -- (-0.5+10,-0.5-14) -- (1.5+10,-0.5-14) -- (1.5+10,1-14) -- (2.5+10,1-14) -- (2.5+10,-1.5-14) -- (-1.5+10,-1.5-14) -- cycle;
\draw (-0.5+10,0.5-14) -- (0+10,0.5-14) node[pos=0.5,sloped,above]{$A$};
\draw (1+10,0.5-14) -- (1.5+10,0.5-14) node[pos=0.5,sloped,above]{$B$};
\draw (-2.5+10,-0.25-14) -- (-1.5+10,-0.25-14) node[pos=0.5,sloped,above]{$C$};
\draw (2.5+10,-0.25-14) -- (3.5+10,-0.25-14) node[pos=0.5,sloped,above]{$D$};
\node at (0.5+10,-1-14){$\Theta$};
\draw (-2.5+10,1.5-14) -- (0+10,1.5-14);
\draw (1+10,1.5-14) -- (3.5+10,1.5-14);

\draw[] (0+10,0-4-14) rectangle (1+10,1-4-14);
\draw[fill=SCcolor] (-1.5+10,1-4-14) -- (-0.5+10,1-4-14) -- (-0.5+10,-0.5-4-14) -- (1.5+10,-0.5-4-14) -- (1.5+10,1-4-14) -- (2.5+10,1-4-14) -- (2.5+10,-1.5-4-14) -- (-1.5+10,-1.5-4-14) -- cycle;
\draw (-0.5+10,0.5-4-14) -- (0+10,0.5-4-14) node[pos=0.5,sloped,above]{$A$};
\draw (1+10,0.5-4-14) -- (1.5+10,0.5-4-14) node[pos=0.5,sloped,above]{$B$};
\draw (-2.5+10,-0.25-4-14) -- (-1.5+10,-0.25-4-14) node[pos=0.5,sloped,above]{$C$};
\draw (2.5+10,-0.25-4-14) -- (3.5+10,-0.25-4-14) node[pos=0.5,sloped,above]{$D$};
\node at (0.5+10,-1-4-14){$\Theta$};

\draw[] (0+10,0-8-14) rectangle (1+10,1-8-14);
\draw[fill=SCcolor] (-1.5+10,1-8-14) -- (-0.5+10,1-8-14) -- (-0.5+10,-0.5-8-14) -- (1.5+10,-0.5-8-14) -- (1.5+10,1-8-14) -- (2.5+10,1-8-14) -- (2.5+10,-1.5-8-14) -- (-1.5+10,-1.5-8-14) -- cycle;
\draw (-0.5+10,0.5-8-14) -- (0+10,0.5-8-14) node[pos=0.5,sloped,above]{$A$};
\draw (1+10,0.5-8-14) -- (1.5+10,0.5-8-14) node[pos=0.5,sloped,above]{$B$};
\draw (-2.5+10,-0.25-8-14) -- (-1.5+10,-0.25-8-14) node[pos=0.5,sloped,above]{$C$};
\draw (2.5+10,-0.25-8-14) -- (3.5+10,-0.25-8-14) node[pos=0.5,sloped,above]{$D$};
\node at (0.5+10,-1-8-14){$\Theta$};

\draw (-2.5+10,1.5-14) -- (-3.5+10,-3.375-14);
\draw (-2.5+10,-0.25-14) -- (-3.5+10,-3.375-14);
\draw (-2.5+10,-0.25-4-14) -- (-3.5+10,-3.375-14);
\draw (-2.5+10,-0.25-8-14) -- (-3.5+10,-3.375-14);

\draw[] (0.5+10,-0-14) -- (0.5+10,-0.2-14);
\draw[dotted] (0.5+10,-0.2-14) -- (0.5+10,-0.4-14);
\draw[dotted] (0.5+10,-1.6-14) -- (0.5+10,-1.8-14);
\draw[] (0.5+10,-1.8-14) -- (0.5+10,-3-14);
\draw[] (0.5+10,-0-4-14) -- (0.5+10,-0.2-4-14);
\draw[dotted] (0.5+10,-0.2-4-14) -- (0.5+10,-0.4-4-14);
\draw[dotted] (0.5+10,-1.6-4-14) -- (0.5+10,-1.8-4-14);
\draw[] (0.5+10,-1.8-4-14) -- (0.5+10,-3-4-14);

\draw (3.5+10,2-14) -- (3.5+10,-8.5-14);
\draw[] (3.5+10,2-14) to [out=290,in=70] (3.5+10,-8.5-14);
\draw[snake] (2.5+2+10,-3.25-14) -- (3.5+2+10,-3.25-14) node[pos=0.5,sloped,above]{$p$};

\end{tikzpicture}}
\caption{\label{Fig:parallel} Different parallel strategies for $n=3$: a) Product strategy. Here both the input state and the intermediate channel are of product form. b) Parallel strategy with product channels. This strategy allows for an entangled input state. c) Parallel strategy with product states. This strategy uses only product states but allows for joint quantum channel $\cN$ (as indicated by the connecting line). d) Parallel strategy: This is the most general parallel strategy allowing for both, an entangled input state and a joint quantum channel $\cN$.}
\end{figure}

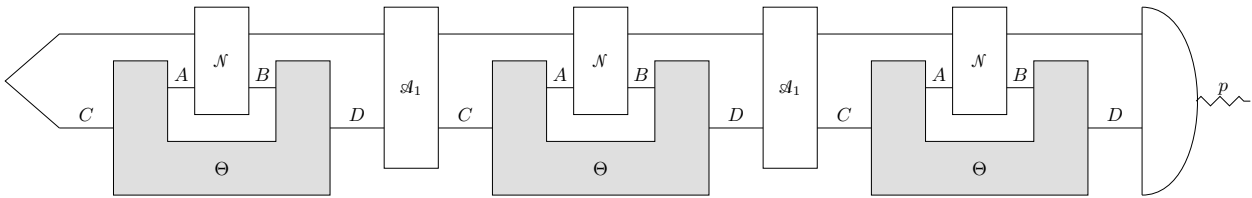
\begin{figure}[t]
\centering
\resizebox{\textwidth}{!}{
\begin{tikzpicture}
\draw[] (0,0) rectangle (1,2) node[pos=0.5]{$\cN$};
\draw[fill=SCcolor] (-1.5,1) -- (-0.5,1) -- (-0.5,-0.5) -- (1.5,-0.5) -- (1.5,1) -- (2.5,1) -- (2.5,-1.5) -- (-1.5,-1.5) -- cycle;
\draw (-0.5,0.5) -- (0,0.5) node[pos=0.5,sloped,above]{$A$};
\draw (1,0.5) -- (1.5,0.5) node[pos=0.5,sloped,above]{$B$};
\draw (-2.5,-0.25) -- (-1.5,-0.25) node[pos=0.5,sloped,above]{$C$};
\draw (2.5,-0.25) -- (3.5,-0.25) node[pos=0.5,sloped,above]{$D$};
\node at (0.5,-1){$\Theta$};

\draw (-2.5,-0.25) -- (-3.5,0.625) -- (-2.5,1.5) -- (0,1.5);
\draw[] (3.5,-1) rectangle (4.5,2) node[pos=0.5]{$\cA_1$};
\draw (1,1.5) -- (3.5,1.5);
\draw (4.5,1.5) -- (7,1.5);
\draw (1+7,1.5) -- (3.5+7,1.5);
\draw (4.5+7,1.5) -- (7+7,1.5);
\draw (8+7,1.5) -- (10.5+7,1.5);

\draw[] (0+7,0) rectangle (1+7,2) node[pos=0.5]{$\cN$};
\draw[fill=SCcolor] (-1.5+7,1) -- (-0.5+7,1) -- (-0.5+7,-0.5) -- (1.5+7,-0.5) -- (1.5+7,1) -- (2.5+7,1) -- (2.5+7,-1.5) -- (-1.5+7,-1.5) -- cycle;
\draw (-0.5+7,0.5) -- (0+7,0.5) node[pos=0.5,sloped,above]{$A$};
\draw (1+7,0.5) -- (1.5+7,0.5) node[pos=0.5,sloped,above]{$B$};
\draw (-2.5+7,-0.25) -- (-1.5+7,-0.25) node[pos=0.5,sloped,above]{$C$};
\draw (2.5+7,-0.25) -- (3.5+7,-0.25) node[pos=0.5,sloped,above]{$D$};
\node at (0.5+7,-1){$\Theta$};

\draw[] (3.5+7,-1) rectangle (4.5+7,2) node[pos=0.5]{$\cA_1$};

\draw[] (0+14,0) rectangle (1+14,2) node[pos=0.5]{$\cN$};
\draw[fill=SCcolor] (-1.5+14,1) -- (-0.5+14,1) -- (-0.5+14,-0.5) -- (1.5+14,-0.5) -- (1.5+14,1) -- (2.5+14,1) -- (2.5+14,-1.5) -- (-1.5+14,-1.5) -- cycle;
\draw (-0.5+14,0.5) -- (0+14,0.5) node[pos=0.5,sloped,above]{$A$};
\draw (1+14,0.5) -- (1.5+14,0.5) node[pos=0.5,sloped,above]{$B$};
\draw (-2.5+14,-0.25) -- (-1.5+14,-0.25) node[pos=0.5,sloped,above]{$C$};
\draw (2.5+14,-0.25) -- (3.5+14,-0.25) node[pos=0.5,sloped,above]{$D$};
\node at (0.5+14,-1){$\Theta$};

\draw (17.5,2) -- (17.5,-1.5);
\draw[] (17.5,2) to [out=0,in=0] (17.5,-1.5);
\draw[snake] (2.5+16,0.25) -- (3.5+16,0.25) node[pos=0.5,sloped,above]{$p$};

\end{tikzpicture}}
\caption{\label{Fig:suc-adap} The successive adaptive strategy for $n=3$. }
\end{figure}

Parallel strategies are more general than the previous product strategy as we allow for joint inputs. Here we start by considering a joint entangled input state to all superchannels but for now keep the limitation to product channels. This approach is already more complicated to analyze, nevertheless it is similar in spirit. One fixes a channel $\cN_{AR\rightarrow BR}$ for the superchannel to act on. This results in a new channel $\Theta_i\otimes\id_R(\cN)$ and the problem is now equivalent to attempting channel discrimination with a parallel strategy. In this case we know that the optimal rate is the regularized channel relative entropy and by optimizing over all channels $\cN_{AR\rightarrow BR}$ we get the optimal rate for superchannel discrimination in the form of the state-regularized channel relative entropy, 
\begin{align}
\zeta_{sp}(\Theta_1,\Theta_2) = D^{s\infty} (\Theta_1\|\Theta_2). 
\end{align}
Parallel strategies with product channels are, as in the channel case, a special case of certain adaptive strategies. We will later see that while these particular adaptive strategies, as depicted in Figure~\ref{Fig:suc-adap}, are the most general in the channel case, they are not for superchannels. One can construct different adaptive strategies than the ones discussed in in this section and we therefore name this particular class \textit{successive adaptive strategies}. As before one can simply fix a channel $\cN_{AR\rightarrow BR}$ and consider the task at hand a problem of channel discrimination which gives the amortized relative entropy as optimal error rate. Optimizing over the channels $\cN_{AR\rightarrow BR}$ gives that the state-amortized relative entropy is the optimal rate achievable with successive adaptive strategies when discrimination superchannels, 
\begin{align}
\zeta_{sa}(\Theta_1,\Theta_2) = D^{sA} (\Theta_1\|\Theta_2). 
\end{align}

It follows from the non-additivity of the channel relative entropy that parallel strategies with product channels can be strictly more powerful than the product strategies in the previous section. 
Another observation that can be made regarding the amortized quantity is that, similar to the channel version, one can apply the chain rule from~\cite{fang2019chain} to get an upper bound. In our case we have 
\begin{align}
D((\Theta_1\otimes\id_R)(\cN_{AR\rightarrow BR})(\rho_{CR}) &\| (\Theta_2\otimes\id_R)(\cN_{AR\rightarrow BR})(\tau_{CR}) )  -  D(\rho \| \tau ) \nonumber\\ 
&\leq \bar D^\infty((\Theta_1\otimes\id_R)(\cN_{AR\rightarrow BR}) \| (\Theta_2\otimes\id_R)(\cN_{AR\rightarrow BR}) ), 
\end{align}
where $\bar D^\infty$ denotes $ D^\infty$ without an additional reference following the notation in~\cite{fang2019chain}. 
Optimizing gives
\begin{align}
D^{sA}(\Theta_1 \| \Theta_2 ) \leq D^{s\infty}(\Theta_1 \| \Theta_2).  \label{Eq:inAmReg}
\end{align}
Since we know that parallel strategies with product channels are a special case of sequential adaptive strategies we have immediately that 
\begin{align}
D^{sA}(\Theta_1 \| \Theta_2 ) = D^{s\infty}(\Theta_1 \| \Theta_2),  \label{Eq:eqAmReg}
\end{align}
implying that successive adaptive strategies do not provide an advantage over parallel strategies with product channels. 

Note that for all results in this section we had to take the limit of $\epsilon\rightarrow 0$ as we do not know whether the given rates are strong converse rates. We remark that via the reduction to a channel problem one can easily get strong converse bounds in terms of the max-relative entropy~\cite{BHKW} or the amortized geometric Renyi divergence~\cite{fang2019geometric}, see also Remark~\ref{Rmk:strongConverse} for more details.

The results in this section are all based on optimizing over product channels. Clearly, we are not making full use of the additional possibilities provided by a superchannel and in the next section we will discuss what happens when we lift this restriction.

\subsection{Parallel strategy with multi-party channels}\label{Sec:ParMPC}

One gets a first hint that the structure of superchannel discrimination is much richer than that of channel discrimination by considering a parallel strategy for which one allows the use of an $n$-party channel instead of limiting ourselves to $n$ separate channels, see Figure~\ref{Fig:parallel}d). It is not clear whether these strategies can be cast as a successive adaptive strategy and therefore one could expect that they provide a strictly more powerful class of strategies for superchannel discrimination. 

One of the main ingredient for most of our converse proofs will be the following observation from the original proof of Stein's lemma for quantum states~\cite{HP91}, see also~\cite[Proposition 16]{BHKW}. Let $p$ and $q$ be the binary probability distributions  resulting from measuring the final output state of a given discrimination strategy, then
\begin{align}
D(p\| q)  &= (1-p) \log\frac{1-p}{1-q} + p \log(p/q) \nonumber\\ \label{Eq:relEntLB}
&\geq - h_2(\varepsilon) - (1-\varepsilon)\log\beta_n(\cS),
\end{align}
where $D(p\| q)$ is the classical relative entropy for two binary probability distributions.
By rearranging the above equation, it follows that
\begin{align}
-\log\beta_n(\cS) \leq \frac{1}{1-\varepsilon}\Big(D(p\| q) +h_2(\varepsilon)\Big).
\end{align}

From this and data processing of the relative entropy we get by standard arguments that a weak converse in the full parallel setting is given by
\begin{align}
\zeta^n_{fp}(\epsilon,\Theta_1,\Theta_2) \leq \frac{1}{1-\varepsilon}\left(D^{\infty} (\Theta_1\|\Theta_2) +\frac{h_2(\varepsilon)}{n}\right).
\end{align}
One can further envision an intermediate strategy with interchanged focus on the optimization, namely one allows for arbitrary channels but restricts the input states to product states, see Figure~\ref{Fig:parallel}c). It should be noted that also for this simpler strategy it is not clear how to cast it as a successive adaptive strategy. In this case we can use the same approach and get the channel-regularized relative entropy as weak converse rate
\begin{align}
\zeta^n_{cp}(\epsilon,\Theta_1,\Theta_2) \leq \frac{1}{1-\varepsilon}\left(D^{c\infty} (\Theta_1\|\Theta_2)+\frac{h_2(\varepsilon)}{n}\right).
\end{align}

We are left with showing that these rates are also achievable. The main ingredient in this proof will be the following results from~\cite{tomamichel2013hierarchy, li2014second}:
\begin{align}
D_h^\epsilon(\rho^{\otimes n}\| \sigma^{\otimes n}) = n D(\rho\|\sigma) + \sqrt{n V(\rho\|\sigma)} \Phi^{-1}(\epsilon) + O(\log n), \label{Eq:expansion}
\end{align}
where $V(\rho\|\sigma) := \tr(\rho(\log\rho - \log\sigma)^2)$ is the quantum information variance and $\Phi$ is the cumulative normal distribution. 

\begin{lemma}
For two superchannels $\Theta_1$ and $\Theta_2$, considering fully parallel strategies, we have
\begin{align}
\zeta_{fp}(\epsilon,\Theta_1,\Theta_2) \geq D^{\infty} (\Theta_1\|\Theta_2). 
\end{align}
\end{lemma}
\begin{proof}
Let's say we have $n := k m$ copies of the superchannels. We prepare $k$ copies of an arbitrary input state on $C^mR$ and a channel $\cN:A^mR \rightarrow B^mR$. We now send each state through $m$ parallel superchannels acting on $\cN$, leaving us with the output state 
\begin{align}
(\rho_i)^{\otimes k} = \left((\Theta_i^{\otimes m}\otimes\id_R)(\cN)(\rho)\right)^{\otimes k}.
\end{align}
By Equation~\eqref{Eq:expansion} we have
\begin{align}
D_h^\epsilon((\rho_1)^{\otimes k}\| (\rho_2)^{\otimes k}) = k D(\rho_1\|\rho_2) + \sqrt{k V(\rho_1\|\rho_2)} \Phi^{-1}(\epsilon) + O(\log k),
\end{align}
now dividing by $n$ and taking the limit $k\rightarrow\infty$ (and therefore $n\rightarrow\infty$) gives
\begin{align}
\lim_{k\rightarrow\infty} \frac1n D_h^\epsilon((\rho_1)^{\otimes k}\| (\rho_2)^{\otimes k}) = \frac1m D(\rho_1\|\rho_2).
\end{align}
Finally, we can choose $m$ arbitrarily large and pick the optimal $\rho$ and $\cN$. Noticing that the strategy used here is a special case of the fully general parallel strategy, the statement of the lemma follows.
\end{proof}
In summary, we have
\begin{align}
\zeta_{fp}(\Theta_1,\Theta_2) = D^{\infty} (\Theta_1\|\Theta_2), 
\end{align}
giving the superchannel Stein's lemma for fully parallel strategies. 
We can get an equivalent result for the parallel strategies with product states.
\begin{lemma}
For two superchannels $\Theta_1$ and $\Theta_2$, considering parallel strategies restricted to product states, we have
\begin{align}
\zeta_{cp}(\epsilon,\Theta_1,\Theta_2) \geq D^{c\infty} (\Theta_1\|\Theta_2). 
\end{align}
\end{lemma}
\begin{proof}
The proof for the parallel strategy with product states works exactly as that of the fully general parallel strategy.
\end{proof}
This gives,
\begin{align}
\zeta_{cp}(\Theta_1,\Theta_2) = D^{c\infty} (\Theta_1\|\Theta_2). 
\end{align}
We have therefore successfully determined the optimal rates for a Stein's lemma for superchannels using different parallel strategies. 
Finally, strong converse rates in terms of different divergences follow with similar proof altercations as mentioned for product strategies and we refer to Remark~\ref{Rmk:strongConverse} for details.

\subsection{Nested adaptive strategies}\label{Sec:NestedAd}

\begin{figure}[t!]
\centering
\resizebox{\textwidth}{!}{
\begin{tikzpicture}
\draw[] (0,0) rectangle (1,2) node[pos=0.5]{$\cA_3$};
\draw[fill=SCcolor] (-1.5,1) -- (-0.5,1) -- (-0.5,-0.5) -- (1.5,-0.5) -- (1.5,1) -- (2.5,1) -- (2.5,-1.5) -- (-1.5,-1.5) -- cycle;
\draw (-0.5,0.5) -- (0,0.5) node[pos=0.5,sloped,above]{$A$};
\draw (1,0.5) -- (1.5,0.5) node[pos=0.5,sloped,above]{$B$};
\draw (-2.5-8,-0.25) -- (-1.5-8,-0.25) node[pos=0.5,sloped,above]{$C$};
\draw (-2.5-2-4,-0.25) -- (-1.5-2-4,-0.25) node[pos=0.5,sloped,above]{$A$};
\draw (-2.5-4,-0.25) -- (-1.5-4,-0.25) node[pos=0.5,sloped,above]{$C$};
\draw (-2.5-2,-0.25) -- (-1.5-2,-0.25) node[pos=0.5,sloped,above]{$A$};
\draw (-2.5,-0.25) -- (-1.5,-0.25) node[pos=0.5,sloped,above]{$C$};
\draw (2.5,-0.25) -- (3.5,-0.25) node[pos=0.5,sloped,above]{$D$};
\draw (4.5,-0.25) -- (5.5,-0.25) node[pos=0.5,sloped,above]{$B$};
\draw (6.5,-0.25) -- (7.5,-0.25) node[pos=0.5,sloped,above]{$D$};
\draw (8.5,-0.25) -- (9.5,-0.25) node[pos=0.5,sloped,above]{$B$};
\draw (6.5+4,-0.25) -- (7.5+4,-0.25) node[pos=0.5,sloped,above]{$D$};
\draw (-2.5-8,1.5) -- (-1.5-8,1.5) node[pos=0.5,sloped,above]{$R$};
\draw (10.5,-0.25) -- (11.5,-0.25);

\draw[] (-7.5,-1) rectangle (-6.5,2) node[pos=0.5]{$\cA_1$};
\draw[] (-3.5,-1) rectangle (-2.5,2) node[pos=0.5]{$\cA_2$};
\draw[] (3.5,-1) rectangle (4.5,2) node[pos=0.5]{$\cA_4$};
\draw[] (7.5,-1) rectangle (8.5,2) node[pos=0.5]{$\cA_5$};

\draw (-5.5-2,1.5) -- (-6.5-4,1.5) -- (-7.5-4,0.625) -- (-6.5-4,-0.25) -- (-5.5-4,-0.25);
\draw (-3.5,1.5) -- (-6.5,1.5) node[pos=0.5,sloped,above]{$R$};
\draw (-2.5,1.5) -- (0,1.5) node[pos=0.5,sloped,above]{$R$};
\draw (1,1.5) -- (3.5,1.5) node[pos=0.5,sloped,above]{$R$};
\draw (4.5,1.5) -- (7.5,1.5) node[pos=0.5,sloped,above]{$R$};
\draw (8.5,1.5) -- (11.5,1.5) node[pos=0.5,sloped,above]{$R$};

\draw[fill=SCcolor] (-5.5,1) -- (-4.5,1) -- (-4.5,-2) -- (5.5,-2) -- (5.5,1) -- (6.5,1) -- (6.5,-3) -- (-5.5,-3) -- cycle;

\draw[fill=SCcolor] (-9.5,1) -- (-8.5,1) -- (-8.5,-3.5) -- (9.5,-3.5) -- (9.5,1) -- (10.5,1) -- (10.5,-4.5) -- (-9.5,-4.5) -- cycle;

\node at (0.5,-1){$\Theta$};
\node at (0.5,-1-1.5){$\Theta$};
\node at (0.5,-1-3){$\Theta$};

\draw (11.5,2) -- (11.5,-1.5);
\draw[] (11.5,2) to [out=0,in=0] (11.5,-1.5);
\draw[snake] (2.5+10,0.25) -- (3.5+10,0.25) node[pos=0.5,sloped,above]{$p$};

\end{tikzpicture}}
\caption{\label{Fig:nes-adap} The nested adaptive strategy for $n=3$. }
\end{figure}
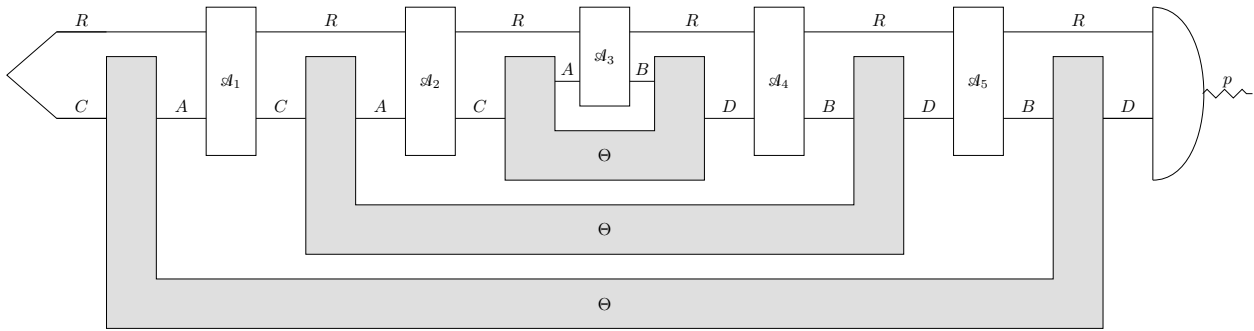

As mentioned in the previous section, parallel strategies allowing for n-party channels don't seem to be a special case of successive adaptive strategies, however, one can cast them as a different class of adaptive strategies which we call \textit{nested adaptive strategies}, see Figure~\ref{Fig:nes-adap}. For a visualization of how to embed parallel strategies into nested adaptive ones see also Appendix~\ref{App:Bonus}. In this section, we will discuss this class and give a meta-converse bound in terms of the amortized superchannel divergence. Also note that of course parallel strategies with product channels are a special case of nested adaptive strategies (because general parallel strategies are), but successive adaptive strategies do not seem to be in this class. This can be seen as furthering the intuition that superchannels should ultimately be seen as a function on channels rather than states. Additionally, since successive adaptive strategies do not outperform parallel strategies with product channels which in turn are a special case of nested adaptive strategies, the latter are always at least as good as the successive adaptive ones. 

We begin by formalizing the class of nested adaptive strategies. The basic structure can be taken from Figure~\ref{Fig:nes-adap}. The strategy now consists of an input state $\rho_{C_1R_1}$, $2n-1$ adaptively chosen channels $\cA_i$ and a measurement $Q$. Therefore, we optimize over all $\cS = \{ \rho, \{\cA_i\}_{i=1}^{2n-1}, \cQ\}$. We will describe the strategy in an iterative fashion.

To start, we choose an input state $\rho_{C_1R_1}$ and define a sequence of channels $\cN_i$ (or $\cM_i$) from $C_iR_i$ to $D_iR'_i$ in the following way:  $\cN_1 = \Theta_1(\cA_n)$ (or $\cM_1 = \Theta_2(\cA_n)$) and now every following channel in the sequence is defined via
\begin{align}
\cN_{i+1}  &=  \Theta_1\left(\cA_{n -i}\circ\cN_i\circ\cA_{n+i}\right), \\
\cM_{i+1}  &=  \Theta_2\left(\cA_{n -i}\circ\cM_i\circ\cA_{n+i}\right).
\end{align}
Now, the outcome of a $n$-step nested adaptive strategy can simply be given as $\cN_n(\rho_{C_1R_1})$ (or $\cM_n(\rho_{C_1R_1})$).   Afterwards we measure these resulting states with the measurement $\cQ$ and get a classical binary distributions $p$ or $q$. Following the idea of the meta converse in~\cite{BHKW}, we can now get the following result.
\begin{theorem}[Meta-converse for nested adaptive strategies]
For two superchannels $\Theta_1$ and $\Theta_2$, we have for any $n$-round nested adaptive strategy,
\begin{align}
\bD(p \| q) \leq \bD(\cN_n(\rho_{C_1R_1})\| \cM_n(\rho_{C_1R_1})) \leq n \bD^{A}(\Theta_1 \| \Theta_2 ). 
\end{align}
\end{theorem}
\begin{proof}
We begin by fixing an arbitrary strategy $\cS$ from all possible nested adaptive protocol for discrimination of the superchannels $\Theta_1$ and $\Theta_2$ and let $p$ and $q$ denote the final decision probabilities.
Now, consider the following chain of arguments, 
\begin{align}
&\mathbf{D}(p\Vert q) \nonumber\\
&\leq\mathbf{D}(\cN_n(\rho_{C_1R_1})\Vert\cM_n(\rho_{C_1R_1})) \nonumber\\
&=\mathbf{D}(\Theta_1\left(\cA_{1}\circ\cN_{n -1}\circ\cA_{2n-1}\right)(\rho_{C_1R_1})\Vert\Theta_2\left(\cA_{1}\circ\cM_{n -1}\circ\cA_{2n-1}\right)(\rho_{C_1R_1})) \nonumber\\
&\leq\mathbf{D}(\Theta_1\left(\cA_{1}\circ\cN_{n -1}\circ\cA_{2n-1}\right)\Vert\Theta_2\left(\cA_{1}\circ\cM_{n -1}\circ\cA_{2n-1}\right)) \nonumber\\
&\leq\bD^A(\Theta_1\left(\cA_{1}\circ\cN_{n -1}\circ\cA_{2n-1}\right)\Vert\Theta_2\left(\cA_{1}\circ\cM_{n -1}\circ\cA_{2n-1}\right)) \nonumber\\
&=\bD^A(\Theta_1\left(\cA_{1}\circ\cN_{n -1}\circ\cA_{2n-1}\right)\Vert\Theta_2\left(\cA_{1}\circ\cM_{n -1}\circ\cA_{2n-1}\right)) \nonumber\\
&\quad - \bD^A(\cA_{1}\circ\cN_{n -1}\circ\cA_{2n-1}\Vert\cA_{1}\circ\cM_{n -1}\circ\cA_{2n-1}) +  \bD^A(\cA_{1}\circ\cN_{n -1}\circ\cA_{2n-1}\Vert\cA_{1}\circ\cM_{n -1}\circ\cA_{2n-1}) \nonumber\\
&=\bD^A(\Theta_1(\cN)\Vert\Theta_2(\cM)) - \bD^A(\cN\Vert\cM) + \bD^A(\cA_{1}\circ\cN_{n -1}\circ\cA_{2n-1}\Vert\cA_{1}\circ\cM_{n -1}\circ\cA_{2n-1}), \label{eq:bridge-step1}
\end{align}
The first inequality follows from data processing under the final measurement $\cQ$. The second inequality follows by taking a supremum over all $\rho$, the third because amortization makes channel divergences only bigger. The first two equalities are straightforward and the third follows by making the following substitutions,
\begin{align*}
\cN &= \cA_{1}\circ\cN_{n -1}\circ\cA_{2n-1} \\
\cM &= \cA_{1}\circ\cM_{n -1}\circ\cA_{2n-1}.
\end{align*}
With this we continue as follows, 
\begin{align}
&\text{Eq.~\eqref{eq:bridge-step1}} \nonumber\\
&\leq \sup_{\cN,\cM}\left[\bD^A(\Theta_1(\cN)\Vert\Theta_2(\cM)) - \bD^A(\cN\Vert\cM)\right] + \bD^A(\cA_{1}\circ\cN_{n -1}\circ\cA_{2n-1}\Vert\cA_{1}\circ\cM_{n -1}\circ\cA_{2n-1}) \nonumber\\
&= \mathbf{D}^{A}(\Theta_1\Vert\Theta_2) + \bD^A(\cA_{1}\circ\cN_{n -1}\circ\cA_{2n-1}\Vert\cA_{1}\circ\cM_{n -1}\circ\cA_{2n-1}) \nonumber\\
&\leq \mathbf{D}^{A}(\Theta_1\Vert\Theta_2) + \bD^A(\cN_{n -1}\Vert\cM_{n -1}), 
\end{align}
where the first inequality follows by optimizing, the first equality by definition and the final inequality because the amortized channel divergence obeys data-processing  under superchannels. In summary, we showed that 
\begin{align*}
\mathbf{D}(p\Vert q)&\leq\mathbf{D}^{A}(\cN_n\Vert\cM_n) \leq \mathbf{D}^{A}(\Theta_1\Vert\Theta_2) + \bD^A(\cN_{n -1}\Vert\cM_{n -1}).
\end{align*}
Now, we can simply apply the same procedure another $n-1$ times and we get that 
\begin{align}
\mathbf{D}(p\Vert q)&\leq n \mathbf{D}^{A}(\Theta_1\Vert\Theta_2), 
\end{align}
which concludes the proof.
\end{proof}

It follows from Equation~\eqref{Eq:relEntLB} that the amortized superchannel relative entropy is a weak converse for a superchannel Stein's Lemma considering all nested adaptive strategies, 
\begin{align}\label{eq:weak-conv-stein-unconstrained}
\zeta^n_{na}(\eps,\Theta_1,\Theta_2)\leq \frac{1}{1-\eps} \left(D^{A}(\Theta_1\|\Theta_2)+\frac{h_2(\varepsilon)}{n}\right). 
\end{align}

Knowing that the amortized superchannel relative entropy is a valid converse raises the question whether nested adaptive strategies can perform better than parallel strategies. Since the latter are a subset of the former, we know already that 
\begin{align}
D^{\infty}(\Theta_1\|\Theta_2) = \zeta_{fp}(\Theta_1,\Theta_2) \leq \zeta_{na}(\Theta_1,\Theta_2) \leq D^{A}(\Theta_1\|\Theta_2).  \label{Eq:inftyAmortized}
\end{align}
Proving the inequality in the opposite direction,  
\begin{align}
D^{A}(\Theta_1\|\Theta_2) \stackrel{?}{\leq} D^{\infty}(\Theta_1\|\Theta_2),
\end{align}
would show that nested adaptive strategies do not outperform parallel strategies. In Appendix~\ref{Ap:chainRules}, we show that under a technical assumption, namely an asymptotic equipartition property for the smooth channel max-relative entropy, Equation~\eqref{Eq:inftyAmortized} indeed becomes an equality. We conjecture that the assumption is true in general and therefore that parallel strategies are as powerful as nested adaptive strategies. The argument in Appendix~\ref{Ap:chainRules} relies on a novel chain rule for relative entropy under superchannels. 

Note, that if the two quantities are indeed the same, we immediately learn that $D^{A}(\Theta_1\|\Theta_2)$ is also an achievable rate for nested adaptive strategies (and parallel ones). An alternative route to proving achievability directly, could be to generalize the achievability proof for channels in~\cite[Theorem 6]{wang2019resource}. The idea there is to use an appropriate number of channel uses within the resource theory of asymmetric distinguishability to recursively prepare the optimal input states in the amortized channel divergence. 

The careful reader might notice that earlier in this manuscript we defined the quantity $D^{cA}(\Theta_1\|\Theta_2)$ that is not currently associated to any strategy. The interest in it stems mostly from the observation that under the same technical assumption as before one can show that
\begin{align}
D^{cA}(\Theta_1\|\Theta_2) \stackrel{?}{\leq} D^{c\infty}(\Theta_1\|\Theta_2), \label{Eq:regAmoCEq}
\end{align}
as shown in Appendix~\ref{Ap:chainRules}. For a complete picture it would be desirable to find a class of adaptive strategies for which $D^{cA}(\Theta_1\|\Theta_2)$ is a converse rate and that includes parallel strategies with product states. We conjecture that the inequality in Equation~\eqref{Eq:regAmoCEq} holds and will turn out to be indeed even an equality.

\begin{remark}\label{Rmk:strongConverse}
Note that for all week converses presented in this section, one could replace Equation~\eqref{Eq:relEntLB} with
\begin{align}
D_{\max}(p\Vert q) & = \log \max\{(1-\varepsilon) / q, \varepsilon / (1-q)\}\geq \log(1-\varepsilon) - \log q.
\end{align}
and we get a strong converse bound based on the max-relative entropy instead of the relative entropy with
\begin{align}
\zeta^n_{na}(\eps,\Theta_1,\Theta_2) \leq D^A_{\max}(\Theta_1\| \Theta_2)  + \frac{1}{n}\log\!\left(\frac{1}{1-\varepsilon}\right).
\end{align}
Similarly we can get a strong converse bound in terms of the geometric Renyi divergence via~\cite{fang2019geometric}
\begin{align}
\widehat D_\alpha(p\|q) \geq \frac{\alpha}{\alpha-1} \log(1-\epsilon) - \log q, 
\end{align}
leading to 
\begin{align}
\zeta^n_{na}(\eps,\Theta_1,\Theta_2) \leq \widehat D^A_{\alpha}(\Theta_1\| \Theta_2)  + \frac{1}{n}\frac{\alpha}{\alpha-1}\log\!\left(\frac{1}{1-\varepsilon}\right).
\end{align}
In the channel case, and therefore the succesive adaptive strategies for superchannels, we have that the amortization collapses for both the amortized channel max-relative entropy~\cite{BHKW} and the amortized channel geometric Renyi divergence~\cite{fang2019geometric}. We leave it open whether this also holds in more general superchannel scenarios. 
\end{remark}

\subsection{Fully general adaptive strategies} 

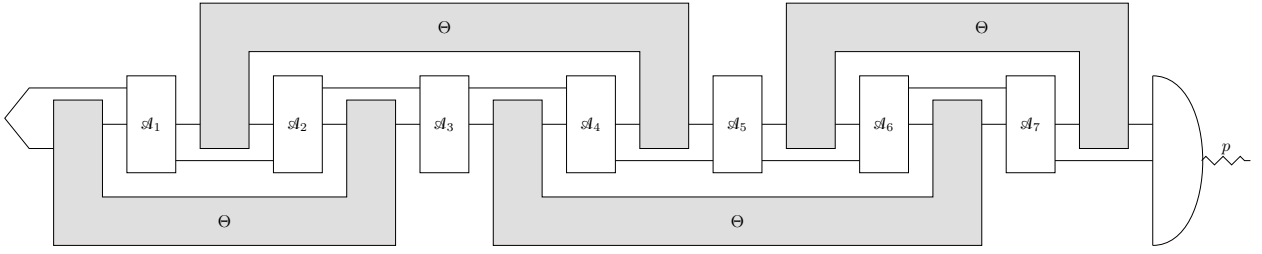
\begin{figure}[t]
\centering
\resizebox{\textwidth}{!}{
\begin{tikzpicture}
\draw (-1.5,1) -- (21,1);
\draw[fill=white] (0,0) rectangle (1,2) node[pos=0.5]{$\cA_1$};
\draw[fill=white] (0+3,0) rectangle (1+3,2) node[pos=0.5]{$\cA_2$};
\draw[fill=white] (0+6,0) rectangle (1+6,2) node[pos=0.5]{$\cA_3$};
\draw[fill=white] (0+9,0) rectangle (1+9,2) node[pos=0.5]{$\cA_4$};
\draw[fill=white] (0+12,0) rectangle (1+12,2) node[pos=0.5]{$\cA_5$};
\draw[fill=white] (0+15,0) rectangle (1+15,2) node[pos=0.5]{$\cA_6$};
\draw[fill=white] (0+18,0) rectangle (1+18,2) node[pos=0.5]{$\cA_7$};
\draw[fill=SCcolor] (-1.5,1.5) -- (-0.5,1.5) -- (-0.5,-0.5) -- (4.5,-0.5) -- (4.5,1.5) -- (5.5,1.5) -- (5.5,-1.5) -- (-1.5,-1.5) -- cycle;
\draw[fill=SCcolor] (1.5,0.5) -- (2.5,0.5) -- (2.5,2.5) -- (10.5,2.5) -- (10.5,0.5) -- (11.5,0.5) -- (11.5,3.5) -- (1.5,3.5) -- cycle;
\draw[fill=SCcolor] (-1.5+9,1.5) -- (-0.5+9,1.5) -- (-0.5+9,-0.5) -- (7.5+9,-0.5) -- (7.5+9,1.5) -- (8.5+9,1.5) -- (8.5+9,-1.5) -- (-1.5+9,-1.5) -- cycle;
\draw[fill=SCcolor] (1.5+12,0.5) -- (2.5+12,0.5) -- (2.5+12,2.5) -- (7.5+12,2.5) -- (7.5+12,0.5) -- (8.5+12,0.5) -- (8.5+12,3.5) -- (1.5+12,3.5) -- cycle;

\draw (-1.5,0.5) -- (-2,0.5) -- (-2.5,0.5+0.625) -- (-2,1.75) -- (0,1.75);

\draw (1,0.25) -- (3,0.25);
\draw (10,0.25) -- (12,0.25);
\draw (13,0.25) -- (15,0.25);
\draw (19,0.25) -- (21,0.25);
\draw (4,1.75) -- (6,1.75);
\draw (7,1.75) -- (9,1.75);
\draw (16,1.75) -- (18,1.75);

\node at (2,-1){$\Theta$};
\node at (12.5,-1){$\Theta$};
\node at (6.5,3){$\Theta$};
\node at (17.5,3){$\Theta$};

\draw (11.5+0.5+9,2) -- (11.5+0.5+9,-1.5);
\draw[] (11.5+0.5+9,2) to [out=0,in=0] (11.5+0.5+9,-1.5);
\draw[snake] (2.5+0.5+10+9,0.25) -- (3.5+0.5+10+9,0.25) node[pos=0.5,sloped,above]{$p$};

\end{tikzpicture}}
\caption{\label{Fig:braided-adap} The braided adaptive strategy for $n=4$. }
\end{figure}

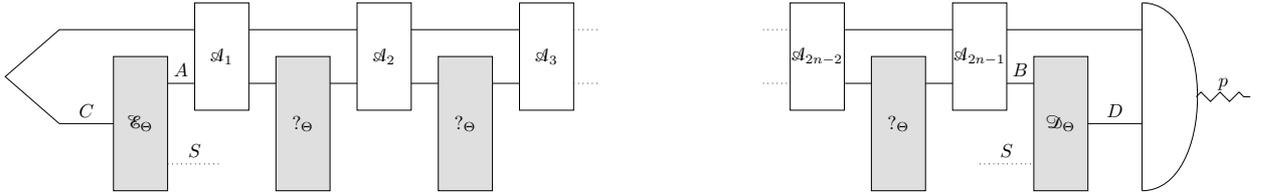
\begin{figure}[t]
\centering
\resizebox{\textwidth}{!}{
\begin{tikzpicture}
\draw[] (0,0) rectangle (1,2) node[pos=0.5]{$\cA_1$};
\draw[fill=SCcolor] (-1.5,1) rectangle (-0.5,-1.5) node[pos=0.5]{$\cE_\Theta$};
\draw (-0.5,0.5) -- (0,0.5) node[pos=0.5,sloped,above]{$A$};
\draw[dotted] (-0.5,-1) -- (0.5,-1) node[pos=0.5,sloped,above]{$S$};
\draw (1,0.5) -- (1.5,0.5);
\draw (2.5,0.5) -- (3,0.5);
\draw[dotted] (4+3,0.5) -- (4.5+3,0.5);
\draw[dotted] (4+3,1.5) -- (4.5+3,1.5);
\draw (-2.5,-0.25) -- (-1.5,-0.25) node[pos=0.5,sloped,above]{$C$};

\draw (1+3,0.5) -- (1.5+3,0.5);
\draw (2.5+3,0.5) -- (3+3,0.5);
\draw (1+3,1.5) -- (3+3,1.5);
\draw[] (0+3,0) rectangle (1+3,2) node[pos=0.5]{$\cA_2$};
\draw[fill=SCcolor] (-1.5+3,1) rectangle (-0.5+3,-1.5) node[pos=0.5]{$?_\Theta$};

\draw[] (0+6,0) rectangle (1+6,2) node[pos=0.5]{$\cA_3$};
\draw[fill=SCcolor] (-1.5+6,1) rectangle (-0.5+6,-1.5) node[pos=0.5]{$?_\Theta$};

\draw (-2.5,-0.25) -- (-3.5,0.625) -- (-2.5,1.5) -- (0,1.5);
\draw (1,1.5) -- (3,1.5);
\draw (5+7,1.5) -- (7+7,1.5);
\draw (8+7,1.5) -- (10.5+7,1.5);

\draw[] (0+11,0) rectangle (1+11,2) node[pos=0.5]{$\cA_{2n-2}$};
\draw[fill=SCcolor] (1.5+11,1) rectangle (2.5+11,-1.5) node[pos=0.5]{$?_\Theta$};
\draw[dotted] (-0.5+11,1.5) -- (0+11,1.5);
\draw[dotted] (-0.5+11,0.5) -- (0+11,0.5);
\draw (1+11,0.5) -- (1.5+11,0.5);

\draw[] (0+14,0) rectangle (1+14,2) node[pos=0.5]{$\cA_{2n-1}$};
\draw[fill=SCcolor] (1.5+14,1) rectangle (2.5+14,-1.5) node[pos=0.5]{$\cD_\Theta$};
\draw (-0.5+14,0.5) -- (0+14,0.5);
\draw (1+14,0.5) -- (1.5+14,0.5) node[pos=0.5,sloped,above]{$B$};
\draw (2.5+14,-0.25) -- (3.5+14,-0.25) node[pos=0.5,sloped,above]{$D$};
\draw[dotted] (0.5+14,-1) -- (1.5+14,-1) node[pos=0.5,sloped,above]{$S$};

\draw (17.5,2) -- (17.5,-1.5);
\draw[] (17.5,2) to [out=0,in=0] (17.5,-1.5);
\draw[snake] (2.5+16,0.25) -- (3.5+16,0.25) node[pos=0.5,sloped,above]{$p$};

\end{tikzpicture}}
\caption{\label{Fig:gen-adap} Every adaptive strategy can be depicted by deconstructing the superchannel $\Theta$ into its components $\cE_\Theta$, $\cD_\Theta$. While the first and last channel are fixed, the remaining channels can be of either type, depicted by $?_\Theta$, only restricted by 1) each channel has to appear a total of $n$ times, 2) each channel $\cE_\Theta$ has to appear before the $\cD_\Theta$ which belongs to the same $\Theta$. For simplicity we omit most of the systems $S$, they are naturally connecting the $\cE_\Theta$, $\cD_\Theta$ belonging to the same $\Theta$.}
\end{figure}

With the nested adaptive strategies in the previous section we have seen a class of strategies that is unique to the superchannel case in the sense that there is no comparable set of strategies in the case of quantum channels. We have also learned that these strategies are at least as powerful as the most general parallel strategies. One might at first be tempted to think that this exhausts the possibilities of superchannel discrimination. However, one also quickly notices that looking at successive and nested adaptive strategies, neither of the two classes seems to be included in the other one and it could in principle be useful to employ a mixture of the two strategies, e.g. use $\frac{n}{2}$ rounds of successive, followed by $\frac{n}{2}$ rounds of nested adaptive superchannel uses. 

But that is not all yet: All strategies discussed so far either place superchannels as a whole before, after or within another superchannel. However, in principle there is no a priori reason to limit the use of superchannels in this way. As an example for a strategy that does not follow this rule we will discuss what we call a \textit{braided adaptive strategy} as depicted in Figure~\ref{Fig:braided-adap}.
In this section we will finally discuss a converse bound on quantum superchannel discrimination that provably also holds for the most general class of strategies.

Now, how can we ensure a description that covers all possible strategies of a superchannel? The key to this question is to take apart the superchannel into its components, meaning instead of discussing the superchannel $\Theta$ we will state the strategy in terms of the underlying channels $\{ \cE_\Theta, \cD_\Theta \}$. We now allow for any adaptive strategy using the channels in any order that doesn't violate the superhannel structure, i.e. the channel $\cE_\Theta$ of a particular use of $\Theta$ has to be used before the corresponding channel $\cD_\Theta$. It should also be noted that the adaptively chosen intermediate channels $\cA_i$ are not allowed to act on the $S$ systems as this would give a forbidden advantage. The resulting structure is depicted in Figure~\ref{Fig:gen-adap}.

Our goal will now be to find a converse bound that includes all possible strategies. As described above we will view the strategies following Figure~\ref{Fig:gen-adap}: We start with an input state $\rho^0_{CR}$ to which we apply $\cE_{\Theta_i}$ followed by an adaptively chosen map $\cA_1$ and continuing by alternating applications of superchannel fragments and adaptive operations, finally ending in a last application of  $\cD_{\Theta_i}$. We will now show that for any such strategy we can give a meta-converse. As for the nested strategies, the set of possible strategies is given by $\cS = \{ \rho, \{\cA_i\}_{i=1}^{2n-1}, Q\}$, with the difference being that components of the superchannels can appear in any order between the adaptively chosen channels $\cA_i$. 

\begin{theorem}[Meta-converse for arbitrary strategies]\label{thm:mc-general}
For any adaptive superchannel discrimination strategy, we have
\begin{align}
\mathbf{D}(p\Vert q) \leq n \mathbf{D}^{A^*}(\Theta_1\|\Theta_2). \label{Eq:sec-ineq-gC}
\end{align}
\end{theorem}

\begin{proof}
As usual we begin with the divergence between the two possible output probability distributions,
\begin{align}
&\mathbf{D}(p\Vert q) \nonumber\\
&\leq\mathbf{D}( \cD_{\Theta_1} \circ \cA_{2n-1} \circ \dots \circ \cA_1 \circ \cE_{\Theta_1}(\rho) \| \cD_{\Theta_2} \circ \cA_{2n-1} \circ \dots \circ \cA_1 \circ \cE_{\Theta_2}(\rho)) \nonumber\\
&\leq\mathbf{D}^A( \cD_{\Theta_1} \circ \cA_{2n-1} \circ \dots \circ \cA_1 \circ \cE_{\Theta_1} \| \cD_{\Theta_2} \circ \cA_{2n-1} \circ \dots \circ \cA_1 \circ \cE_{\Theta_2})  \label{Eq:initialStrategy}\\
&=  \mathbf{D}^A( \Theta_1( \cA_{2n-1} \circ \dots \circ \cA_{j+1})\circ \cA_j \circ\dots\circ \cA_1 \circ \cE_{\Theta_1}\|  \Theta_2( \cA_{2n-1} \circ \dots \circ \cA_{j+1})\circ \cA_j \circ\dots\circ \cA_1 \circ \cE_{\Theta_2}) \label{eq:bridge-stepfg1}
\end{align}
where the first line is data processing and explicitly writing out the discrimination strategy, the second a property of amortization and the equality follows by explicitly writing the superchannel that belongs to the final $\cD_{\Theta_i}$ assuming that $\cA_j$ is the last adaptive map before the corresponding $\cE_{\Theta_i}$.

The sequence $\cA_{2n-1} \circ \dots \circ \cA_{j+1}$ and the sequence $\cA_j \circ\dots\circ \cA_1$ may both include components of other copies of the superchannel $\Theta_i$. 
For ease of notation, we define 
\begin{align}
\cN_i = \cA_{2n-1} \circ \dots \circ \cA_{j+1}, \\
\cM_i = \cA_j \circ\dots\circ \cA_1\circ \cE_{\Theta_i},
\end{align}
where $i\in\{1,2\}$ depending on whether the superchannel fragments are from $\Theta_1$ or $\Theta_2$, i.e. whether they appear in the first or second argument. Let $\cF$ be the channel that achieves the infimum in $ \mathbf{D}^{A^*}(\Theta_1\|\Theta_2)$. 
We now continue the derivation with
\begin{align}
&\text{Eq.~\eqref{eq:bridge-stepfg1}} \nonumber\\
=&  \mathbf{D}^A( \Theta_1( \cN_1)\circ \cM_1\|  \Theta_2( \cN_2)\circ \cM_2) \nonumber\\
=&  \bD^A( \Theta_1( \cN_1)\circ \cM_1\|  \Theta_2( \cN_2)\circ \cM_2) - \bD^A(  \cN_1\circ\cF\circ \cM_1\|  \cN_2\circ\cF\circ \cM_2) \nonumber\\
&+ \bD^A(  \cN_1\circ\cF\circ \cM_1\|  \cN_2\circ\cF\circ \cM_2) \nonumber\\
\leq& \sup_{\cN,\cM,\bar\cN,\bar\cM} \left[ \bD^A( \Theta_1( \cN)\circ \cM\|  \Theta_2( \bar\cN)\circ \bar\cM) - \bD^A(  \cN\circ\cF\circ \cM\|  \bar\cN\circ\cF\circ \bar\cM) \right] \nonumber\\
&+ \bD^A(  \cN_1\circ\cF\circ \cM_1\|  \cN_2\circ\cF\circ \cM_2) \nonumber\\
=& \mathbf{D}^{A^*}(\Theta_1\|\Theta_2)+  \bD^A(  \cN_1\circ\cF\circ \cM_1\|  \cN_2\circ\cF\circ \cM_2) \nonumber\\
=& \mathbf{D}^{A^*}(\Theta_1\|\Theta_2)   \nonumber\\
&+\bD^A(  \cA_{2n-1} \circ \dots \circ \cA_{j+1}\circ\cF\circ \cA_j \circ\dots\circ \cA_1\circ \cE_{\Theta_1}\|  \cA_{2n-1} \circ \dots \circ \cA_{j+1}\circ\cF\circ \cA_j \circ\dots\circ \cA_1\circ \cE_{\Theta_2}) \nonumber\\
\leq& \mathbf{D}^{A^*}(\Theta_1\|\Theta_2)   \nonumber\\
&+\bD^A(  \cD_{\Theta_1} \circ\cA_{2n-2} \circ \dots \circ \cA_{j+1}\circ\cF\circ \cA_j \circ\dots\circ \cA_1\circ \cE_{\Theta_1}\|  \nonumber\\
&\phantom{+\bD^A( }\cD_{\Theta_1} \circ\cA_{2n-2} \circ \dots \circ \cA_{j+1}\circ\cF\circ \cA_j \circ\dots\circ \cA_1\circ \cE_{\Theta_2})
\end{align}
where the first equality is by definition, the second simply adding a zero, the first inequality follows by optimizing over the channels in the first two terms and the third equality is the definition of the amortized superchannel divergence. The final equality follows by definition and the final inequality is removing $\cA_{2n-1}$ via data processing. Note that the final amortized channel divergence is taken between quantum channels resulting from a discrimination strategy that corresponds to the initial one after removing the final superchannel. 

Now, observe that the remaining amortized channel divergence has, after merging $\cA_{j+1}\circ\cF\circ \cA_j $ into a single channel and appropriate relabeling, the same form as the initial divergence in Equation~\eqref{Eq:initialStrategy}, but for a discrimination strategy based on $n-1$ applications of the superchannel. Iterating the same procedure as above to remove all $n$ superchannels results in
\begin{align}
&\mathbf{D}(p\Vert q) \leq n  \mathbf{D}^{A^*}(\Theta_1\|\Theta_2) , 
\end{align}
which concludes the proof. 
\end{proof}

Now, with the same technique as described in the previous sections, we get a weak converse bound for a Stein's Lemma with arbitrary strategies,
\begin{align}
\zeta^n_{fg}(\eps,\Theta_1,\Theta_2) \leq  \frac{1}{1-\eps} \left(  D^{A^*}(\Theta_1\|\Theta_2) +\frac1n h_2(\varepsilon)\right).  \label{Eq:fg-converse-Stein}
\end{align}

Note that of course for certain strategies one can also mix the previously used proof strategies to get convex combinations of amortized superchannel divergences, e.g. in the aforementioned case where one applies $\frac{n}{2}$ rounds of successive, followed by $\frac{n}{2}$ rounds of nested adaptive superchannel uses. On the other hand, the braided adaptive strategy from Figure~\ref{Fig:braided-adap} gives an example where, based on our proof technique, one seems to necessarily always end up with using $D^{A^*}$. 

These observations lead to a host of interesting questions that demand further investigation. Most importantly, is $D^{A^*}$ achievable or can one find a tighter converse bound? If it is optimal, are there cases where $D^{A^*}>D^{A}$, which would then imply that adaptive strategies are strictly more powerful then parallel strategies. For example, does there exist an example where a braided adaptive strategy is strictly better than all nested adaptive strategies? 

Finally we remark that similar to the previous sections we can get strong converse bounds in terms of the amortized max-relative entropy or the amortized geometric Renyi-divergence using the same technique as in Remark~\ref{Rmk:strongConverse}. It is however again unclear whether all involved amortized quantities collapse when using max-relative entropy or the geometric Renyi-divergence.

\begin{table}[t!]
\centering
{
\begin{tabular}{ |p{5.2cm}||p{4.3cm}|p{2cm}|p{1.6cm}|  }
 \hline
 \multicolumn{4}{|c|}{Summary of results for superchannels} \\
 \hline
 Strategy & Converse bound (Stein) & Achievable? & Figure \\
 \hline
Product   & $D(\Theta_1 \| \Theta_2 )$    & \cmark &   Fig.~\ref{Fig:parallel}~a\\
Parallel with product channels &   $D^{s\infty}(\Theta_1 \| \Theta_1 )$~*  & \cmark   &   Fig.~\ref{Fig:parallel}~b\\
Parallel with product states &   $D^{c\infty}(\Theta_1 \| \Theta_1 )$  & \cmark   &   Fig.~\ref{Fig:parallel}~c\\
 Parallel & $D^{\infty}(\Theta_1 \| \Theta_1 )$~\tiny\textcopyright & \cmark &     Fig.~\ref{Fig:parallel}~d\\
 Succesive adaptive  & $D^{sA}(\Theta_1 \| \Theta_1 )$~* &  \cmark &  Fig.~\ref{Fig:suc-adap}\\
 Nested adaptive &   $D^{A}(\Theta_1 \| \Theta_1 )$~\tiny\textcopyright  & \tiny\textcopyright & Fig.~\ref{Fig:nes-adap}\\
Braided adaptive & $D^{A*}(\Theta_1 \| \Theta_1 )$  & ?   & Fig.~\ref{Fig:braided-adap}\\ 
 Fully general adaptive & $D^{A*}(\Theta_1 \| \Theta_1 )$  & ?   & Fig.~\ref{Fig:gen-adap}\\ \hline
\end{tabular}}
\caption{List of investigated strategies for superchannel discrimination along with the converse bounds provided in this work for the Stein's setting. Rates marked with * are equal, therefore choosing the more general class doesn't provide any advantage. Rates with {\tiny\textcopyright} are equal given an assumption discussed in Appendix~\ref{Ap:chainRules} and if true $D^A$ is also achievable. The third column states whether the given converse bound is optimal, i.e. it can be achieved by a particular strategy.}
\label{table:ResultsSC}
\end{table}

\section{Examples}\label{sec:examples}

In this section, we will discuss several examples of superchannels and discuss how the bounds in the previous section simplify for these special cases. Examples include classical superchannels, superchannels with trivial side-channel, environment parametrized and side-channel parametrized superchannels.

\subsection{Classical superchannels}
In this section, we consider the problem of distinguishing classical superchannels. In the case of classical channels it was shown by Hayashi that adaptive strategies do not improve the discrimination error rate~\cite{Hayashi09}. We will see here that the same holds for the Stein's Lemma for classical superchannels. To the best of our knowledge this has not been investigated previously. 

From the definition of the superchannel we can take a classical superchannel as a pair of conditional probability distributions $\theta\equiv\{ \ce(a,s|c), \cd(d|b,s)\}$ that transform a channel $\cn(b|a)$ as 
\begin{align}
\cm(d|c) = \sum_{a,b,s} \cd(d|b,s)\cn(b|a)\ce(a,s|c). 
\end{align}
Since everything is classical, we are also free to make copies of every accessible system. Equivalently, we can give the experimenter access to the systems directly. That is, all variables except $s$ which is an internal variable of the superchannel. As a result either of the following two outputs is available to the experimenter: 
\begin{align}
q_1(d,b,r'a,c,r) :=&\sum_s \cd_{\theta_1}(d|b,s)\cn(b,r'|a,r)\ce_{\theta_1}(a,s|c)\cp(c,r), \\
q_2(d,b,r'a,c,r) :=&\sum_s \cd_{\theta_2}(d|b,s)\cn(b,r'|a,r)\ce_{\theta_2}(a,s|c)\cp(c,r).
\end{align}
Therefore, the relative entropy of two classical superchannels becomes
\begin{align}
D(\theta_1\|\theta_2) = \sup_{\cp,\cn} D( q_1(d,b,r'a,c,r) \| q_2(d,b,r'a,c,r) ).
\end{align}
We will now show that for classical channels the amortized superchannel relative entropies always collapse to the superchannel relative entropy. 
\begin{theorem}\label{thm:class-superchannel}
For two classical superchannels $\theta_1$ and $\theta_2$, we have
\begin{align}
D(\theta_1\|\theta_2) = D^{sA}(\theta_1\|\theta_2) = D^{A}(\theta_1\|\theta_2)= D^{A*}(\theta_1\|\theta_2)
\end{align}
and therefore the product strategy is optimal for asymptotic asymmetric superchannel discrimination for classical superchannels, in particular also adaptive strategies do not provide any advantage. 
\end{theorem}
\begin{proof}
By Lemma~\ref{lemma:ineq-DA} it suffices to show that $D^{A^*}(\theta_1\|\theta_2)\leq D(\theta_1\|\theta_2)$. In the following we abbreviate the variables for readability, but they should be clear from context. 
First, note that for two classical channels the amortized channel relative entropy always collapses to the channel relative entropy, therefore
\begin{align*}
D^{A^*}(\theta_1\|\theta_2) &= \sup_{ \substack{\cn,\cm \\ \bar\cn,\bar\cm}} \inf_{\cf} D^A( \cd_{\theta_1}\cn\ce_{\theta_1}\bar\cn \| \cd_{\theta_2}\cm\ce_{\theta_2}\bar\cm )  - D^A(\cn\cf\bar\cn \| \cm\cf\bar\cm) \\
&= \sup_{ \substack{\cn,\cm \\ \bar\cn,\bar\cm}}  \inf_{\cf} D( \cd_{\theta_1}\cn\ce_{\theta_1}\bar\cn \| \cd_{\theta_2}\cm\ce_{\theta_2}\bar\cm )  - D(\cn\cf\bar\cn \| \cm\cf\bar\cm) \\
&= \sup_{ \substack{\cn,\cm \\ \bar\cn,\bar\cm}}  \inf_{\cf} \sup_{\cp} D( \cd_{\theta_1}\cn\ce_{\theta_1}\bar\cn\cp \| \cd_{\theta_2}\cm\ce_{\theta_2}\bar\cm\cp )  - \sup_{\cp'} D(\cn\cf\bar\cn\cp' \| \cm\cf\bar\cm\cp') \\
&\leq \sup_{ \substack{\cn,\cm \\ \bar\cn,\bar\cm}} \sup_{\cp} D( \cd_{\theta_1}\cn\ce_{\theta_1}\bar\cn\cp \| \cd_{\theta_2}\cm\ce_{\theta_2}\bar\cm\cp )  -  D(\cn\ce_{\theta_1}\bar\cn\cp \| \cm\ce_{\theta_1}\bar\cm\cp) \\
&= \sup_{\cn, \bar\cn} D( \cd_{\theta_1}\cn\ce_{\theta_1}\bar\cn \| \cd_{\theta_2}\cn\ce_{\theta_2}\bar\cn ) \\
&\leq \sup_{\cn} D( \cd_{\theta_1}\cn\ce_{\theta_1}\| \cd_{\theta_2}\cn\ce_{\theta_2} ) \\
&= D(\theta_1\|\theta_2),
\end{align*}
where the third equality is by definition, the first inequality by bounding with a joint supremum and bounding the infimum by the concrete choice of $\ce_{\theta_1}$, the forth equality can be checked by direct calculation and the second inequality follows from data-processing. The statement of the lemma follows directly from here. 
\end{proof}
We conclude that simple product strategies are optimal when discrimination between classical superchannels. 
The case of classical channels was extended to classical-quantum channels in~\cite{BHKW,BHKW2} and it would be interesting to see if the same results hold for classical-quantum superchannels. 

\begin{figure}[t!]
\centering
\resizebox{\textwidth}{!}{
\begin{tikzpicture}

\node at (-2.1+7,1.3){a)};
\draw[] (0+7,0) rectangle (1+7,1);
\node at (0.5+7,0.5){$\cN$};
\draw[fill=SCcolor] (-1.5+7,1) rectangle (-0.5+7,-1.5) node[pos=0.5]{$\cE_\Theta$};
\draw[fill=SCcolor] (1.5+7,-1.5) rectangle (2.5+7,1) node[pos=0.5]{$\cD_\Theta$};
\draw (-0.5+7,0.5) -- (0+7,0.5) node[pos=0.5,sloped,above]{$A$};
\draw (1+7,0.5) -- (1.5+7,0.5) node[pos=0.5,sloped,above]{$B$};
\draw (-2.5+7,-0.25) -- (-1.5+7,-0.25) node[pos=0.5,sloped,above]{$C$};
\draw (2.5+7,-0.25) -- (3.5+7,-0.25) node[pos=0.5,sloped,above]{$D$};

\node at (-2.1+14,1.3){b)};
\draw[] (0+14,0) rectangle (1+14,1);
\node at (0.5+14,0.5){$\cN$};
\draw[fill=SCcolor] (-1.5+14,1) rectangle (-0.5+14,-1.5) node[pos=0.5]{$\cE$};
\draw[fill=SCcolor] (1.5+14,-2.5) rectangle (2.5+14,1) node[pos=0.5]{$\cD$};
\draw (-0.5+14,-1) -- (1.5+14,-1);
\draw (-0.5+14,0.5) -- (0+14,0.5);
\draw (1+14,0.5) -- (1.5+14,0.5);
\draw (-2.5+14,-0.25) -- (-1.5+14,-0.25);
\draw (2.5+14,-0.25) -- (3.5+14,-0.25);
\draw (-2+14,-1) -- (-1.5+14,-1) node[pos=0.5,sloped,above]{$E_1$};
\draw (-2+14,-2) -- (1.5+14,-2) node[pos=0.5,sloped,above]{$E_2$};
\draw (-2+14,-1) -- (-2.5+14,-1.5) -- (-2+14,-2);
\node at (-3+14,-1.5){$\omega^i_{E_1E_2}$};

\node at (-2.1+21,1.3){c)};
\draw[] (0+21,0) rectangle (1+21,1);
\node at (0.5+21,0.5){$\cN$};
\draw[fill=SCcolor] (-1.5+21,1) rectangle (-0.5+21,-1.5) node[pos=0.5]{$\cE$};
\draw[fill=SCcolor] (1.5+21,-1.5) rectangle (2.5+21,1) node[pos=0.5]{$\cD$};
\draw (-0.5+21,-1) -- (1.5+21,-1);
\draw (-0.5+21,0.5) -- (0+21,0.5);
\draw (1+21,0.5) -- (1.5+21,0.5);
\draw (-2.5+21,-0.25) -- (-1.5+21,-0.25);
\draw (2.5+21,-0.25) -- (3.5+21,-0.25);
\draw[fill=SCcolor] (0+21,-1.5) rectangle (1+21,-0.5) node[pos=0.5]{$\cS_i$};
\end{tikzpicture}}
\caption{\label{Fig:superExamples} Depiction of superchannels: a) A superchannel with trivial side-channel.  b) An environment parametrized superchannel. c) A side-channel parametrized superchannel.}
\end{figure}
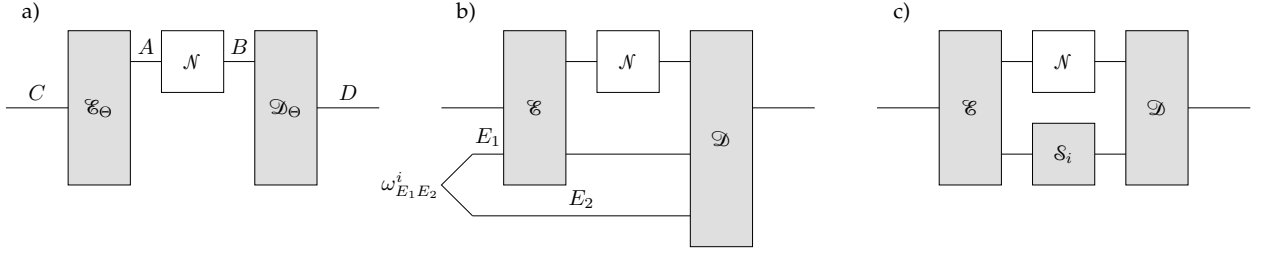

\subsection{The role of the side-channel}
This section is meant to give an intuition that the allowed communication between the two parts of the superchannel is the crucial difference to channel discrimination. That is, we will show that if the side-channel via $S$ is trivial in both superchannels, i.e.$|S|=1$, the discrimination problem drastically simplifies as it essentially reduces to discriminating two pairs of quantum channels, see Figure~\ref{Fig:superExamples}~a). 

First, we will give a meta-converse for any strategy that is different to the ones presented before. 
\begin{lemma}\label{Lem:stupidBoundisOK}
For two superchannels $\Theta_1$ and $\Theta_2$ we have for any strategy $\cS$,
\begin{align}
\mathbf{D}( p\| q) \leq n \mathbf{D}^{A}(\cE_{\Theta_1}\|\cE_{\Theta_2}) + n \mathbf{D}^{A}(\cD_{\Theta_1}\|\cD_{\Theta_2}).\label{Eq:stupidBoundisOK}
\end{align}
\end{lemma}
\begin{proof}
This is easily shown by using the channel amortization technique in~\cite{BHKW}, step-by-step removing either $\cE_{\Theta_i}$ or $\cD_{\Theta_i}$ depending on whats next in line. 
\end{proof}
This bound has some disadvantages in the superchannel case. In particular, it is dependent on the chosen decomposition of the superchannel. 
This is of course not desirable and we can easily see that in general the new bound in Equation~\eqref{Eq:stupidBoundisOK} is far from optimal because the bound does not take into account that the $S$ system is inaccessible: Consider two superchannels $\Theta_i$ for which the right hand side of Equation~\eqref{Eq:stupidBoundisOK} is finite and that can be decomposed into $\{ \cE_{\Theta_i}, \cD_{\Theta_i} \}$. Now, we construct the following channels,
\begin{align}
\hat\cE_{\Theta_i} = \cE_{\Theta_i} \otimes |i\rangle\langle i|_{S'} \\
\hat\cD_{\Theta_i} = \cD_{\Theta_i} \otimes \tr_{S'}.
\end{align}
It is now easy to see that the superchannels $\hat\Theta_i$ constructed from $\{\hat\cE_{\Theta_i}, \hat\cD_{\Theta_i}\}$ are equivalent to the respective $\Theta_i$ and should not be easier to distinguish. Nevertheless, even through $D^A(\cE_{\Theta_1} \| \cE_{\Theta_2})$ is assumed to be finite, $D^A(\hat\cE_{\Theta_1} \| \hat\cE_{\Theta_2})$ is infinite rendering the bound in Equation~\eqref{Eq:stupidBoundisOK} useless. This simple calculation also provides evidence that a good bound has to be based on the superchannel as a whole.

Now, if one only considers superchannels with trivial communication system $S$ we get that the above bound is indeed optimal for any sufficiently strong discrimination strategy. 
\begin{lemma}\label{lemma:trivialS}
For two superchannels $\Theta_1$ and $\Theta_2$ with trivial system $S$ we have in the Stein's setting,
\begin{align}
\zeta_{*}(\Theta_1,\Theta_2) = D^{A}(\cE_{\Theta_1}\|\cE_{\Theta_2}) + D^{A}(\cD_{\Theta_1}\|\cD_{\Theta_2}). 
\end{align}
where $*$ can be any strategy that is at least as powerful as the successive adaptive strategy.
\end{lemma}
\begin{proof}
The converse follows directly from Lemma~\ref{Lem:stupidBoundisOK}. The achievability follows since the successive adaptive strategy can, given a trivial system $S$, be separated into two adaptive channel discrimination strategies for the respective channel pairs as follows. As initial input state we choose a product state $\rho_{CR}\otimes\rho_{BR'}$ and the adaptive operations within the superchannel are simply used to swap the system into place. The adaptive operations in between superchannels are used for arbitrary operations in product form $\cA_1\otimes\cA_2$, leading to 
\begin{align*}
(\cA_1\otimes\cA_2)(\cE_{\Theta}\otimes\cD_{\Theta} )(\rho_{CR}\otimes\rho_{BR'}),
\end{align*}
which gives the beginning of the tensor product of two adaptive channel discrimination strategies. Iterating swap operations and adaptively chosen product channels gives the full strategy. 
The results follows then from the Stein's Lemma for quantum channels. 
\end{proof}
One should however remember that this is usually not true when $S$ is not trivial as seen by the example above. As a remark, the same also holds for general $k$-combs and in terms of channel discrimination means that if an experimenter has access to $n$ uses of each channel in the  pair the optimal discrimination rate is the sum of the individual discrimination rates. 

\subsection{Replacer superchannels}\label{sec:replacer}

For channels we have replacer channels which essentially reduce channel problems to state problems. We can similarly define replacer superchannels which should also reduce the problem by one level. A replacer superchannel $\Theta_i$ takes any channel $\cN$ and outputs a fixed channel $\cN_i$, i.e. $\Theta_i(\cN) = \cN_i$. These replacer channels are simply constructed by taking the input state, applying $\cN_i$ and outputing the result, while at the same time feeding a dummy state into $\cN$ and immediately tracing out the result. Here, it is important that all replacer superchannels use the same dummy state, not to add any additional distinguishability, and we will therefore usually choose it to be the maximally mixed state. For a depiction of a replacer superchannel see Figure~\ref{Fig:super}. Note that replacer channels are a special case of the side-channel seizable channels discussed in the next section, but we state the following lemma for completeness.
\begin{lemma}
For two replacer superchannels we have,
\begin{align}
D(\Theta_1 \| \Theta_1 ) &= D(\cN_1 \| \cN_2 ), \\
D^{sA}(\Theta_1 \| \Theta_2 ) = D^{A}(\Theta_1 \| \Theta_2 )= D^{A^*}(\Theta_1 \| \Theta_2 ) &= D^A(\cN_1 \| \cN_2 ).
\end{align}
\end{lemma}
\begin{proof}
The upper bounds $D(\Theta_1 \| \Theta_1 ) \leq D(\cN_1 \| \cN_2 )$ and $D^{A^*}(\Theta_1 \| \Theta_2 ) \leq D^A(\cN_1 \| \cN_2 )$ are a special case of those for side-channel parametrized superchannels shown in Lemma~\ref{lemma:SidePara}. 

We can now furthermore see that $D(\Theta_1 \| \Theta_1 ) \geq D(\cN_1 \| \cN_2 )$ and $D^{sA}(\Theta_1 \| \Theta_2 ) \geq D^A(\cN_1 \| \cN_2 )$ by choosing all channels in the suprema in the quantities as identity channels. The inequalities then follow directly from the structure of the replacer channels which concludes the proof. 
\end{proof}
It follows that two replacer superchannels can be discriminated with the same rate as the fixed channels they output. For any strategy that's at least as powerful as the successive adaptive strategy the rate is given by the amortized channel relative entropy which is the optimal channel discrimination rate. Product strategies are naturally not sufficient to reach that rate and result in a discrimination rate according to the channel relative entropy.

\subsection{Environment and side-channel parametrized superchannels} 

An often considered class of channels are the so called environment parametrized channels. Two channels $\cN_1$, $\cN_2$ are called (jointly) environment parametrized if they differ only in the choice of an auxiliary state on an additional environment system, i.e. 
\begin{align}
\cN_i (\rho_A) = \cP_{AE\rightarrow B} (\rho_A\otimes\omega^i_E), 
\end{align} 
where $\omega^i_E$ is the parameterizing state. 

A natural extension of the concept is to define an environment parametrized superchannel as one with both defining channels being environment parametrized by a joint environment state $\omega^i_{E_1E_2}$, see Figure~\ref{Fig:superExamples}~b). Interestingly it is unclear how to give a good upper bound on $\bD^{A^*}$ for these channels, nevertheless we can give convenient converse bounds via the following meta-converse. Resolving this tension is directly related to the achievability of $D^{A^*}$. 
\begin{lemma}\label{lemma:EnvPara}
For two jointly environment parametrized superchannels $\Theta_1$ and $\Theta_2$ we have for any strategy $\cS$ and sub-additive divergence $\bD$,
\begin{align}
\mathbf{D}^{}(p\| q) \leq n \mathbf{D}^{}(\omega^1_{E_1E_2}\|\omega^2_{E_1E_2}).
\end{align}
\end{lemma}
\begin{proof}
For the proof we simply note that since the superchannels are, besides the parameterizing environment state, the same, we can regard any discrimination strategy as a source preparing either $(\omega^1_{E_1E_2})^{\otimes n}$ or $(\omega^2_{E_1E_2})^{\otimes n}$ followed by a fixed identical channel. Removing that channel via data-processing and using sub-additivity leads to the desired bound. 
\end{proof}
Similar to the channel case~\cite{BHKW}, we call two jointly environment parametrized channels environment seizable if there additionally exists a superchannel $\Psi$ such that
\begin{align}
\Psi(\Theta_i) = \cR_{\omega^i_{E_1E_2}}, 
\end{align}
where $\cR_{\omega^i_{E_1E_2}}$ is the replacer channel that always outputs $\omega^i_{E_1E_2}$. 
In this case we can easily get achievability results, e.g. for the Stein's Lemma setting.
\begin{lemma}\label{lemma:EnvSeizable}
For two jointly environment seizable channels $\Theta_1$ and $\Theta_2$ we have
\begin{align}
\zeta_{*}(\Theta_1,\Theta_2) = D(\omega^1_{E_1E_2}\| \omega^2_{E_1E_2}). 
\end{align}
where $*$ can be any strategy that is at least as powerful as the product strategy, which includes every strategy discussed in the work.
\end{lemma}
\begin{proof}
The proof follows easily, as an achieving strategy is simply to apply the superchannel $\Psi$ to every $\Theta_i$ and then discriminate the output states. 
\end{proof}
This reduces the discrimination of jointly environment-seizable superchannels with product, any parallel or adaptive strategies to that of quantum states. It also means that none of these strategies is more powerful that the simple product strategy.

Now, we discuss a different class of superchannels with a similar intuition: superchannels parametrized by a particular choice of channel. 
\begin{definition}
We call two superchannels $\Theta_1$ and $\Theta_2$ jointly side-channel parametrized if their action can be decomposed as
\begin{align}
\Theta_i(\cN) = \cD_{BS\rightarrow D}\circ(\cN_{A\rightarrow B}\otimes\cS_{i, S\rightarrow S})\circ\cE_{C\rightarrow AS}, \label{Eq:SideChannel}
\end{align}
for $i\in\{1,2\}$. 
\end{definition}
For a visualization see Figure~\ref{Fig:superExamples}~c). 
Note that the previously discussed replacer superchannels are a special case of side-channel parametrized superchannels. 
We now get the following bound.
\begin{lemma}\label{lemma:SidePara}
For two jointly side-channel parametrized superchannels $\Theta_1$ and $\Theta_2$ we have
\begin{align}
D^{}(\Theta_1\|\Theta_2)  &\leq D(\cS_1\| \cS_2), \\
D^{sA}(\Theta_1\|\Theta_2) \leq D^{A}(\Theta_1\|\Theta_2) \leq D^{A^*}(\Theta_1\|\Theta_2) &\leq D^A(\cS_1\| \cS_2).
\end{align}
\end{lemma}
\begin{proof}
The first line follows easily by using data-processing. 
In the second line, the first two inequalities follow by Lemma~\ref{lemma:ineq-DA}. 
For the third, let $\rho^*_{CR}$, $\sigma^*_{CR}$ be the optimizing states in $D^A(\Theta_1(\cN)\circ\bar\cN \| (\Theta_2)(\cM)\circ\bar\cM )$ for the chosen superchannels. We then observe, 
\begin{align}
&D^{A^*}(\Theta_1 \| \Theta_2 ) \nonumber\\
&\leq \sup_{ \substack{\cN,\cM \\ \bar\cN,\bar\cM}} D^A(\cD\circ(\cN\otimes\cS_{1})\circ\cE\circ\bar\cN \| \cD\circ(\cM\otimes\cS_{2})\circ\cE\circ\bar\cM ) 
 -  D^A(\cN\circ\cE\circ\bar\cN \| \cM\circ\cE\circ\bar\cM) \nonumber\\
&\leq \sup_{ \substack{\cN,\cM \\ \bar\cN,\bar\cM}} D(\cD\circ(\cN\otimes\cS_{1})\circ\cE\circ\bar\cN(\rho^*) \| \cD\circ(\cM\otimes\cS_{2})\circ\cE\circ\bar\cM(\sigma^*) )  -  D(\cN\circ\cE\circ\bar\cN(\rho^*) \| \cM\circ\cE\circ\bar\cM(\sigma^*)) \nonumber\\
&\leq \sup_{ \substack{\cN,\cM \\ \bar\cN,\bar\cM}} D((\cN\otimes\cS_{1})\circ\cE\circ\bar\cN(\rho^*) \| (\cM\otimes\cS_{2})\circ\cE\circ\bar\cM(\sigma^*) )  -  D(\cN\circ\cE\circ\bar\cN(\rho^*) \| \cM\circ\cE\circ\bar\cM(\sigma^*)) \nonumber\\
&\leq D^A(\cS_1\| \cS_2),
\end{align}
where the first inequality is by bounding the infimum by choosing $\cE$, the second by how we chose the states, the third uses data-processing and the forth the chain rule from~\cite{fang2019chain} separating the side-channels from the rest. 
\end{proof}

Similar to the environment seizable channels defined in~\cite{BHKW}, we can identify a restricted class for which the above inequalities are tight.
\begin{definition}
We call two superchannels $\Theta_1$ and $\Theta_2$ side-channel seizable if they are environment parametrized and additionally there exists a superchannel $\Psi$ such that
\begin{align}
\Psi(\Theta_i(\cN)) = \cS_i,
\end{align}
for $i\in\{1,2\}$ and some channel $\cN$. 
\end{definition}
For these superchannels we have the following.
\begin{lemma}\label{lemma:SideSeizable}
For two jointly side-channel seizable channels $\Theta_1$ and $\Theta_2$ we have
\begin{align}
D^{}(\Theta_1\|\Theta_2)  &= D(\cS_1\| \cS_2) \\
D^{sA}(\Theta_1\|\Theta_2) = D^{A}(\Theta_1\|\Theta_2) = D^{A^*}(\Theta_1\|\Theta_2) &= D^A(\cS_1\| \cS_2)
\end{align}
\end{lemma}
\begin{proof}
The proof follows easily because both the channel relative entropy and the amortized channel relative entropy are monotone under superchannels. 
\end{proof}
We can conclude that for side-channel seizable channels the task of asymptotic asymmetric discrimination reduces to channel discrimination with same rates as discriminating the parameterizing channels $\cS_1$ and $\cS_2$, which is a direct generaliztion of the earlier result for replacer channels.  

Finally it should be remarked that the decomposition in Equation~\ref{Eq:SideChannel} is not unique, similar to the discussion in Section~\ref{SSec:Networks}. In particular, for any such superchannel one can always find a different decomposition with the same structure where the channels $\cS_i$ are isometries which are perfectly discriminable by a finite number of rounds, i.e. any converse bound via Lemma~\ref{lemma:SidePara} would become infinite. Therefore one has to find the right decomposition in order to get the best converse bound. In the case of side-channel seizable superchannels that decomposition follows from the construction and is given by one that matches the achievable rate.

\section{Beyond Steins Lemma}\label{sec:BeyondStein}

So far we have mostly considered the asymptotic asymmetric discrimination setting, a.k.a. Steins setting. However, since we have stated many previous results in terms of generalized divergences they also allow us to state bounds on other settings, specifically the strong converse exponent and the symmetric setting. The results below can mostly be found by applying the meta-converses from the previous sections to these scenarios using techniques that can be found e.g. in~\cite{BHKW}. We therefore keep the discussion short. Also note that the examples in Section~\ref{sec:examples} often only use few properties of the relative entropy, e.g. its chain rule, and hence many of them can easily be transfered to the settings described in this section if the used quantities fulfill the same properties. 

\subsection{Symmetric Discrimination --- Chernoff's bound} 

Symmetric hypothesis testing describes an alternative scenario in which we aim to simultaneously minimize the two possible errors. 
This is sometimes also described as the Bayesian setting of hypothesis testing. Given an {\it a priori} probability $p\in(0,1)$ that the first superchannel $\Theta_1$ is selected, the non-asymptotic symmetric error exponent is defined as
\begin{align}
\xi^n(p,\Theta_1,\Theta_2)\coloneqq\sup_{\cS}-\frac{1}{n}\log\Big(p\cdot \alpha_n(\cS)+(1-p)\beta_n(\cS)\Big).
\end{align}
Given that the expression above involves an optimization over all final measurements $\cQ$, we can employ a well known result relating optimal error probability to trace distance, see also Section~\ref{sec:one-shot}, to conclude that
\begin{equation}
\xi^n(p,\Theta_1,\Theta_2)=\sup_{\cS}-\frac{1}{n}\log\left(\frac{1}{2}\left(1- \left\Vert p \rho- (1-p)\tau\right\Vert_1\right)\right),
\end{equation}
where $\rho$ and $\tau$ are the possible output states of the chosen strategy depending on whether the superchannel was $\Theta_1$ or $\Theta_2$, respectively. 
We are then interested in the asymptotic symmetric error exponent
\begin{align}
\underline{\xi}(\Theta_1,\Theta_2)\coloneqq\liminf_{n \to \infty} \xi_n(p,\Theta_1,\Theta_2), \\
\overline{\xi}(\Theta_1,\Theta_2)\coloneqq\limsup_{n \to \infty} \xi_n(p,\Theta_1,\Theta_2).
\end{align}

Thanks to the meta-converses in Section~\ref{sec:asymDisc}, we can easily get converse bounds on the superchannel Chernoff bound as well. 
\begin{theorem}\label{thm:convese-Chernoff}
For two superchannels $\Theta_1$ and $\Theta_2$, we have depending on the choice of strategy $\cS$, the following bounds:
\begin{align}
\xi_{p}^n(p,\Theta_1,\Theta_2) &\leq -\frac1n \log[p(1-p)] + \tilde D_{1/2}(\Theta_1 \| \Theta_2 ),  \\
\xi_{sp}^n(p,\Theta_1,\Theta_2) &\leq -\frac1n \log[p(1-p)] + \tilde D^{s\infty}_{1/2}(\Theta_1 \| \Theta_2 ), \\
\xi_{cp}^n(p,\Theta_1,\Theta_2) &\leq -\frac1n \log[p(1-p)] + \tilde D^{c\infty}_{1/2}(\Theta_1 \| \Theta_2 ), \\
\xi_{fp}^n(p,\Theta_1,\Theta_2) &\leq -\frac1n \log[p(1-p)] + \tilde D^{\infty}_{1/2}(\Theta_1 \| \Theta_2 ), \\
\xi_{sa}^n(p,\Theta_1,\Theta_2) &\leq -\frac1n \log[p(1-p)] + \tilde D^{sA}_{1/2}(\Theta_1 \| \Theta_2 ),  \\
\xi_{na}^n(p,\Theta_1,\Theta_2) &\leq -\frac1n \log[p(1-p)] + \tilde D^{A}_{1/2}(\Theta_1 \| \Theta_2 ),  \\
\xi_{fg}^n(p,\Theta_1,\Theta_2) &\leq -\frac1n \log[p(1-p)] + \tilde D^{A^*}_{1/2}(\Theta_1 \| \Theta_2 ), 
\end{align}
\end{theorem} 
\begin{proof}
The results follow all similarly by first taking the following inequality from~\cite{BHKW}, 
\begin{align}
-\log\left( \frac12 (1-\| p\rho - (1-p)\sigma\|_1)\right) \leq -\log(p(1-p)) + \tilde D_{1/2}(\rho\|\sigma), 
\end{align}
and then either using the meta-converse for the corresponding strategy or, for the regularized quantities, simply optimizing over all possible channels and input states.
\end{proof}
In the case of the product strategy we can furthermore obtain the asymptotic result 
\begin{align}
\xi_{p}(\Theta_1,\Theta_2) = C(\Theta_1,\Theta_2),
\end{align}
by reduction to a state discrimination problem as previously described for the Stein's setting in Section~\ref{SSec:Prod}. From this we get for all the other strategies asymptotic statements of the following form:
\begin{align}
C(\Theta_1,\Theta_2) \leq \underline{\xi}_{fg}(\Theta_1,\Theta_2) \leq \overline{\xi}_{fg}(\Theta_1,\Theta_2) \leq \tilde D^{A^*}_{1/2}(\Theta_1 \| \Theta_2 )
\end{align}
in the case of fully general strategies and similarly for all the others. Note that also in the channel setting no tighter bounds are known.

\subsection{Strong converse exponent -- Han-Kobayashi}\label{sec:han-kob}

The strong converse exponent is a refinement of the asymmetric hypothesis testing quantity discussed above. For $r>0$, we are interested in characterizing the non-asymptotic quantity
\begin{align}
H^n(r,\Theta_1,\Theta_2)\coloneqq\inf_{\cS}\left\{-\frac{1}{n}\log(1-\alpha_n(\cS))\middle|\beta_n(\cS)\leq2^{-rn}\right\},
\end{align}
as well as the asymptotic quantities
\begin{equation}
\underline{H}(r,\Theta_1,\Theta_2)\coloneqq\liminf_{n \to \infty} H^n(r,\Theta_1,\Theta_2), \qquad \overline{H}(r,\Theta_1,\Theta_2)\coloneqq\limsup_{n \to \infty} H^n(r,\Theta_1,\Theta_2).
\end{equation}
The interpretation is that the type~II error probability is constrained to tend to zero exponentially fast at a rate $r > 0$, but then if $r$ is too large, the type~I error probability will necessarily tend to one exponentially fast, and we are interested in the exact rate of exponential convergence. Note that this strong converse exponent is only non-trivial if $r$ is sufficiently large.

Based on the results in Section~\ref{sec:asymDisc} and the techniques from~\cite{BHKW} we get the following result.
\begin{theorem}\label{Thm:StrongConverseExp}
For two superchannels $\Theta_1$ and $\Theta_2$, we have
\begin{align}
H_{p}^n(r,\Theta_1,\Theta_2) &\geq \sup_{\alpha>1}\frac{\alpha-1}{\alpha} \left( r -  \tilde D_{\alpha}(\Theta_1 \| \Theta_2 ) \right),  \\
H_{sp}^n(r,\Theta_1,\Theta_2) &\geq\sup_{\alpha>1}\frac{\alpha-1}{\alpha} \left( r -  \tilde D^{s\infty}_{\alpha}(\Theta_1 \| \Theta_2 ) \right), \\
H_{cp}^n(r,\Theta_1,\Theta_2) &\geq \sup_{\alpha>1}\frac{\alpha-1}{\alpha} \left( r -  \tilde D^{c\infty}_{\alpha}(\Theta_1 \| \Theta_2 ) \right), \\
H_{fp}^n(r,\Theta_1,\Theta_2) &\geq \sup_{\alpha>1}\frac{\alpha-1}{\alpha} \left( r -  \tilde D^{\infty}_{\alpha}(\Theta_1 \| \Theta_2 ) \right), \\
H_{sa}^n(r,\Theta_1,\Theta_2) &\geq \sup_{\alpha>1}\frac{\alpha-1}{\alpha} \left( r -  \tilde D^{sA}_{\alpha}(\Theta_1 \| \Theta_2 ) \right),  \\
H_{na}^n(r,\Theta_1,\Theta_2) &\geq \sup_{\alpha>1}\frac{\alpha-1}{\alpha} \left( r -  \tilde D^{A}_{\alpha}(\Theta_1 \| \Theta_2 ) \right),  \\
H_{fg}^n(r,\Theta_1,\Theta_2) &\geq \sup_{\alpha>1}\frac{\alpha-1}{\alpha} \left( r -  \tilde D^{A^*}_{\alpha}(\Theta_1 \| \Theta_2 ) \right), 
\end{align}
\end{theorem} 
\begin{proof}
The results follow all similarly by first taking the following inequality from~\cite{BHKW}, 
\begin{align}
\frac{\alpha}{\alpha-1}\log(1-\alpha_n(\cS)) +n r \leq \tilde D_{\alpha}(\rho\|\sigma), 
\end{align}
and then either using the meta-converse of the corresponding strategy or, for the regularized quantities, simply optimizing over all possible channels and input states.
\end{proof}
For more on the state and channel case of the strong converse exponent we refer to~\cite{mosonyi2015quantum,fawzi2020defining}.

\begin{remark}
A different refinement is the error exponent, or Hoeffdings bound, in the sense that the type~II error probability is constrained to decrease exponentially with exponent $r>0$. We are then interested in characterizing the  error exponent of the type~I error probability under this constraint. That is, we are interested in characterizing the non-asymptotic quantity
\begin{align}
B_n(r,\Theta_1,\Theta_2)\coloneqq\sup_{\cS}\left\{-\frac{1}{n}\log\alpha_n(\cS)\middle|\beta_n(\cS)\leq2^{-rn}\right\}.
\end{align}
Note that this error exponent is non-trivial only if $r$ is not too large. Already in the channel case this scenario is much more difficult to handle and bounds are only known in some special cases. We leave the investigation of this case for quantum superchannels and quantum networks for future research. For the state case we refer to~\cite{mosonyi2011quantum} and the channel case to~\cite{BHKW}. 
\end{remark}

\section{Quantum Networks}\label{sec:networks}
We have so far focused on superchannels, which are instances of networks with exactly one access point. In this section, we generalize the results to general quantum networks. Most of the tools used in this section are generalizations of the superchannel case and we therefore keep the presentation short. First we have to define the generalized divergences for networks. We start with the rather simple case of the generalized quantum network divergence. For two quantum networks, represented by $k$-combs $\Theta_1^k$ and $\Theta_2^k$, we have~\cite[Definition 1]{wang2019resource}
\begin{align}
\bD(\Theta_1^k \| \Theta_2^k) = \sup_{(\cA_i)_{k-1}, \rho} \bD( \Theta_1^k((\cA_i)_{k-1})(\rho) \| \Theta_2^k((\cA_i)_{k-1})(\rho) ), 
\end{align}
where $(\cA_i)_{k-1} = (\cA_1, \dots , \cA_{k-1})$ are the $k-1$ channels the $k$-comb acts on. 
From here a regularized divergence is defined in the usual way,
\begin{align}
\bD^\infty(\Theta_1^k \| \Theta_2^k) = \lim_{n\rightarrow\infty} \frac1n \bD((\Theta_1^k)^{\otimes n} \| (\Theta_2^k)^{\otimes n}) 
\end{align}
And lastly, we need to define the amortized quantum network divergences,
\begin{align}
\bD^A(\Theta_1^k \| \Theta_2^k) = \sup_{\substack{(\cA_i)_{k-1}, \\ (\bar \cA_i)_{k-1}}} \inf_{\Omega_1^{k-1}} &\bD^{A}(\Theta_1^k( (\cA_i)_{k-1}) \| \Theta_2^k( (\bar \cA_i)_{k-1}) )  \nonumber\\ 
&-  \bD^{A}( \cA_{k-1}\circ \Omega_1^{k-1}((\cA_i)_{k-2}) \| \bar \cA_{k-1}\circ\Omega_1^{k-1}((\bar \cA_i)_{k-2})),  \\
\bD^{A^*}(\Theta_1^k \| \Theta_2^k) = \sup_{\substack{(\cA_i)_{k-1}, \\ (\bar \cA_i)_{k-1}, \\ \cA_0,\bar\cA_0}} \inf_{\Omega_1^{k-1}} &\bD^{A}(\Theta_1^k( (\cA_i)_{k-1}) \circ\cA_0\| \Theta_2^k( (\bar \cA_i)_{k-1})\circ\bar\cA_0 )  \nonumber\\
&-  \bD^{A}( \cA_{k-1}\circ\Omega_1^{k-1}((\cA_i)_{k-2})\circ\cA_0 \| \bar \cA_{k-1}\circ\Omega_1^{k-1}((\bar \cA_i)_{k-2})\circ\bar\cA_0), 
\end{align}
which generalize the amortized superchannel divergence and the fully-amortized superchannel divergence. 
Here $\Omega_1^{k-1}$ is a $k-1$-comb that takes the role of the channel $\cF$ in the superchannel case. 
Note that just as in the case of superchannels, one can define several intermediate versions of the regularized and amortized quantum network divergences, e.g. one where only the input state or certain channels are amortized. The definitions follow similarly as in the superchannel case and we omit them here for brevity. 

Of course one now needs to classify the possible discrimination strategies in order to investigate their error performance. The previous discussion on superchannel discrimination should convince the reader that the number of possible strategies, with potentially different ability to discriminate, is too large to discuss them here one-by-one. We will instead limit the discussion to a few notable strategies, however remark that all in some way more limited strategies can as well be investigated similarly considering the suitable extensions of the superchannel setting. 

The first example is again the simple product strategy. By reduction to the state discrimination case one easily gets, 
\begin{align}
\zeta_{p}(\Theta^k_1,\Theta^k_2) &=  D(\Theta_1^k\|\Theta_2^k), \\
\xi_{p}(\Theta^k_1,\Theta^k_2) &=  C(\Theta_1^k\|\Theta_2^k).  
\end{align}
Furthermore, one can check that the proof strategy for parallel discrimination of superchannels in Section~\ref{Sec:ParMPC} also extends to the general network case. This leads to the following result in the Stein's setting for fully parallel strategies, 
\begin{align}
\zeta_{fp}(\Theta^k_1,\Theta^k_2) &=  D^\infty(\Theta_1^k\|\Theta_2^k).
\end{align}

The main result in this section is the following generalization of Theorem~\ref{thm:mc-general} to general quantum networks. 
\begin{theorem}[Meta-converse for arbitrary strategies]\label{thm:mc-qn-general}
For two $k$-combs $\Theta_1^k$ and $\Theta_2^k$ and any adaptive network discrimination strategy we have
\begin{align}
\mathbf{D}(p\Vert q) &\leq n \mathbf{D}^{A^*}(\Theta_1^k\|\Theta_2^k). \label{Eq:sec-ineq-gn}
\end{align}
\end{theorem}
\begin{proof}
The proof is a direct extension of that of Theorem~\ref{thm:mc-general} considering $k$-combs while always assuming that when removing the currently last comb all access point might act on channel sequences which include components of previous combs. Iteratively removing all the combs leads to the desired result. 
\end{proof}
This covers the most general setting and allows us to extend the discrimination results from the superchannels case to networks, where we get fundamental converse bounds on the asymptotic distinguishability of two quantum networks. 
\begin{theorem}
For two $k$-combs $\Theta_1^k$ and $\Theta_2^k$ and any adaptive network discrimination strategy we have in the Stein's setting
\begin{align}
\zeta^n_{fg}(\epsilon,\Theta_1^k,\Theta_2^k)   \leq  \frac{1}{1-\eps} \left(  D^{A^*}(\Theta_1^k\|\Theta_2^k) +\frac1n h_2(\varepsilon)\right),
\end{align}
in the symmetric Chernoff setting,
\begin{align}
\xi_{fg}^n(p,\Theta_1^k,\Theta_2^k) &\leq -\frac1n \log[p(1-p)] + \tilde D^{A^*}_{1/2}(\Theta_1^k \| \Theta_2^k ).
\end{align}
and for the strong converse exponent,
\begin{align}
H_{fg}^n(r,\Theta_1^k,\Theta_2^k) &\geq \sup_{\alpha>1}\frac{\alpha-1}{\alpha} \left( r -  \tilde D^{A^*}_{\alpha}(\Theta_1^k \| \Theta_2^k ) \right).
\end{align}
\end{theorem}
\begin{proof}
This follows analogously to Equation~\eqref{Eq:fg-converse-Stein}, Theorem~\ref{thm:convese-Chernoff} and Theorem~\ref{Thm:StrongConverseExp}, using the meta-converse in Theorem~\ref{thm:mc-qn-general} applied to the appropriate generalized divergence. 
\end{proof}

Combining the above results implies that for the best possible (fully general) strategy, we have
\begin{align}
 D^\infty(\Theta_1^k\|\Theta_2^k) &\leq  \zeta_{fg}(\Theta^k_1,\Theta^k_2) \leq   D^{A^*}(\Theta_1^k\|\Theta_2^k) ,\\
  C(\Theta_1^k\|\Theta_2^k) &\leq  \xi_{fg}(\Theta^k_1,\Theta^k_2) \leq   \tilde D^{A^*}_{1/2}(\Theta_1^k\|\Theta_2^k) .
\end{align}

Of course the crucial question is now for which channels those inequalities become equalities. A simple example we can discuss here is that of classical networks for which the amortized network relative entropy always collapses. 
\begin{lemma}
For two classical networks $\theta_1^k$ and $\theta_2^k$, we have
\begin{align}
D(\theta_1^k\|\theta_2^k) =  D^{A^*}(\theta_1^k\|\theta_2^k).
\end{align}
\end{lemma}
\begin{proof}
The proof is similar to that of Theorem~\ref{thm:class-superchannel}. 
\end{proof}
It follows that for two classical networks product strategies are optimal for discrimination in the Stein's setting. 
The bounds presented here give fundamental limits on the ability to discriminate between two quantum networks, however, it is also clear that many important questions in the general setting are still unanswered and invite further research.

\section{Applications }\label{sec:applications}
Quantum channel discrimination has found many applications~\cite{Hayashi09,Cooney2016,L08}. Often it can make sense to generalize these applications to networks, for example when we want to allow for additional interaction in a given protocol. In the following we will discuss a generalization of the well known quantum illumination problem in which we allow for a relay station on the object we wish to detect that can aid by altering the signal it receives. 

\subsection{Active quantum illumination}

Quantum illumination~\cite{L08} describes the task of detecting the presence or absence of a certain object making use of quantum tools to enhance detection. There are usually two possible scenarios to consider, both start with shining a light source at the suspected location of an object. In the first case, one assumes that either the object is missing and the light can pass unhindered or the object is present and the light is blocked. Let's say the default noise the light probe $\rho$ is experiencing is given by a channel $\cN$, then the receiver behind the object will either receive $\cN(\rho)$ if the probe passes or the fixed state $\tau$ if the probe is blocked. As $\rho$ can be subject to optimization, the resulting task is equivalent to discriminating the channels $\cN$ and $\cR_\tau$, with the latter being the associated replacer channel. In the second case, the position of the detector is changed such that a missing object is represented by receiving $\tau$ while the presence of the object gives the reflection of the probe, i.e. $\cN(\rho)$, therefore the role of null and alternative hypothesis are interchanged. 

In~\cite{Cooney2016} it was shown that at least in the first case, adaptive strategies do not give any advantage in many scenarios, namely when one is interested in the error rates according to Stein's Lemma or the strong converse scenario. See also~\cite{BHKW} for a proof based on amortized channel divergences. 

In the quantum illumination task described above, the object is considered passive in the sense that it can only absorb or alternatively reflect the object. Here we propose active quantum illumination in which we allow the object to actively aid the discrimination task by receiving the probe state, acting on it and sending it on. As typically one might use quantum illumination to detect an object in space, we might already have a relay station on that object capable of using quantum optics to aid its observation. In full generality, the task now becomes to discriminate between the superchannel output $\Theta_1(\cN)(\rho)$ and a fixed state $\tau$, i.e. the replacer superchannel $\Theta_{\cR_\tau}$ that always outputs the replacer channel $\cR_\tau$. 

While we leave the general investigation of the problem to future work, we show here that in the case of the superchannel Stein's setting the simple product strategy is optimal, giving an extension of the results in~\cite{Cooney2016}. 
\begin{lemma}\label{Lem:Illum}
Let $\Theta_1$ be an arbitrary superchannel and $\Theta_{\cR_\tau}$ the replacer superchannel that always outputs the replacer channel $\cR_\tau$. We have, 
\begin{align}
\zeta_p(\Theta_1, \Theta_{\cR_\tau}) = \zeta_{fg}(\Theta_1, \Theta_{\cR_\tau}) = D(\Theta_1 \| \Theta_{\cR_\tau} ). 
\end{align}
\end{lemma}
\begin{proof}
Based on the results from the previous sections the only thing left to show is that 
\begin{align}
D^{A^*}(\Theta_1 \| \Theta_{\cR_\tau} ) \leq D(\Theta_1 \| \Theta_{\cR_\tau} ). 
\end{align}
We will do so here,
\begin{align*}
D^{A^*}(\Theta_1 \| \Theta_{\cR_\tau} ) &\leq \sup_{\cN,\cM,\rho,\sigma} D(\Theta_1(\cN)(\rho) \| \Theta_{\cR_\tau}(\cM)(\sigma)) - D(\cN\circ\cE_1(\rho) \| \cM\circ\cE_1(\sigma)) \\
&= \sup_{\cN,\cM,\rho,\sigma} D(\Theta_1(\cN)(\rho) \| \tau\otimes\cM(\sigma)_R) - D(\cN\circ\cE_1(\rho) \| \cM\circ\cE_1(\sigma))  \\
&\leq \sup_{\cN,\cM,\rho,\sigma} D(\Theta_1(\cN)(\rho) \| \tau\otimes\cM(\sigma)_R) - D(\cN(\rho)_R \| \cM(\sigma)_R)  \\
&= \sup_{\cN,\rho} D(\Theta_1(\cN)(\rho) \| \tau\otimes\cN(\rho)_R) \\ 
&= D(\Theta_1 \| \Theta_{\cR_\tau} ),
\end{align*}
where the first inequality follows from Lemma~\ref{Lem:AleqA} which we will prove in Appendix~\ref{Ap:amortize}, the first equality follows by definition of $\Theta_{\cR_\tau}$, the second inequality is data-processing. The subscript $R$ denotes restriction on the $R$ system, i.e. the result of tracing out all other systems. For the second equality notice that $\cN(\rho)_R=\tr_D \Theta_1(\cN)(\rho)$ and by direct calculation one can confirm the rule
\begin{align*}
D(\rho_{AR} \| \sigma_A \otimes \sigma_R) - D(\rho_R \| \sigma_R) = D(\rho_{AR} \| \sigma_A \otimes \rho_R),
\end{align*}
using $\log (A\otimes B) = \log A \otimes \Id + \Id \otimes \log B$, see also~\cite[Lemma 38]{BHKW}. The final equality is by definition. This concludes the proof.
\end{proof}

\section{Conclusions}\label{sec:conclusions}

In this manuscript we have extended the framework of quantum hypothesis testing to the setting of discriminating two quantum networks when many uses are available. The additional structure provided by networks makes this a rich and enticing problem and we have discussed a host of potential discrimination strategies. For each strategy we provided a converse bound, giving a fundamental limit to their performance. In most cases we also discussed achievability of those converse bounds, determining the optimal asymptotic error behavior. 

Nevertheless, many open problems remain. First and foremost, the achievability for the fully general class of strategies. Most interestingly, could there be some class of strategies that can outperform the nested adaptive strategies? It seems that this is indeed a possibility, based on our converse bounds. That would imply that there are particular strategies that outperform parallel strategies which would be a departure from the results in the channel discrimination setting. It is however also entirely possible that one can find a sharper converse bound that improves on the result presented in this work.

Secondly, proving the asymptotic equipartition property for the smooth max-relative entropy as discussed in Appendix~\ref{Ap:chainRules}. This would allow us to to simplify the relationship between several classes of strategies, show directly that our bound for nested adaptive strategies in the Stein's setting is achievable and therefore optimal, and also be valuable in its own right. 

Additionally, there are a lot of other open problems, for examples a better understanding of the symmetric hypothesis testing setting, even for the case of channels. In particular in the network setting, it would also be interesting to look at more limited classes of strategies, e.g. where different parts of the networks can only be acted on by locally separated parties, which is commonly known as distributed hypothesis testing. In the opposite direction, one could also look at more general settings including non-causal strategies and explore whether the tools developed in this work can also apply there. It has previously been shown that non-causal strategies can outperform causal ones in certain settings~\cite{bavaresco2021strict}.  Finally, considering the numerous applications of state and channel discrimination it should be worthwhile to further investigate those of network discrimination, for example as a tool to determine bounds on the communication rates over different network architectures.

\begin{acknowledgments}
I would like to thank Mario Berta, Mark M. Wilde and Andreas Winter for helpful comments and references. Furthermore, I acknowledge financial support from VILLUM FONDEN via the QMATH Centre of Excellence (Grant No.10059) and the QuantERA ERA-NET Cofund in Quantum Technologies implemented within the European Union’s Horizon 2020 Programme (QuantAlgo project) via the Innovation Fund Denmark.
\end{acknowledgments}

\appendix


\section{AEPs and chain rules} \label{Ap:chainRules}

The aim of this section will be to develop tools that allow us to better compare amortized and regularized superchannel relative entropies. We will show, given that a certain technical assumption is true, that the regularized superchannel relative entropy defined in the main text is equal to the amortized superchannel relative entropy. The path will lead us via generalizing the chain rule for quantum channels proven in~\cite{fang2019chain}. To this end, we will mostly work with the max-relative entropy and its smoothed version,
\begin{align}
D_{\max} (\rho\|\sigma) = \inf\{ \lambda:\rho\leq2^\lambda\sigma\}, \\
D^\epsilon_{\max} (\rho\|\sigma) = \inf_{\tilde\rho\in B_\epsilon(\rho)} D_{\max} (\tilde\rho\|\sigma), 
\end{align}
where $B_\epsilon(\rho) = \{ \tilde\rho\in\cS_\leq(A)\, : \, P(\tilde\rho, \rho)\leq \epsilon\}$ is the $\epsilon$-ball around $\rho$ with respect to the purified distance and $\cS_\leq(A)$ denoting the set of subnormalized states. 
An important property of the smooth max-relative entropy is that it asymptotically converges to the relative entropy as given by the asymptotic equipartition property (AEP)~\cite{tomamichel2012framework,tomamichel2013hierarchy}, i.e.
\begin{align}
\lim_{n\rightarrow\infty} \frac1n D^\epsilon_{\max} (\rho^{\otimes n}\|\sigma^{\otimes n}) = D(\rho\|\sigma). \label{sAEP}
\end{align}
We will also need the smooth max-channel relative entropy, 
\begin{align}
D^\epsilon_{\max} (\cN\|\cM) = \inf_{\tilde\cN \approx_\epsilon \cN} D^\epsilon_{\max} (\tilde\cN\|\cM), 
\end{align}
where the infimum is over all quantum channels $\tilde\cN$ with $\frac12 \| \tilde\cN - \cN \|_\diamond\leq\epsilon$. Its regularization is given by,
\begin{align}
D^{\epsilon, \infty}_{\max} (\cN\|\cM) = \lim_{n\rightarrow\infty} \frac1n D^\epsilon_{\max} (\cN^{\otimes n}\|\cM^{\otimes n}).
\end{align}
Unfortunately, it is unknown whether this quantity has a similar asymptotic behavior as its state equivalent. We do know that the following holds,
\begin{align}
D^\infty(\cN\| \cM) \leq D^{\epsilon, \infty}_{\max} (\cN\|\cM) \leq D_{\max}(\cN\|\cM). \label{Eq:weknow}
\end{align}
This follows e.g. from~\cite{wang2019resource,BHKW}. Here, we however want to assume that the quantity behaves nicely and we make the following technical assumption when needed, 
\begin{align}
\sup_{\epsilon>0} D^{\epsilon, \infty}_{\max} (\cN\|\cM) \stackrel{?}{=} D^\infty(\cN\| \cM). \label{cAEP}
\end{align} 
Some brief remarks here: First, the regularization of the relative entropy is necessary~\cite{fang2019chain} and, second, the smoothing is necessary, as we know that without smoothing the right inequality in Equation~\eqref{Eq:weknow} becomes an equality~\cite{BHKW}. It is currently unclear whether the supremum over $\epsilon$ is actually needed, but the stated version is sufficient for our application. 

The equality in Equation~\eqref{cAEP} was first conjectured by Andreas Winter~\cite{WinterAEP} and formally discussed in~\cite{liu2019resource}. 
Proving Equation~\eqref{cAEP} would also be useful beyond the questions discussed here and in~\cite{liu2019resource}, e.g. solve an open problem in the resource theory of asymmetric distinguishability of quantum channels~\cite[Section VI.6]{wang2019resource}. Work towards proving the conjecture can also been found in~\cite{gour2019quantify}. 

With these preliminaries out of the way, we can now turn to the chain rules. In~\cite{fang2019chain}, Fang \etal proved a new chain rule for the relative entropy under quantum channels, which lead to the realization that the amortized channel relative entropy is equal to the regularized channel relative entropy. We will now generalize this result to superchannels by following the same approach. For that we first prove a chain rule for the max-relative entropy. 
\begin{theorem}\label{Thm:DmaxCR}
For two superchannels $\Theta_1$ and $\Theta_2$ with decomposition $\{ \cE_1, \cD_1 \}$ and $\{ \cE_2, \cD_2 \}$ respectively, quantum channels $\cN$ and $\cM$ and $\epsilon$, $\epsilon_s$, $\epsilon_c > 0$, we have
\begin{align}
&D^{\tilde\epsilon}_{\max} ( \Theta_1( \cN)(\rho)^{\otimes m} \| \Theta_2( \cM)(\sigma)^{\otimes m}) \nonumber\\
&\leq \sup_{\bar\rho, \bar\cN} D^\epsilon_{\max} ( \Theta_1( \bar\cN)(\bar\rho)^{\otimes m} \| \Theta_2( \bar\cN)(\bar\rho)^{\otimes m}) -m\log(1-\epsilon_s) + m D^{\epsilon_c}_{\max} (\cN\|\cM) + m D^{\epsilon_s}_{\max} (\rho\|\sigma),
\end{align}
with $\tilde\epsilon=\epsilon +  m \epsilon_s +\sqrt{m \epsilon_s} + m \sqrt{2\epsilon_c}$. 
\end{theorem}
\begin{proof}
First, by definition there exists a $\rho^{\epsilon_s}\in B_{\epsilon_s}(\rho)$, such that
\begin{align}
\rho^{\epsilon_s} \leq 2^{D^{\epsilon_s}_{\max} (\rho\|\sigma)}\sigma. \label{Eq:max-state}
\end{align}
We define, $\nu^{\epsilon_s} = \frac{\rho^{\epsilon_s}}{\tr\rho^{\epsilon_s}}$. Furthermore, there also exists a $\cN^{\epsilon_c}$ with $\cN^{\epsilon_c}\approx_{\epsilon_c}\cN$ such that
\begin{align}
\forall \rho,\qquad \cN^{\epsilon_c}(\rho) \leq 2^{D^{\epsilon_c}_{\max} (\cN^{\epsilon_c}\|\cM)}\cM(\rho). \label{Eq:max-channel}
\end{align}
To continue, we also have a $\tau\in B_\epsilon(\cD_1\circ \cN^{\epsilon_c}\circ\cE_1(\nu^{\epsilon_s})^{\otimes m})$, such that
\begin{align*}
\tau &\leq 2^{ \sup_{\bar\rho, \bar\cN} D^\epsilon_{\max} ( \cD_1\circ \bar\cN\circ\cE_1(\bar\rho)^{\otimes m} \| \cD_2\circ \bar\cN\circ\cE_2(\bar\rho)^{\otimes m})} \cD_2\circ \cN^{\epsilon_c}\circ\cE_2(\nu^{\epsilon_s})^{\otimes m} \\
&\leq 2^{  \sup_{\bar\rho, \bar\cN} D^\epsilon_{\max} ( \cD_1\circ \bar\cN\circ\cE_1(\bar\rho)^{\otimes m} \| \cD_2\circ \bar\cN\circ\cE_2(\bar\rho)^{\otimes m}) -m\log(1-\epsilon_s)} \cD_2\circ \cN^{\epsilon_c}\circ\cE_2(\rho^{\epsilon_s})^{\otimes m} \\
&\leq 2^{  \sup_{\bar\rho, \bar\cN} D^\epsilon_{\max} ( \cD_1\circ \bar\cN\circ\cE_1(\bar\rho)^{\otimes m} \| \cD_2\circ \bar\cN\circ\cE_2(\bar\rho)^{\otimes m}) -m\log(1-\epsilon_s) + m D^{\epsilon_c}_{\max} (\cN\|\cM)} \cD_2\circ \cM\circ\cE_2(\rho^{\epsilon_s})^{\otimes m} \\
&\leq 2^{  \sup_{\bar\rho, \bar\cN} D^\epsilon_{\max} ( \cD_1\circ \bar\cN\circ\cE_1(\bar\rho)^{\otimes m} \| \cD_2\circ \bar\cN\circ\cE_2(\bar\rho)^{\otimes m}) -m\log(1-\epsilon_s) + m D^{\epsilon_c}_{\max} (\cN\|\cM) + m D^{\epsilon_s}_{\max} (\rho\|\sigma)} \cD_2\circ \cM\circ\cE_2(\sigma)^{\otimes m},
\end{align*}
where the second inequality follows from the definition of $\nu^{\epsilon_s}$ and the fact that $\tr\rho^{\epsilon_s}\geq 1-\epsilon_s$, see also~\cite[Proof of Proposition 3.2]{fang2019chain}. The third inequality follows by Equation~\eqref{Eq:max-channel} and the forth by Equation~\eqref{Eq:max-state}. 
We are now left with bounding, 
\begin{align}
&P(\tau , \cD_1\circ \cN\circ\cE_1(\rho)^{\otimes m}) \nonumber\\
\leq &P(\tau , \cD_1\circ \cN^{\epsilon_c}\circ\cE_1(\nu^{\epsilon_s})^{\otimes m}) + P(\cD_1\circ \cN^{\epsilon_c}\circ\cE_1(\nu^{\epsilon_s})^{\otimes m}, \cD_1\circ \cN\circ\cE_1(\rho)^{\otimes m}) \nonumber\\
\leq &\epsilon +  P(\cN^{\epsilon_c}\circ\cE_1(\nu^{\epsilon_s})^{\otimes m}, \cN\circ\cE_1(\rho)^{\otimes m}) \nonumber\\
\leq &\epsilon +  P(\cN^{\epsilon_c}\circ\cE_1(\nu^{\epsilon_s})^{\otimes m}, \cN^{\epsilon_c}\circ\cE_1(\rho)^{\otimes m}) + P(\cN^{\epsilon_c}\circ\cE_1(\rho)^{\otimes m}, \cN\circ\cE_1(\rho)^{\otimes m}) \nonumber\\
\leq &\epsilon +  P(\cN^{\epsilon_c}\circ\cE_1(\nu^{\epsilon_s})^{\otimes m}, \cN^{\epsilon_c}\circ\cE_1(\rho)^{\otimes m}) + m \epsilon_c \nonumber\\
\leq &\epsilon +  m \epsilon_s +\sqrt{m \epsilon_s} + m \sqrt{2\epsilon_c}, 
\end{align}
where the fist and third inequality follows by triangle inequality, the second and fourth by definition and the fifth along the lines of~\cite{fang2019chain}. We also used in the fourth inequality that $P(\rho,\sigma) \leq \sqrt{\| \rho-\sigma\|_1}$, see e.g.~\cite{tomamichel2012framework}. This concludes the proof. 
\end{proof}
We can now make use of the asymptotic equipartition property described in Equation~\eqref{sAEP} to get the following chain rule. 
\begin{corollary}\label{Cor:CR}
With the definitions as above, we have, 
\begin{align}
D( \Theta_1(\cN)(\rho) \| \Theta_2(\cM)(\sigma)) \leq  D^\infty (\Theta_1 \| \Theta_2) +  D^{\epsilon_c, \infty}_{\max} (\cN\|\cM) + D(\rho\|\sigma),
\end{align}
\end{corollary}
\begin{proof}
To prove the above inequality, we use the result from Theorem~\ref{Thm:DmaxCR} and reasign $\cD_i\to\cD_i^{\otimes n}$, $\cE_i\to\cE_i^{\otimes n}$, $\cN\to\cN^{\otimes n}$, $\cM\to\cM^{\otimes n}$, $\rho\to\rho^{\otimes n}$ and $\sigma\to\sigma^{\otimes n}$. Using the AEP in Equation~\eqref{sAEP} we get by taking the appropiate limits, similar to~\cite{fang2019chain}, the desired result.
\end{proof}
While this is interesting in its own right, we want to combine the result with our AEP assumption. 
\begin{corollary}
Assuming that Equation~\eqref{cAEP} holds, we have, 
\begin{align}
D( \Theta_1(\cN)(\rho) \| \Theta_2(\cM)(\sigma)) \leq  D^\infty (\Theta_1 \| \Theta_2) +  D^{\infty} (\cN\|\cM) + D(\rho\|\sigma),
\end{align} 
and therefore
\begin{align}
D^A (\Theta_1 \| \Theta_2) =  D^\infty (\Theta_1 \| \Theta_2).
\end{align} 
\end{corollary}
\begin{proof}
The first statement is a direct application of the assumption. The second follows because $D^{\infty} (\cN\|\cM)=D^{A} (\cN\|\cM)$, reordering and then taking the supremum over all channels and states. 
\end{proof}

We now want to point out that by simplifying the above proof strategy, we can also get the following results. 
\begin{corollary}
With the definitions as above, we have, 
\begin{align}
D( \Theta_1(\cN)(\rho) \| \Theta_2(\cM)(\rho)) \leq  D^{c\infty} (\Theta_1 \| \Theta_2) +  D^{\epsilon_c, \infty}_{\max} (\cN\|\cM),
\end{align}
\end{corollary}
\begin{corollary}
Assuming that Equation~\eqref{cAEP} holds, we have, 
\begin{align}
D( \Theta_1(\cN)(\rho) \| \Theta_2(\cM)(\sigma)) \leq  D^{c\infty} (\Theta_1 \| \Theta_2) +  D^{\infty} (\cN\|\cM),
\end{align} 
and therefore
\begin{align}
D^{cA} (\Theta_1 \| \Theta_2) \leq  D^{c\infty} (\Theta_1 \| \Theta_2). \label{Eq:DDc}
\end{align} 
\end{corollary}
As described in Section~\ref{Sec:NestedAd}, what we are currently missing to make Equation~\eqref{Eq:DDc} an equality is a discrimination strategy that includes parallel strategies with product states as a special case and has $D^{cA} (\Theta_1 \| \Theta_2)$ as a valid converse bound.

Finally, we remark that the chain rule in Corollary~\ref{Cor:CR} can be extended to $k$-combs following a similar derivation, resulting in additional terms for all pairs of channels the combs act on. Investigating the applications for the discrimination of general quantum networks appears to be a promising future direction.


\section{Other ways to amortize?}\label{Ap:amortize}

In this section we will briefly comment on yet another potentially different way of amortizing superchannel divergences. Based on the definition of the state-amortized superchannel divergence and a possible intuition that it covers only amortization at the input state of the superchannel, one might be tempted to define the following amortized quantity,
\begin{align}
\bD^{\tilde A}(\Theta_1 \| \Theta_2) = \sup_{\cN, \cM, \rho, \sigma} \bD(\Theta_1(\cN)(\rho) \| \Theta_2(\cM)(\sigma) ) - \bD(\cN\circ\cE_1(\rho) \| \cM\circ\cE_1(\sigma)), 
\end{align}
which conveys the intuition that the input states and the input channels are amortized simultaneously. 

Indeed, using techniques established in the previous sections, one can show a meta-converse for any strategy $\cS$, including the fully general ones,
\begin{align}
\bD(p\| q) \leq \bD^{\tilde A}(\Theta_1 \| \Theta_2),
\end{align}
and one can also see that the bound is tight for classical channels and when the second superchannel always outputs a fixed state (see Lemma~\ref{Lem:Illum} and its proof). 

However, we also have the following result that states that bounds in terms of $\bD^{\tilde A}(\Theta_1 \| \Theta_2)$ are never better than those we gave in the main text of this work. 
\begin{lemma}\label{Lem:AleqA} 
For two superchannels $\Theta_1$ and $\Theta_2$, we have
\begin{align}
\bD^{A^*}(\Theta_1 \| \Theta_2) \leq \bD^{\tilde A}(\Theta_1 \| \Theta_2).
\end{align}
\end{lemma}
\begin{proof}
We begin by fixing $\bar\cN$ and $\bar\cM$ as the channels that achieve the corresponding supremum in $\bD^{A^*}(\Theta_1 \| \Theta_2)$ as defined in Definition~\ref{Def:supAmoDiv}, 
\begin{align}
&\bD^{\tilde A}(\Theta_1 \| \Theta_2) \nonumber\\
&= \sup_{\cN, \cM, \rho, \sigma} \bD(\Theta_1(\cN)(\rho) \| \Theta_2(\cM)(\sigma) ) - \bD(\cN\circ\cE_1(\rho) \| \cM\circ\cE_1(\sigma)) \nonumber\\ 
&\geq \sup_{\cN, \cM, \rho, \sigma} \bD(\Theta_1(\cN)(\bar\cN(\rho)) \| \Theta_2(\cM)(\bar\cM(\sigma)) ) - \bD(\cN\circ\cE_1(\bar\cN(\rho)) \| \cM\circ\cE_1(\bar\cM(\sigma))) \nonumber\\ 
&= \sup_{\cN, \cM, \rho, \sigma} \bD(\Theta_1(\cN)(\bar\cN(\rho)) \| \Theta_2(\cM)(\bar\cM(\sigma)) ) - \bD(\cN\circ\cE_1(\bar\cN(\rho)) \| \cM\circ\cE_1(\bar\cM(\sigma)))  + \bD(\rho \| \sigma) - \bD(\rho \| \sigma),  \label{eq:bridge-stepApA1}
\end{align}
where the first equality is by definition, the inequality by restricting the supremum to output states of $\bar\cM$ and $\bar\cN$ and the final equality by adding a zero. 
We now fix $\hat\rho$ and $\hat\sigma$ as the optimizing states in $\bD^{A}(\Theta_1(\cN)\circ\bar\cN\|\Theta_2(\cM)\circ\bar\cM)$ and continue,
\begin{align*}
&\text{Eq.~\eqref{eq:bridge-stepApA1}} \\
&\geq \sup_{\cN, \cM} \bD(\Theta_1(\cN)(\bar\cN(\hat\rho)) \| \Theta_2(\cM)(\bar\cM(\hat\sigma)) ) - \bD(\cN\circ\cE_1(\bar\cN(\hat\rho)) \| \cM\circ\cE_1(\bar\cM(\hat\sigma)))  + \bD(\hat\rho \| \hat\sigma) - \bD(\hat\rho \| \hat\sigma) \\
&= \sup_{\cN, \cM} \bD^{A}(\Theta_1(\cN)\circ\bar\cN\|\Theta_2(\cM)\circ\bar\cM) - \bD(\cN\circ\cE_1\circ\bar\cN(\hat\rho) \| \cM\circ\cE_1\circ\bar\cM(\hat\sigma))  + \bD(\hat\rho \| \hat\sigma) \\
&\geq \sup_{\cN, \cM} \bD^{A}(\Theta_1(\cN)\circ\bar\cN\|\Theta_2(\cM)\circ\bar\cM) - \sup_{\rho,\sigma}\left[ \bD(\cN\circ\cE_1\circ\bar\cN(\rho) \| \cM\circ\cE_1\circ\bar\cM(\sigma))  + \bD(\rho \| \sigma)  \right] \\
&= \sup_{\cN, \cM} \bD^{A}(\Theta_1(\cN)\circ\bar\cN\|\Theta_2(\cM)\circ\bar\cM) - \bD^A(\cN\circ\cE_1\circ\bar\cN \| \cM\circ\cE_1\circ\bar\cM) \\
&\geq \bD^{A^*}(\Theta_1 \| \Theta_2) ,
\end{align*}
where the first equality follows by the choice of $\hat\rho$ and $\hat\sigma$, the second inequality by taking an infimum and the second and third equality by definition of the corresponding amortized quantity. This concludes the proof. 
\end{proof}
Interestingly it remains open whether the two quantities are actually different. Currently we only know of, rather particular, examples where they turn out to be the same. Should they indeed be different, that would further our understanding that superchannel divergences should be based on channel divergences. If they are the same we would get a possibly more convenient way of working with $\bD^{A^*}$. 

\section{Bonus content: parallel $\subseteq$ adaptive}\label{App:Bonus}

In the main text we have used many times the idea that parallel strategies are a special case of certain adaptive strategies. This can be verified by picking a particular adaptive strategy that implements the parallel strategy. 
For the convenience of the reader, we want to depict the transformation of a fully parallel strategy into a nested adaptive strategy in an animation in Figure~\ref{anim:nes}. The first frame gives the fully parallel strategy for $n=3$. Running the animation by using the arrows below shows how to reorganize the different elements in order to get a special case of the nested adaptive strategies as depicted in Figure~\ref{Fig:nes-adap}. Generally speaking, crossing lines in the final frame can be seen as an operation that swaps the position of two quantum system. To get the general adaptive operations, we simply allow for general quantum channels acting on all available systems instead of the swap operations. 

\textit{Running the animation requires to open this document in a modern pdf viewer, e.g. Adobe Acrobat. }

\begin{figure}[t!]
\begin{animateinline}[controls]{8}
   \multiframe{21}{nxb=0+.05}{%

    \begin{tikzpicture}
    
    \draw[white] (-8,3) rectangle (10,-10) ;
    
    \coordinate (S) at (0,0) ;
    \coordinate (P1s) at (-3.5,0.625-4);
    \coordinate (P1e) at (-7.5,0.5);
    \coordinate (P1) at ($(P1s)!\nxb!(P1e)$);
    
    \coordinate (S1u3e) at (0,1.65) ;
    \coordinate (S1u3) at ($(S)!\nxb!(S1u3e)$);
    \coordinate (S1le) at (-4,0) ;
    \coordinate (S1l) at ($(S)!\nxb!(S1le)$);
    
    \coordinate (S2une) at (0,4/3) ;
    \coordinate (S2un) at ($(S)!\nxb!(S2une)$);
    
    \coordinate (S1re) at (4,0);
    \coordinate (S1r) at ($(S)!\nxb!(S1re)$);
        
\draw[] ($(0,0)+(S2un)$) rectangle (1,2) node[pos=0.5]{$\cN$};
\draw[fill=SCcolor] (-1.5,1) -- (-0.5,1) -- (-0.5,-0.5) -- (1.5,-0.5) -- (1.5,1) -- (2.5,1) -- (2.5,-1.5) -- (-1.5,-1.5) -- cycle;
\draw (-0.5,0.5) -- (0,0.5) node[pos=0.5,sloped,above]{$A$};
\draw (1,0.5) -- (1.5,0.5) node[pos=0.5,sloped,above]{$B$};
\draw ($(-2.5,-0.25)+(S1u3)$) to [out=0,in=180] (-1.5,-0.25);    
\draw (2.5,-0.25) to [out=0,in=180] ($(3.5,-0.25)+(S1u3)$);   
\draw ($(3.5,-0.25)+(S1u3)$) -- ($(3.5,-0.25)+(S1u3)+(S1r)$); 
\node at (0.5,-1){$\Theta$};
\draw (-2.5,1.5) -- (0,1.5);
\draw (1,1.5) -- ($(3.5,1.5)+(S1r)$);
\draw ($(-2.5,-0.25)+(S1u3)$) -- ($(-2.5,-0.25)+(S1u3)+(S1l)$) -- (P1) -- ($(-2.5,1.5)!\nxb!(-6.5,1.5)$) -- (-2.5,1.5);

    \coordinate (S21e) at (-2,3.5) ;
    \coordinate (S21) at ($(S)!\nxb!(S21e)$);
    
    \coordinate (S21ae) at (-1.5,3.25) ;
    \coordinate (S21a) at ($(S)!\nxb!(S21ae)$);
    
    \coordinate (S22e) at (-2,2.5) ;
    \coordinate (S22) at ($(S)!\nxb!(S22e)$);
    \coordinate (S23e) at (2,3.5) ;
    \coordinate (S23) at ($(S)!\nxb!(S23e)$);
    
    \coordinate (S23be) at (1.5,3.25) ;
    \coordinate (S23b) at ($(S)!\nxb!(S23be)$);

    \coordinate (S24e) at (2,2.5) ;
    \coordinate (S24) at ($(S)!\nxb!(S24e)$);

    \coordinate (S2u1e) at (0,2.5) ;
    \coordinate (S2u1) at ($(S)!\nxb!(S2u1e)$);

    \coordinate (S2u2e) at (0,3.25) ;
    \coordinate (S2u2) at ($(S)!\nxb!(S2u2e)$);
    
    \coordinate (S2ne) at (0,2) ;
    \coordinate (S2n) at ($(S)!\nxb!(S2ne)$);
    
    \coordinate (S2u3e) at (0,2.05) ;
    \coordinate (S2u3) at ($(S)!\nxb!(S2u3e)$);
    
    \coordinate (S2le) at (-2,0) ;
    \coordinate (S2l) at ($(S)!\nxb!(S2le)$);

\draw[] ($(0,0-4)+(S2n)+(S2un)+(S2un)$) rectangle ($(1,2-4)+(S2n)+(S2un)$) node[pos=0.5]{$\cN$};
\draw[fill=SCcolor] ($(-1.5,1-4)+(S21)$) -- ($(-1.5,1-4)+(1,0)+(S21)$) -- ($(-0.5,-0.5-4)+(S22)$) -- ($(1.5,-0.5-4)+(S24)$) -- ($(1.5,1-4)+(S23)$) -- ($(2.5,1-4)+(S23)$) -- ($(2.5,-1.5-4)+(S24)$) -- ($(-1.5,-1.5-4)+(S22)$) -- cycle;
\draw ($(-0.5,0.5-4)+(S21)$)  to [out=0,in=180]   ($(0,0.5-4)+(S21a)+(S1u3)$);   
\draw ($(0,0.5-4)+(S21a)+(S1u3)$) --   ($(0,0.5-4)+(S2u2)+(S1u3)$);
\draw ($(1,0.5-4)+(S23b)+(S1u3)$)  to [out=0,in=180]   ($(1.5,0.5-4)+(S23)$);   
\draw ($(1,0.5-4)+(S2u2)+(S1u3)$) --   ($(1,0.5-4)+(S23b)+(S1u3)$); 
\draw ($(-2.5,-0.25-4)+(S21)+(S2u3)$)  to [out=0,in=180]    ($(-1.5,-0.25-4)+(S21)$);  
\draw ($(2.5,-0.25-4)+(S23)$)  to [out=0,in=180] ($(3.5,-0.25-4)+(S23)+(S2u3)$);  
\draw ($(3.5,-0.25-4)+(S23)+(S2u3)$) -- ($(3.5,-0.25-4)+(S23)+(S2u3)-(S2l)$);
\node at ($(0.5,-1-4)+(S2u1)$){$\Theta$};
\draw ($(-2.5,-0.25-4)+(S21)+(S2u3)$) -- ($(-2.5,-0.25-4)+(S21)+(S2u3)+(S2l)$) -- (P1);

    \coordinate (S31e) at (-4,7) ;
    \coordinate (S31) at ($(S)!\nxb!(S31e)$);
    
    \coordinate (S31ae) at (-3.5,6.8) ;
    \coordinate (S31a) at ($(S)!\nxb!(S31ae)$);
    
    \coordinate (S32e) at (-4,5) ;
    \coordinate (S32) at ($(S)!\nxb!(S32e)$);
    \coordinate (S33e) at (4,7) ;
    \coordinate (S33) at ($(S)!\nxb!(S33e)$);
    
    \coordinate (S33be) at (3.5,6.8) ;
    \coordinate (S33b) at ($(S)!\nxb!(S33be)$);
    
    \coordinate (S34e) at (4,5) ;
    \coordinate (S34) at ($(S)!\nxb!(S34e)$);
    \coordinate (S3u1e) at (0,5) ;
    \coordinate (S3u1) at ($(S)!\nxb!(S3u1e)$);

    \coordinate (S3u2e) at (0,6.8) ;
    \coordinate (S3u2) at ($(S)!\nxb!(S3u2e)$);
    
        \coordinate (S3ne) at (0,6.66) ;
    \coordinate (S3n) at ($(S)!\nxb!(S3ne)$);
    
    \coordinate (S3u3e) at (0,2) ;
    \coordinate (S3u3) at ($(S)!\nxb!(S3u3e)$);

\draw[] ($(0,0-8)+(S3n)+(S2un)$) rectangle ($(1,2-8)+(S3n)$) node[pos=0.5]{$\cN$};
\draw[fill=SCcolor] ($(-1.5,1-8)+(S31)$) -- ($(-1.5,1-8)+(1,0)+(S31)$) -- ($(-0.5,-0.5-8)+(S32)$) -- ($(1.5,-0.5-8)+(S34)$) -- ($(1.5,1-8)+(S33)$) -- ($(2.5,1-8)+(S33)$) -- ($(2.5,-1.5-8)+(S34)$) -- ($(-1.5,-1.5-8)+(S32)$) -- cycle;
\draw ($(-0.5,0.5-8)+(S31)$) to [out=0,in=180]  ($(0,0.5-8)+(S31a)+(S3u3)$); 
\draw ($(0,0.5-8)+(S31a)+(S3u3)$) --  ($(0,0.5-8)+(S3u2)+(S3u3)$);
\draw ($(1,0.5-8)+(S33b)+(S3u3)$) to [out=0,in=180]  ($(1.5,0.5-8)+(S33)$);   
\draw ($(1,0.5-8)+(S3u2)+(S3u3)$) --  ($(1,0.5-8)+(S33b)+(S3u3)$);
\draw ($(-2.5,-0.25-8)+(S31)$) -- ($(-1.5,-0.25-8)+(S31)$) node[pos=0.5,sloped,above]{$C$};
\draw ($(2.5,-0.25-8)+(S33)$) -- ($(3.5,-0.25-8)+(S33)$) node[pos=0.5,sloped,above]{$D$};
\node at ($(0.5,-1-8)+(S3u1)$){$\Theta$};
\draw ($(-2.5,-0.25-8)+(S31)$) -- (P1);

    \pgfmathparse{\nxb==0}
    \let\r\pgfmathresult
\ifnum \r=1
\draw[] (0.5,-0) -- (0.5,-0.2);
\draw[dotted] (0.5,-0.2) -- (0.5,-0.4);
\draw[dotted] (0.5,-1.6) -- (0.5,-1.8);
\draw[] (0.5,-1.8) -- (0.5,-2);
\draw[] (0.5,-0-4) -- (0.5,-0.2-4);
\draw[dotted] (0.5,-0.2-4) -- (0.5,-0.4-4);
\draw[dotted] (0.5,-1.6-4) -- (0.5,-1.8-4);
\draw[] (0.5,-1.8-4) -- (0.5,-2-4);
\fi

     \coordinate (M1e) at (4,0) ;
    \coordinate (M1) at ($(S)!\nxb!(M1e)$);
      \coordinate (M2e) at (4,5) ;
    \coordinate (M2) at ($(S)!\nxb!(M2e)$);
      \coordinate (M3e) at (3.51,2.5) ;
    \coordinate (M3) at ($(S)!\nxb!(M3e)$);
 
\draw ($(3.5,2)+(M1)$) -- ($(3.5,-8.5)+(M2)$);
\draw[] ($(3.5,2)+(M1)$) to [out=290,in=70] ($(3.5,-8.5)+(M2)$);
\draw[snake] ($(2.53+2,-3.25)+(M3)$) -- ($(3.5+2,-3.25)+(M3)$) node[pos=0.5,sloped,above]{$p$};

    \end{tikzpicture}
    }
\end{animateinline}
\caption{\label{anim:nes} Animation of how the fully parallel strategy is transformed into a special case of a nested adaptive strategy. Control with arrows at the bottom. Requires viewing in a modern pdf reader.} 
\end{figure}

\bibliographystyle{quantum}

\bibliography{reference}

\begin{thebibliography}{10}

\bibitem{wang2012one}
Ligong Wang and Renato Renner.
\newblock ``One-shot classical-quantum capacity and hypothesis testing''.
\newblock \href{https://dx.doi.org/10.1103/PhysRevLett.108.200501}{Physical
  Review Letters {\bf 108}, 200501}~(2012).

\bibitem{datta2013smooth}
Nilanjana Datta, Milan Mosonyi, Min-Hsiu Hsieh, and Fernando~GSL Brandao.
\newblock ``A smooth entropy approach to quantum hypothesis testing and the
  classical capacity of quantum channels''.
\newblock \href{https://dx.doi.org/10.1109/TIT.2013.2282160}{IEEE transactions
  on information theory {\bf 59}, 8014--8026}~(2013).

\bibitem{brandao2020adversarial}
Fernando~GSL Brandao, Aram~W Harrow, James~R Le, and Yuval Peres.
\newblock ``Adversarial hypothesis testing and a quantum stein’s lemma for
  restricted measurements''.
\newblock \href{https://dx.doi.org/10.1109/TIT.2020.2979704}{IEEE Transactions
  on Information Theory {\bf 66}, 5037--5054}~(2020).

\bibitem{cooney2016operational}
T.~{Cooney}, C.~{Hirche}, C.~{Morgan}, J.~P. {Olson}, K.~P. {Seshadreesan},
  J.~{Watrous}, and M.~M. {Wilde}.
\newblock ``{Operational meaning of quantum measures of recovery}''.
\newblock \href{https://dx.doi.org/10.1103/PhysRevA.94.022310}{Physical Review
  A {\bf 94}, 022310}~(2016).

\bibitem{hirche2017discrimination}
Christoph Hirche, Masahito Hayashi, Emilio Bagan, and John Calsamiglia.
\newblock ``Discrimination power of a quantum detector''.
\newblock \href{https://dx.doi.org/10.1103/PhysRevLett.118.160502}{Physical
  Review Letters {\bf 118}, 160502}~(2017).

\bibitem{Audenaert2008}
K.~M.~R. Audenaert, M.~Nussbaum, A.~Szko{\l}a, and F.~Verstraete.
\newblock ``Asymptotic error rates in quantum hypothesis testing''.
\newblock \href{https://dx.doi.org/10.1007/s00220-008-0417-5}{Communications in
  Mathematical Physics {\bf 279}, 251--283}~(2008).

\bibitem{BBH}
Mario Berta, Fernando G. S.~L. Brandao, and Christoph Hirche.
\newblock ``On composite quantum hypothesis testing''.
\newblock \href{https://dx.doi.org/10.1007/s00220-021-04133-8}{Commun. Math.
  Phys. {\bf 385}, 55--77}~(2021).

\bibitem{BHKW}
Mark~M Wilde, Mario Berta, Christoph Hirche, and Eneet Kaur.
\newblock ``Amortized channel divergence for asymptotic quantum channel
  discrimination''.
\newblock \href{https://dx.doi.org/10.1007/s11005-020-01297-7}{Letters in
  Mathematical Physics {\bf 110}, 2277--2336}~(2020).

\bibitem{wang2019resource}
Xin Wang and Mark~M. Wilde.
\newblock ``Resource theory of asymmetric distinguishability for quantum
  channels''.
\newblock \href{https://dx.doi.org/10.1103/PhysRevResearch.1.033169}{Physical
  Review Research {\bf 1}, 033169}~(2019).

\bibitem{fang2019chain}
Kun Fang, Omar Fawzi, Renato Renner, and David Sutter.
\newblock ``A chain rule for the quantum relative entropy''.
\newblock \href{https://dx.doi.org/10.1103/PhysRevLett.124.100501}{Phys. Rev.
  Lett. {\bf 124}, 100501}~(2020).

\bibitem{Hayashi09}
Masahito Hayashi.
\newblock ``{Discrimination of two channels by adaptive methods and its
  application to quantum system}''.
\newblock \href{https://dx.doi.org/10.1109/TIT.2009.2023726}{IEEE Transactions
  on Information Theory {\bf 55}, 3807--3820}~(2009).

\bibitem{BHKW2}
Mario Berta, Christoph Hirche, Eneet Kaur, and Mark~M Wilde.
\newblock ``Stein’s lemma for classical-quantum channels''.
\newblock In 2019 IEEE International Symposium on Information Theory (ISIT).
\newblock \href{https://dx.doi.org/10.1109/ISIT.2019.8849562}{Pages
  2564--2568}.
\newblock IEEE~(2019).

\bibitem{salek2020adaptive}
Farzin Salek, Masahito Hayashi, and Andreas Winter.
\newblock ``Usefulness of adaptive strategies in asymptotic quantum channel
  discrimination''.
\newblock \href{https://dx.doi.org/10.1103/PhysRevA.105.022419}{Phys. Rev. A
  {\bf 105}, 022419}~(2022).

\bibitem{hirche2018asymptotic}
Christoph Hirche.
\newblock ``From asymptotic hypothesis testing to entropy
  inequalities''~(2018).
\newblock  \href{http://arxiv.org/abs/1812.05142}{arXiv:1812.05142}.

\bibitem{Ume62}
Hisaharu Umegaki.
\newblock ``Conditional expectations in an operator algebra, {IV} (entropy and
  information)''.
\newblock \href{https://dx.doi.org/10.2996/kmj/1138844604}{Kodai Mathematical
  Seminar Reports {\bf 14}, 59--85}~(1962).

\bibitem{P85}
D\'enes Petz.
\newblock ``Quasi-entropies for states of a von {Neumann} algebra''.
\newblock \href{https://dx.doi.org/10.2977/PRIMS/1195178929}{Publ. RIMS, Kyoto
  University {\bf 21}, 787--800}~(1985).

\bibitem{P86}
D\'enes Petz.
\newblock ``Quasi-entropies for finite quantum systems''.
\newblock \href{https://dx.doi.org/10.1016/0034-4877(86)90067-4}{Reports in
  Mathematical Physics {\bf 23}, 57--65}~(1986).

\bibitem{ACMBMAV07}
K.~M.~R. {Audenaert}, J.~{Calsamiglia}, R.~{Mu{\~n}oz-Tapia}, E.~{Bagan},
  L.~{Masanes}, A.~{Acin}, and F.~{Verstraete}.
\newblock ``Discriminating states: The quantum {C}hernoff bound''.
\newblock \href{https://dx.doi.org/10.1103/PhysRevLett.98.160501}{Physical
  Review Letters {\bf 98}, 160501}~(2007).

\bibitem{NS09}
Michael Nussbaum and Arleta Szko{\l}a.
\newblock ``The {C}hernoff lower bound for symmetric quantum hypothesis
  testing''.
\newblock \href{https://dx.doi.org/10.1214/08-AOS593}{The Annals of Statistics
  {\bf 37}, 1040--1057}~(2009).

\bibitem{muller2013quantum}
Martin M{\"u}ller-Lennert, Fr{\'e}d{\'e}ric Dupuis, Oleg Szehr, Serge Fehr, and
  Marco Tomamichel.
\newblock ``On quantum {R}{\'e}nyi entropies: a new generalization and some
  properties''.
\newblock \href{https://dx.doi.org/10.1063/1.4838856}{Journal of Mathematical
  Physics {\bf 54}, 122203}~(2013).

\bibitem{WWY14}
Mark~M. Wilde, Andreas Winter, and Dong Yang.
\newblock ``Strong converse for the classical capacity of entanglement-breaking
  and {H}adamard channels via a sandwiched {R}\'enyi relative entropy''.
\newblock \href{https://dx.doi.org/10.1007/s00220-014-2122-x}{Communications in
  Mathematical Physics {\bf 331}, 593--622}~(2014).

\bibitem{U76}
Armin Uhlmann.
\newblock ``The ``transition probability'' in the state space of a *-algebra''.
\newblock \href{https://dx.doi.org/10.1016/0034-4877(76)90060-4}{Reports on
  Mathematical Physics {\bf 9}, 273--279}~(1976).

\bibitem{Datta09}
Nilanjana Datta.
\newblock ``Min- and max-relative entropies and a new entanglement monotone''.
\newblock \href{https://dx.doi.org/10.1109/TIT.2009.2018325}{IEEE Transactions
  on Information Theory {\bf 55}, 2816--2826}~(2009).

\bibitem{Jain02}
Rahul Jain, Jaikumar Radhakrishnan, and Pranab Sen.
\newblock ``Privacy and interaction in quantum communication complexity and a
  theorem about the relative entropy of quantum states''.
\newblock In Proceedings of the 43rd Annual IEEE Symposium on Foundations of
  Computer Science.
\newblock \href{https://dx.doi.org/10.1109/SFCS.2002.1181967}{Pages 429--438}.
\newblock ~(2002).

\bibitem{LM15}
Debbie Leung and William Matthews.
\newblock ``On the power of {PPT}-preserving and non-signalling codes''.
\newblock \href{https://dx.doi.org/10.1109/TIT.2015.2439953}{IEEE Transactions
  of Information Theory {\bf 61}, 4486--4499}~(2015).

\bibitem{WFD17}
Xin Wang, Kun Fang, and Runyao Duan.
\newblock ``Semidefinite programming converse bounds for quantum
  communication''.
\newblock \href{https://dx.doi.org/10.1109/TIT.2018.2874031}{IEEE Transactions
  on Information Theory {\bf 65}, 2583--2592}~(2019).

\bibitem{CG18}
Eric Chitambar and Gilad Gour.
\newblock ``Quantum resource theories''.
\newblock \href{https://dx.doi.org/10.1103/RevModPhys.91.025001}{Reviews of
  Modern Physics {\bf 91}, 025001}~(2019).

\bibitem{chiribella2009theoretical}
Giulio Chiribella, Giacomo~Mauro D’Ariano, and Paolo Perinotti.
\newblock ``Theoretical framework for quantum networks''.
\newblock \href{https://dx.doi.org/10.1103/PhysRevA.80.022339}{Physical Review
  A {\bf 80}, 022339}~(2009).

\bibitem{G18}
Gilad Gour.
\newblock ``Comparison of quantum channels with superchannels''.
\newblock \href{https://dx.doi.org/10.1109/TIT.2019.2907989}{IEEE Transactions
  on Information Theory {\bf 65}, 5880--5904}~(2019).

\bibitem{CDP08}
Giulio Chiribella, Giacomo~Mauro D'Ariano, and Paolo Perinotti.
\newblock ``Transforming quantum operations: Quantum supermaps''.
\newblock \href{https://dx.doi.org/10.1209/0295-5075/83/30004}{EPL (Europhysics
  Letters) {\bf 83}, 30004}~(2008).

\bibitem{bisio11}
A.~Bisio, G.~Chiribella, G.~M. D'Ariano, and P.~Perinotti.
\newblock ``{Quantum Networks: General Theory and Applications}''.
\newblock \href{https://dx.doi.org/10.2478/v10155-011-0003-9}{Acta Physica
  Slovaca {\bf 61}, 273--390}~(2011).

\bibitem{bisio11b}
Alessandro Bisio, Giacomo~Mauro D’Ariano, Paolo Perinotti, and Giulio
  Chiribella.
\newblock ``Minimal computational-space implementation of multiround quantum
  protocols''.
\newblock \href{https://dx.doi.org/10.1103/PhysRevA.83.022325}{Physical Review
  A {\bf 83}, 022325}~(2011).

\bibitem{gour2020dynamical}
Gilad Gour and Carlo~Maria Scandolo.
\newblock ``Dynamical resources''~(2020).
\newblock  \href{http://arxiv.org/abs/2101.01552}{arXiv:2101.01552}.

\bibitem{chiribella2008memory}
Giulio Chiribella, Giacomo~M D’Ariano, and Paolo Perinotti.
\newblock ``Memory effects in quantum channel discrimination''.
\newblock \href{https://dx.doi.org/10.1103/PhysRevLett.101.180501}{Physical
  review letters {\bf 101}, 180501}~(2008).

\bibitem{nakahira2020ultimate}
Kenji Nakahira and Kentaro Kato.
\newblock ``Simple upper and lower bounds on the ultimate success probability
  for discriminating arbitrary finite-dimensional quantum processes''.
\newblock \href{https://dx.doi.org/10.1103/PhysRevLett.126.200502}{Physical
  Review Letters {\bf 126}, 200502}~(2021).

\bibitem{H69}
Carl~W. Helstrom.
\newblock ``Quantum detection and estimation theory''.
\newblock \href{https://dx.doi.org/10.1007/BF01007479}{Journal of Statistical
  Physics {\bf 1}, 231--252}~(1969).

\bibitem{H73}
Alexander~S. Holevo.
\newblock ``Statistical decision theory for quantum systems''.
\newblock \href{https://dx.doi.org/10.1016/0047-259X(73)90028-6}{Journal of
  Multivariate Analysis {\bf 3}, 337--394}~(1973).

\bibitem{Hel76}
Carl~W. Helstrom.
\newblock ``Quantum detection and estimation theory''.
\newblock Academic. New York~(1976).

\bibitem{Cooney2016}
Tom Cooney, Mil{\'a}n Mosonyi, and Mark~M. Wilde.
\newblock ``Strong converse exponents for a quantum channel discrimination
  problem and quantum-feedback-assisted communication''.
\newblock \href{https://dx.doi.org/10.1007/s00220-016-2645-4}{Communications in
  Mathematical Physics {\bf 344}, 797--829}~(2016).

\bibitem{PV10}
Yury Polyanskiy and Sergio Verd\'u.
\newblock ``Arimoto channel coding converse and {R\'enyi} divergence''.
\newblock In Proceedings of the 48th Annual Allerton Conference on
  Communication, Control, and Computation.
\newblock \href{https://dx.doi.org/10.1109/ALLERTON.2010.5707067}{Pages
  1327--1333}.
\newblock ~(2010).

\bibitem{SW12}
Naresh Sharma and Naqueeb~Ahmad Warsi.
\newblock ``On the strong converses for the quantum channel capacity
  theorems''.
\newblock \href{https://dx.doi.org/10.1103/PhysRevLett.110.080501}{Phys. Rev.
  Lett. {\bf 110}, 080501}~(2012).

\bibitem{fawzi2020defining}
Hamza Fawzi and Omar Fawzi.
\newblock ``Defining quantum divergences via convex optimization''.
\newblock \href{https://dx.doi.org/10.22331/q-2021-01-26-387}{Quantum {\bf 5},
  387}~(2020).

\bibitem{fang2019geometric}
Kun Fang and Hamza Fawzi.
\newblock ``Geometric {R}\'{e}nyi divergence and its applications in quantum
  channel capacities''.
\newblock \href{https://dx.doi.org/10.1007/s00220-021-04064-4}{Communications
  in Mathematical Physics {\bf 384}, 1615--1677}~(2021).

\bibitem{HP91}
Fumio Hiai and D\'enes Petz.
\newblock ``The proper formula for relative entropy and its asymptotics in
  quantum probability''.
\newblock \href{https://dx.doi.org/10.1007/BF02100287}{Communications in
  Mathematical Physics {\bf 143}, 99--114}~(1991).

\bibitem{tomamichel2013hierarchy}
Marco Tomamichel and Masahito Hayashi.
\newblock ``A hierarchy of information quantities for finite block length
  analysis of quantum tasks''.
\newblock \href{https://dx.doi.org/10.1109/TIT.2013.2276628}{IEEE Transactions
  on Information Theory {\bf 59}, 7693--7710}~(2013).

\bibitem{li2014second}
Ke~Li et~al.
\newblock ``Second-order asymptotics for quantum hypothesis testing''.
\newblock \href{https://dx.doi.org/10.1214/13-AOS1185}{The Annals of Statistics
  {\bf 42}, 171--189}~(2014).

\bibitem{mosonyi2015quantum}
Mil{\'a}n Mosonyi and Tomohiro Ogawa.
\newblock ``Quantum hypothesis testing and the operational interpretation of
  the quantum r{\'e}nyi relative entropies''.
\newblock \href{https://dx.doi.org/10.1007/s00220-014-2248-x}{Communications in
  Mathematical Physics {\bf 334}, 1617--1648}~(2015).

\bibitem{mosonyi2011quantum}
Mil{\'a}n Mosonyi and Fumio Hiai.
\newblock ``On the quantum r{\'e}nyi relative entropies and related capacity
  formulas''.
\newblock \href{https://dx.doi.org/10.1109/TIT.2011.2110050}{IEEE Transactions
  on Information Theory {\bf 57}, 2474--2487}~(2011).

\bibitem{L08}
Seth Lloyd.
\newblock ``Enhanced sensitivity of photodetection via quantum illumination''.
\newblock \href{https://dx.doi.org/10.1126/science.11606}{Science {\bf 321},
  1463--1465}~(2008).

\bibitem{bavaresco2021strict}
Jessica Bavaresco, Mio Murao, and Marco~T{\'u}lio Quintino.
\newblock ``Strict hierarchy between parallel, sequential, and
  indefinite-causal-order strategies for channel discrimination''.
\newblock \href{https://dx.doi.org/10.1103/PhysRevLett.127.200504}{Physical
  review letters {\bf 127}, 200504}~(2021).

\bibitem{tomamichel2012framework}
Marco Tomamichel.
\newblock ``A framework for non-asymptotic quantum information theory''~(2012).
\newblock  \href{http://arxiv.org/abs/1203.2142}{arXiv:1203.2142}.

\bibitem{WinterAEP}
Andreas Winter.
\newblock ``Open problem session''.
\newblock Rocky Mountain Summit on Quantum Information~(2018).

\bibitem{liu2019resource}
Zi-Wen Liu and Andreas Winter.
\newblock ``Resource theories of quantum channels and the universal role of
  resource erasure''~(2019).
\newblock  \href{http://arxiv.org/abs/1904.04201}{arXiv:1904.04201}.

\bibitem{gour2019quantify}
Gilad Gour and Andreas Winter.
\newblock ``How to quantify a dynamical quantum resource''.
\newblock \href{https://dx.doi.org/10.1103/PhysRevLett.123.150401}{Physical
  review letters {\bf 123}, 150401}~(2019).

\end{thebibliography}

\end{document}